\pdfoutput=1
\documentclass[11pt]{article}

\usepackage[margin=1in]{geometry}

\usepackage{amssymb,amsmath,amsthm}
\usepackage{thmtools}
\usepackage[colorlinks=true,linkcolor=blue,citecolor=magenta,hyperfootnotes=false]{hyperref}

\usepackage{url}            %
\usepackage{booktabs}       %
\usepackage{amsfonts}       %
\usepackage{nicefrac}       %
\usepackage{dsfont}
\usepackage[dvipsnames]{xcolor}         %
\usepackage{tikz}

\usepackage{multirow}
\usepackage{mathtools}
\usepackage[capitalise,nameinlink]{cleveref}
\usepackage{natbib}
\usepackage{array}
\usepackage{makecell}
\usepackage{caption}
\usepackage{enumitem}

\usepackage{mathrsfs}

\newcommand{\authnote}[3]{\textcolor{#3}{[{\footnotesize {\bf #1:} { {#2}}}]}}
\renewcommand{\authnote}[3]{}

\newcommand{\orrpfn}[1]{\footnote{}}

\newcommand{\shafi}[1]{}

\newtheorem{theorem}{Theorem}
\newtheorem*{theorem*}{Theorem}
\newtheorem{proposition}[theorem]{Proposition}
\newtheorem*{proposition*}{Proposition}
\newtheorem{lemma}[theorem]{Lemma}
\newtheorem{corollary}[theorem]{Corollary}
\newtheorem{claim}[theorem]{Claim}
\newtheorem{fact}[theorem]{Fact}

\theoremstyle{definition}
\newtheorem{definition}[theorem]{Definition}
\newtheorem{remark}[theorem]{Remark}

\usepackage[ruled,linesnumbered,noend]{algorithm2e}
\usepackage{setspace}
\usepackage[noend]{algpseudocode}
\SetKwInput{KwInput}{Input}
\SetKwInput{KwOracle}{Oracle}
\SetKwInput{KwOutput}{Output}
\SetKwInput{KwHyperparameters}{Hyperparameters}
\SetKwInput{KwParameters}{Parameters}
\algrenewcommand{\algorithmiccomment}[1]{\hfill \textcolor{orange}{\texttt{\# #1}}}
\DontPrintSemicolon
\newcommand{\Prover}[1]{{\color{Red}{\texttt{Prover}: #1}}}
\newcommand{\Verifier}[1]{{\color{OliveGreen}\texttt{Verifier}: #1}}
\newcommand{\qeq}{\overset{?}{=}}
\crefname{algocf}{Algorithm}{Algorithms}

\newcommand{\U}{\mathcal{U}}
\newcommand{\N}{\mathbb{N}}
\newcommand{\Q}{\mathbb{Q}}
\newcommand{\Z}{\mathbb{Z}}
\newcommand{\R}{\mathbb{R}}

\newcommand{\defeq}{\coloneqq}

\DeclareMathOperator*{\E}{\mathbb{E}}

\DeclarePairedDelimiter{\norm}{\Vert}{\Vert}

\renewcommand{\epsilon}{\varepsilon}
\newcommand{\eps}{\varepsilon}
\newcommand{\poly}{\mathsf{poly}}
\newcommand{\polylog}{\mathsf{polylog}}
\newcommand{\F}{\mathbb{F}}
\newcommand{\bool}{\{0,1\}}

\newcommand{\epsgap}{\varepsilon_{\rm gap}}

\newcommand\rej{\texttt{reject}}
\newcommand\acc{\texttt{accept}}

\DeclareMathAlphabet{\mathdutchcal}{U}{dutchcal}{m}{n}
\newcommand{\PP}{\mathcal{P}}
\DeclareMathOperator{\D}{\mathdutchcal{I}\mskip-2mu\mathit{nc}}

\DeclareMathOperator{\Viol}{\mathrm{Viol}}

\newcommand{\EXPCONSIST}{\mathsf{Explicit\text{-}Consistency}}
\newcommand{\MODELCONSIST}{\mathsf{Model\text{-}Consistency}}

\newcommand{\EkSAT}{\mathsf{Exact}\text{-}k\mathsf{SAT}}

\newcommand{\NP}{\mathsf{NP}}
\newcommand{\NEXP}{\mathsf{NEXP}}
\newcommand{\MIP}{\mathsf{MIP}}

\newcommand{\IPCP}{\mathsf{IPCP}}
\newcommand{\PSPACE}{\mathsf{PSPACE}}
\newcommand{\EXP}{\mathsf{EXP}}

\newcommand\ie{\emph{i.e.}}
\newcommand\eg{\emph{e.g.}}
\newcommand\cf{\emph{cf.}}

\newcommand{\machinechecked}[1]{\authnote{Machine-checked}{#1}{OliveGreen}}
\renewcommand{\machinechecked}[1]{}

{\begin{enumerate}[label=\textbf{Case~\arabic*.}, leftmargin=*, itemsep=0.5\baselineskip, topsep=0.3\baselineskip]}%
		{\end{enumerate}}

\newcommand{\Qcal}{\mathcal{Q}}
\newcommand{\match}{\mathsf{agree}}
\newcommand{\exptwo}{\mathsf{vec2int}}
\newcommand{\weight}{\mathsf{weight}}
\gdef\Incof(#1,#2;#3){\D_{(#1,#2)}(#3)}
\newcommand{\epssound}{\eps_{\mathrm{sound}}}

\newcommand{\Incdag}{\D^{\dagger}}
\newcommand{\Cmass}{\mu^{\dagger}[s = b]}
\newcommand{\Qmass}{\mu^{\dagger}[(s = b) \cap (t = 1)]}
\newcommand{\Cmasshat}{\mu^{\dagger}[\hat s = \hat b]}
\newcommand{\Qmasshat}{\mu^{\dagger}[(\hat s = \hat b) \cap (\hat t = 1)]}
\newcommand{\vInc}{v_{\D}}
\newcommand{\vContext}{v_{\Cmasshat}}
\newcommand{\vQuery}{v_{\Qmasshat}}
\newcommand{\hvInc}{\hat{v}_{\D}}
\newcommand{\hvNormQ}{\hat{v}_{m}}

\newcommand{\fnarg}[1]{\textnormal{\texttt{#1}}}
\newcommand{\rounds}{\fnarg{rounds}}
\newcommand{\val}{\fnarg{val}}
\newcommand{\degree}{\fnarg{deg}}

\newcommand{\hypersf}[2]{\hyperref[#1]{\ensuremath{\mathsf{#2}}}\xspace}
\newcommand{\sfname}[1]{\ensuremath{\mathsf{#1}}\xspace} %

\usepackage{upgreek}
\newcommand{\SumCheck}{\sfname{SumCheck}}
\newcommand{\SumChecks}{\ensuremath{\mathsf{SumCheck}}s\xspace}
\newcommand{\Code}{\sfname{R\upmu}}
\newcommand{\VerifyEncoding}{\sfname{VerEnc}}
\newcommand{\VerifyMarginoid}{\sfname{VerMarginoid}}
\newcommand{\VerifyMarginal}{\mathsf{VerMarginal}}
\newcommand{\VerifyMLE}{\sfname{VerifyMLE}}
\newcommand{\MLE}{\mathrm{MLE}}
\newcommand{\SelfCorrect}{\sfname{SelfCorrect}}
\newcommand{\VerifyConsistency}{\sfname{VerConsist}}

\crefname{claim}{Claim}{Claims}
\Crefname{claim}{Claim}{Claims}
\crefname{fact}{Fact}{Facts}
\Crefname{fact}{Fact}{Facts}

\newcommand{\EQ}{\ensuremath{\mathsf{eq}}}
\newcommand{\eq}{\EQ}
\newcommand{\dist}{\ensuremath{\mathrm{dist}}}

\usetikzlibrary{positioning, fit}

\tikzset{
	srr/.style={draw=gray, rounded corners=2pt, inner sep=4pt, minimum width=0.9cm, minimum height=0.6cm, font=\small},
	oracle/.style={draw=blue!30, circle, inner sep=1pt, minimum size=0.85cm, fill=blue!10, font=\small},
	veval/.style={draw=gray!50, circle, inner sep=1pt, minimum size=0.85cm, fill=gray!15, font=\small},
	macro/.style={draw, dashed, rounded corners=2pt, inner sep=3pt, minimum width=0.9cm, minimum height=0.6cm, font=\small},
}

\newcommand{\inlineoracle}[1]{\tikz[baseline=(X.base)]{\node[draw=blue!30, circle, inner sep=0.8pt, fill=blue!10] (X) {\ensuremath{#1}};}}
\newcommand{\inlineveval}[1]{\tikz[baseline=(X.base)]{\node[draw=gray, circle, inner sep=0.8pt, fill=gray!15] (X) {\ensuremath{#1}};}}
\newcommand{\inlinesrr}[1]{\tikz[baseline=(X.base)]{\node[draw=gray!80, rounded corners=2pt, inner sep=1.5pt] (X) {\ensuremath{#1}};}}

\newcommand{\pparagraph}[1]{\paragraph{\textnormal{{\textit{#1}}}}}

\usepackage[framemethod=tikz]{mdframed}
\newmdenv[
	leftline=true, rightline=false, topline=false, bottomline=false,
	linewidth=2pt, linecolor=gray!50,
	innerleftmargin=10pt, innerrightmargin=10pt,
	skipabove=8pt, skipbelow=8pt
]{rulequote}

\title{How to Verify Probabilistic Consistency of Predictive Models}
\author{Orr Paradise$^{1}$,~~ Oliver Richardson$^{2}$,~~ Yoshua Bengio$^{2}$,~~ Shafi Goldwasser$^{3}$
\\[.5ex]
\small 
\phantom{>>>>>}
$^1$EPFL
\hspace{2.5em}
\small
\parbox[t]{0.35\textwidth}{\centering $^2$Universit\'e de Montr\'eal, LawZero,\\ Mila -- Qu\'ebec AI Institute}
\hspace{2em}
$^3$MIT, UC Berkeley \\[3ex]
}
\begin{document}
\maketitle
\makeatletter
\let\@oldthefootnote\thefootnote
\let\thefootnote\relax
\footnotetext{Correspondence to: \texttt{papers@orrp.net}, \texttt{oliver.richardson@umontreal.ca}}
\let\thefootnote\@oldthefootnote
\makeatother

\begin{abstract}
	When a probabilistic predictor answers many conditional-probability queries, are its answers approximately self-consistent---and if so, can this be verified in polynomial-time?
	This problem is of interest for AI safety, when safety is derived from honesty about probabilistic predictions of unwanted outcomes potentially caused by an AI action. To address it,  we construct an interactive PCP as follows. Let a predictive model be specified by a probability circuit $P$ and a circuit $Q$ which outputs confidence in predictions. $P$ and $Q$ together implicitly specify exponentially many probabilistic claims. We show a protocol in which a polynomial time verifier can verify $(P,Q)$'s approximate consistency. The verifier is given the pair of circuits $(P,Q)$, which it evaluates at only a few points; alongside them it is given a proof oracle, an encoding of a \textit{witnessing probability distribution} allegedly consistent with the predictions of $(P,Q)$, which it reads at a few locations while interacting with a single un-trusted prover.

	En-route to the above result, we need to ensure the \textit{existence} of a \textit{sparse} witnessing probability distribution consistent with the model predictions.  To do so, we first consider witness distributions for the consistency of  \emph{explicit}  (rather than specified by a predictor) probabilistic claims: say $m$ claims, each of the form ``$\Pr[Y=1 \mid X=x] = p$'', over $n$ Boolean variables. Building on a body of literature initiated by Nilsson (Artif. Intel. 1986),
    we place $\ell_2$ approximate probabilistic consistency of explicit claims in $\NP$ with certificates of length $O(mn + \log B)$ in the input bit-precision $B$; and  further show how a small additive completeness--soundness gap removes dependence on the model's precision $B$. This will be important for our interactive PCP constructions.

	Together these results provide a complexity-theoretic foundation for certifying the self-consistency of probabilistic predictors. We view the explicit Interactive PCP we present as the first step toward the eventual training of predictive models to prove their own consistency.
\end{abstract}

\newpage

\section{Introduction}\label{sec:intro}

Suppose a probabilistic predictor claims that
$\Pr[X=1]=0.9$, $\Pr[Y=1 \mid X=1] = 0.9$, and $\Pr[Y=1] = 0.8$.
Although each pair of answers is plausible in isolation, the three cannot all be correct, since the first two imply that $\Pr[Y=1]\geq 0.81$.
Yet intuitively the degree of inconsistency is relatively small.
A predictor that implicitly contains the answers to many overlapping probability questions can exhibit much more complex versions of the same defect:
subsets of claims may make sense locally, yet there is no single probability distribution that is consistent with them all. Is it possible to efficiently verify that the degree of inconsistency within a predictive model (e.g., the weights of a large neural network) is small?

We remark that this question arises naturally in the context of AI safety: we do not want AI systems to manipulate us by giving conflicting answers in different contexts.
The currently most advanced general-purpose AI systems show signs of misalignment, sometimes acting towards goals that go against the human designers' intentions~\citep{bengio2025international}.
The
{\it Scientist AI} \cite{ScientistAI,bengio2026safetyhonestydisinterestedai} project aims to produce an unbiased consistent Bayesian predictor
that promises use not only in scientific prediction, but also as a guardrail for an AI agent: by predicting the probability that a proposed agent action causes a specified type of harm. 
This approach can only work if the guardrail is based on a predictor that reflects a consistent probability distribution about events in the world, chosen in a non-strategic way (such as an approximation of a Bayesian posterior). \emph{Self-consistency}, and coherence with the observed facts, are properties that are more generally desirable for AIs. Indeed, consistent AI systems are also envisioned for scientific discovery: they will generate hypotheses, assign them probabilities in light of the available evidence, and update those probabilities as new information arrives---and in a way that is not designed to please a specific human user or reflect cultural biases inconsistent with facts. Such scientific reasoning is inherently probabilistic: a system does not necessarily know which hypothesis is true, but it must reason coherently under uncertainty. Whereas in the traditional mathematical claim setting we may ask a model to prove that a definite answer is correct (\cf\ \emph{hallucination}), the natural question here is different: do the model's probabilistic claims correspond to a single coherent view of a (probabilistic) world? And if perfect coherence is too much to ask of a learned model \cite{ZhuG24}, is it \emph{approximately} consistent, up to some threshold?

Rewording this question  using the terminology of Interactive Proof systems \cite{GoldwasserMR89}, we ask:    {\it Can powerful but un-trusted Prover(s) convince an efficient Verifier that the probabilistic claims made by a predictive model are jointly consistent?}

In this paper, we  answer these question(s) affirmatively.

\subsection*{Predictive Model Consistency}

We formalize a \emph{predictive model} as a pair $(P,Q)$: a \emph{prediction circuit} $P$ and a \emph{confidence circuit} $Q$.
\begin{itemize}
\item $P$ is a circuit, such as a neural network, that takes as input a description $x$ of a context consisting of values assigned to a set of Boolean variables, and a variable $y$ representing the target event of interest, and produces a (rational) number $p = P(y\mid x) \in [0,1]$ that is to
be interpreted as ``the probability of $y$ given $x$ equals $p$'', \ie, ``$\Pr[y \mid x] = p$.''
 \item $Q(x,y)$  represents the confidence of the model in its prediction. This in particular allows the model to abstain from making predictions when uncertain and set $Q=0$.
\end{itemize}

Informally, we say that the model $(P,Q)$ is  (approximately) \emph{consistent} if there exists a joint distribution $\mu$ over all variables  that (approximately) agrees with $P$ on all possible inputs where $Q > 0$, in the sense that violations of
$\mu(y\mid x) \approx P(y|x)$ are small in aggregate over queries $(x,y)$, weighted by $Q$.
Formally,  we propose  the following measure of
\emph{inconsistency} of a distribution $\mu$, with respect to a model $(P,Q)$.
Writing $\norm{Q}_1$ for the sum of all confidence scores over all queries $(x,y)$, the model inconsistency is given by
\begin{equation}\label{eq:modelinformal}
	\D_{P,Q}(\mu)
	\coloneqq
	\sqrt{\frac{1}{\norm{Q}_1}\sum_{x,y} Q(x,y) \cdot \Big(
		\mu(x \land y) - P(y \mid x) \cdot \mu(x)
		\Big)^2}
\end{equation}
The proposed measure, which has been studied in the past \cite{potyka2014linear}, is a simplified variant  of the inconsistency measure developed in the theory of probabilistic dependency graphs \cite{RichardsonH21,Richardson22}, which unifies and justifies many loss functions in machine learning via a natural information-theoretic measure of inconsistency.
In contrast to \cite{RichardsonH21,Richardson22}  (which uses relative entropy and thus requires averaging logarithms of real numbers),  we use  $\ell_2$ norm directly which allows us to work with polynomials and bring some of the machinery of interactive proofs to bear on the current problem. Furthermore,  we believe it is  better suited to the foundations of \cite{ScientistAI}. Putting it all together, the  technical question we pose is

\begin{description}
\item[The $\MODELCONSIST$ Problem:]
		{\it Given a predictive model $(P,Q)$ and a tolerance $\tau \geq 0$, is there an efficient way to verify that there exists $\mu$ such that $\D_{P,Q}(\mu) \le \tau$?}
\end{description}

We emphasize that we do not ask whether the model's predictions are true of the world, nor whether they are calibrated against empirical data; we ask only whether the model's claims could all hold simultaneously under \emph{some} distribution. This is the probabilistic analogue of logical consistency: a collection of deterministic claims is consistent if some assignment satisfies all of them, and a collection of probabilistic claims is consistent if some distribution does. Consistency is thus a necessary condition for trust rather than a sufficient one, and a model can certainly be coherent and wrong. Whereas consistency does not imply certainty,
an \emph{in}coherent model is not merely uncertain: it is internally inconsistent, and its answers cannot be understood as arising from any single probabilistic belief state.

\subsection*{The Verification Challenge}

So, how can we hope to verify model consistency as defined above? The difficulty is one of scale. A variable is indexed by a $d$-bit string, so the model may refer to $n=2^d$ Boolean variables, and as the sample space consists of the $2^n$ joint assignments to them, the natural alleged witness distribution $\mu$ which would certify consistency is a vector of real numbers whose length is doubly exponential in $d$. This is an enormous object. Even putting aside concerns about real numbers and precision, it is hard to fathom how a verifier could efficiently compute or verify all the required marginals.  For example, even Multi-Prover interactive proofs handle statements whose witnesses are exponentially long \cite{BabaiFL91}, whereas the natural witness for our problem, the distribution $\mu$,   has doubly exponential description length. Thus, it is not clear, a priori, that a probabilistic polynomial-time Verifier can verify the consistency of a model with any amount of help.

\subsection*{The Sparse Support Distribution}

A central step is to show the existence of another, ``sparsely'' supported, distribution that can serve as a proof witness in place of the naive one.
Our route goes through a closely-related sub-problem. Rather than considering exponentially many variables encoded by a circuit, consider just $n = d$ variables, and rather than letting a circuit implicitly encode exponentially many claims, suppose that they are given as an explicit Collection (\ie, multiset) $\PP$ of Probabilistic Claims, referred to as a CPC hereafter.

We call a CPC $\mathcal P$  \emph{consistent} if there exists a joint distribution $\mu$ that agrees with all constituent claims, and
focus on a notion of ``approximate consistency'' measured by a normalized $\ell_2$ norm. Given a CPC $\PP$ and a tolerance $\tau \ge 0$ the problem {\bf Explicit Consistency} is to decide whether there exists a probability measure $\mu \in \Delta(2^n)$ over the $n$-hypercube of joint variable assignments such that
\begin{equation*}
	\D_\PP(\mu) \coloneqq
	\sqrt{\frac{1}{|\PP|}\sum_{\smash{(x,y,p) \in \mathcal P}} \Big(
		\mu(x \cap y) - p \cdot \mu(x)
		\Big)^2}
	\le \tau
	.
\end{equation*}

It is not hard to see that $\EXPCONSIST$ is $\NP$-hard, as in the restricted case of tolerance $\tau=0$ and probabilities $p \in \{0,1\}$, it essentially encodes a constraint satisfaction problem (CSP),
from which this problem inherits hardness of approximation (see also \cref{sec:hardness}). But is it in $\NP$, $\PSPACE$, $\EXP$, $\NEXP$?
The obvious candidate for a consistency witness is a joint distribution $\mu \in [0,1]^{2^n}$, which would seem to require exponential time and space to verify, putting aside issues of precision.
At least in the consistent setting ($\tau=0$), however, this is not the obstacle it appears to be: classical small-model theorems \cite{FaginHM90,KollerM94,GeorgakopoulosKP88} show that a consistent collection of $m$ claims is witnessed by a distribution supported on $O(m)$ points, with probabilities that are rationals of polynomial bit-length, which places that case in $\NP$ (see \Cref{sec:related}).
Prior work formalizing \emph{inconsistent} settings ($\tau > 0$) has suggested the problem to be harder \cite{KilianN95,richardson2023inference}.
What our setting needs is that this survives the passage to approximate consistency, and at a precision fine enough that the witness can be written into a proof oracle.

\begin{figure}[h!]
	\centering
	\tikzset{
        edge/.style={->, semithick, gray!60!black},
        qedge/.style={->, line width=0.5pt, gray!60!black},
        elabel/.style={font=\scriptsize, fill=white, inner sep=1pt},
        note/.style={font=\scriptsize, color=black!60},
        actor/.style={font=\small},
        model/.pic={
                \node[draw=black!70, rounded corners=3pt, minimum width=4.2cm, minimum height=1.25cm] (-box) at (0,0.05) {};
                \foreach \xa/\xb in {-1.6/-0.8, -0.8/0, 0/0.8}
                \foreach \ya in {-0.24,0.05,0.34}
                \foreach \yb in {-0.24,0.05,0.34}
                \draw[black!45, line width=0.4pt] (\xa,\ya) -- (\xb,\yb);
                \foreach \ya in {-0.24,0.05,0.34} \draw[black!45, line width=0.4pt] (0.8,\ya) -- (1.6,0.05);
                \foreach \x in {-1.6,-0.8,0,0.8}
                \foreach \y in {-0.24,0.05,0.34}
                \fill[black!70] (\x,\y) circle (1.4pt);
                \fill[black!70] (1.6,0.05) circle (1.4pt);
            },
        user/.pic={
                \fill[black!55] (0,0.12) circle (0.062);
                \begin{scope}
                    \clip (-0.16,-0.035) rectangle (0.16,0.05);
                    \fill[black!55] (0,-0.13) circle (0.165);
                \end{scope}
            },
		}
    \begin{tikzpicture}
		\node[font=\smaller] at (3.0, 3.65) {$\EXPCONSIST$: verifying $m$ explicit claims};
		\node[font=\smaller] at (11.25, 3.65) {$\MODELCONSIST$: verifying an entire predictive model};
		\draw[gray!40, line width=0.5pt] (6.80,-1.35) -- +(0,5.1);

		\node[draw=black!70, rounded corners=3pt, inner sep=5pt, font=\scriptsize, align=center] (CPC) at (2.2, 2.25)
		{$\Pr[y_1 \mid x_1] = p_1$\\ $\vdots$ \\ $\Pr[y_m \mid x_m] = p_m$};
		\node[fill=white, font=\scriptsize, inner sep=1.5pt] at (CPC.north) {CPC $\PP$};
		\foreach \i in {0,...,5}
            \draw[black!60, fill=blue!35] ({3.5+\i*0.28},0.72) rectangle ++(0.28,0.3);
		\node[font=\scriptsize, anchor=south east] (shortproof) at (5.3, 1.04) {short proof};
        \draw[->] (CPC) -| (shortproof.north);
		\node[actor] (VA) at (2.2,-0.2) {\color{OliveGreen}\texttt{Verifier}};
		\draw[edge] (CPC.south) -- (VA.north);
		\foreach \x in {3.64, 3.92, 4.2, 4.48, 4.76} \draw[qedge] ({\x}, 0.7) -- (2.75, 0.0);
		\draw[qedge] (5.04, 0.7) -- node[elabel, sloped, below=2pt, pos=0.38] {read in full} (2.75, 0.0);
		\draw[edge] (VA.south) -- (2.2,-0.95) node[below, font=\scriptsize] {\acc/\rej};

        \begin{scope}[shift={(0.5,0)}] %
    		\pic (mB) at (8.85, 2.35) {model};
    		\node[fill=white, font=\scriptsize, inner sep=1.5pt] at (mB-box.north) {Predictive model $(P,Q)$};
    		\foreach \i in {0,...,15}
                \draw[black!60, fill=blue!8] ({10.35+\i*0.25},0.72) rectangle ++(0.25,0.3);
    		\foreach \i in {2,5,8,14}
                \draw[black!60, fill=blue!35] ({10.35+\i*0.25},0.72) rectangle ++(0.25,0.3);
    		\node[font=\scriptsize, anchor=south east] (elpo) at (14.4, 1.04) {exponentially long proof oracle};
            \draw[->] (mB-box) -| (elpo);
    		\node[actor] (VB) at (8.85,-0.2) {\color{OliveGreen}\texttt{Verifier}};
    		\node[actor] (PB) at (12.75,-0.2) {\color{Red}\texttt{Prover}};
    		\draw[edge] ([xshift=-14pt]mB-box.south) -- node[elabel, sloped, below=2pt, pos=0.45] {two queries} (VB.north);
    		\draw[edge] ([xshift=14pt]mB-box.south) -- (VB.north);
    		\draw[<->, line width=1.1pt, gray!60!black] (VB.east) -- node[elabel, below=2pt] {interaction} (PB.west);
    		\foreach \x in {10.975, 11.725, 12.475, 13.975} \draw[qedge] ({\x}, 0.7) -- (9.4, 0.0);
    		\node[elabel, text=gray!60!black, rotate=11] at (11.4, 0.37) {poly queries};
    		\draw[edge] (VB.south) -- (8.85,-0.95) node[below, font=\scriptsize] {\acc/\rej};
    	\end{scope}
	\end{tikzpicture}
	\caption{\small Left: $m$ probabilistic claims, collected, \eg, as the outputs of a predictive model, are verified against a short proof certificate that the Verifier reads in full ($\EXPCONSIST$; \Cref{thm:inNP}). Right: the entire model is verified in one go via an Interactive PCP; the Verifier evaluates the model at two points, reads polynomially many positions of an exponentially long proof oracle, and interacts with a single untrusted Prover ($\MODELCONSIST$; \Cref{thm:explicit-iop}). In both settings, the Verifier should accept if its input is probabilistically consistent, and reject otherwise.}\label{fig:two-settings}
\end{figure}
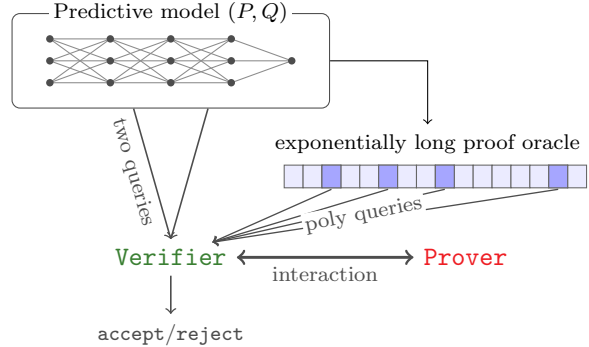

 In this paper we resolve both problems.
 We proceed to detail our full set of contributions: modeling and results.

\subsection{Summary of Contributions}

Throughout the paper, we let
$n$ be the number of Boolean variables in the system, $m := |\PP|$ be the number of claims in the CPC, and $B$ be the maximum binary precision among those claims.

\paragraph{Formalizing verifiable consistency.}
Our first contribution is the above modeling of the problem of verifying the consistency of a probabilistic predictor, such as that of the ScientistAI \cite{bengio2026safetyhonestydisinterestedai,ScientistAI}.  The mutual consistency of explicitly given probabilistic claims is a classical question, going back to \cite{Boole1854} (see \Cref{sec:related});  our new subject is a predictive model, whose claims can be exponentially many but are specified  implicitly by a circuit, and our goal is not to decide consistency but to \emph{verify} it against an untrusted Prover.

\paragraph{Short witnesses for explicit consistency.}

We place $\EXPCONSIST$ in $\NP$, extending a classical line of work on probabilistic satisfiability
(see \Cref{sec:related} for related works).

\begin{proposition*}[\Cref{thm:inNP}, informal]
	There is a verifier for explicit  probabilistic consistency that runs in time $O(m^3 (B + \log m)^2 +  m^2 n)$, with a certificate of length $O(mn + \log B)$, consisting of the support of a sparse witnessing distribution together with a single auxiliary prime, \ie,
	$\EXPCONSIST \in \NP$.

\end{proposition*}

Next, by introducing a very small gap\footnote{The gap is between \emph{values} of the inconsistency measure: given a valid certificate it accepts whenever $\D(\PP) \le \tau - \epsgap$, it rejects every certificate whenever $\D(\PP) > \tau$, and for $\D(\PP)$ in between no guarantee is made (\Cref{prop:near-linear-witnesses}).} between the completeness and soundness guarantees, we obtain a witness whose weights are written out explicitly, at low precision, and which the Verifier checks rather than computes.
We note that
at exact precision the weights alone grow quadratically in the number of claims, which is likely prohibitively expensive given that OpenAI, at the time of writing, serves over 2 billion queries per day.
It is important to pay careful attention to the length of the proof, when envisioning deployment of the Interactive PCP below,
where the witnessing distribution will need to be written out explicitly, weights included, inside the proof oracle.

\begin{proposition*}[\Cref{prop:near-linear-witnesses}, informal]
	$\EXPCONSIST$ can be verified with a gap
	$\epsgap$ with an explicit witness of length $O(nm+ m \log (1/\epsgap))$, and in time $O(m^2 (n + B \log(\nicefrac m \epsgap)))$.
\end{proposition*}

 \paragraph{An Interactive PCP for model consistency.}

Our main result is an \emph{Interactive PCP} (a proof oracle combined with an interactive proof with a single prover \cite{KalaiR08})  for $\MODELCONSIST$, thereby showing the first probabilistically checkable proof system for model consistency.

\begin{theorem*}[\Cref{thm:explicit-iop-pp} with \Cref{cor:boolean-ipcp}, informal]
	$\MODELCONSIST$ admits a polynomial-time Interactive PCP, verifying consistency up to an additive gap $\epsgap = 2^{-\poly(\ell, d, B)}$.
\end{theorem*}

\paragraph{Reed--$\mu$ller: a marginal-verifiable encoding of distributions.}
The proof oracle at the heart of our protocol is built from a primitive that may be of independent interest: a locally-verifiable encoding of sparse distributions that enables \emph{delegation of marginal computation} (\Cref{sec:codecheck}). We present it as a self-contained library, which may be used to verify other properties than consistency: the encoding, a codeword-validity verifier, and a marginal verifier. Improvements to its parameters would carry over directly to \Cref{thm:explicit-iop-pp} (see \Cref{rem:query-eps}).

\paragraph{Distributions as witnesses.}
Combining the existence of NP witnesses  for $\EXPCONSIST$ and  $\MIP = \NEXP$ \cite{BabaiFL91} places $\MODELCONSIST$ in $\MIP$ as a black-box corollary,  (\Cref{cor:mip}),
but here we go substantially further. Building on the low-precision (gap) Carath\'eodory witness of \Cref{prop:near-linear-witnesses}, we give a {direct construction}, of an \emph{Interactive PCP}  in which the proof oracle encodes the witness, a sparsified probability distribution, and the interaction reduces approximate consistency to point queries on that distribution.

The distinction is not merely cosmetic: the generic derivation compiles the problem into a computation-checking instance, losing the fact that the witness is a probability distribution and that the claims are conditional-probability constraints, whereas the explicit protocol preserves this structure. Put differently, where the IP, MIP and PCP tradition encodes proofs of deterministic mathematical facts, the object encoded here is a distribution. The protocol that drives this reduction and its analysis are developed in \Cref{sec:succinct}.

We believe that for sub-classes of probabilistic predictions which arise in an AI setting, an Interactive PCP with efficient provers can be developed in the Self-Proving model of \cite{AmitGPR25}. The two roles of ``the prover'' in an interactive PCP, a stateless proof oracle and a stateful Prover, can be played by a single model, assuming the Verifier can control its context: reset before each oracle query, and carried across interactive rounds (see \Cref{rem:spmip}). Finally, we consider pointing the way to a new era for interactive and probabilistic proofs, a contribution in its own right. Our protocol extends the classical toolkit of probabilistic proof systems (multilinear extensions, the \SumCheck protocol, self-correction), now aimed at a new kind of problem domain.

Ultimately,  we argue, probabilistic consistency is a fundamental aspect of what is needed to trust AI systems.
A highly inconsistent agent may be unable to reliably execute tasks, let alone reliably help with enhancing our knowledge.
Indeed, much of training and inference in machine learning can be viewed as the process of resolving inconsistencies in probabilistic models \cite{Richardson22,richardson2024unified}.
Being highly intelligent and useful does not ensure consistency, and it is at this point inconsistency becomes a safety concern: we do not want highly intelligent systems whispering different things in our ears, which can be manipulative and dishonest.
For these reasons, training for consistency is a key part of the \emph{Scientist AI} approach to designing safe AI systems \cite{ScientistAI}---especially when combined with probabilistic claims that enforce consistency with a Bayesian posterior.
Indeed, the Scientist AI they envision is essentially a circuit that implicitly encodes all succinctly specified probabilistic claims; for such a system, consistency is not a technical nicety but part of what it means to behave as a coherent scientific reasoner.

Finally, many open problems  emerge. Can we train AI systems to produce proofs of self-consistency of their claims for restricted classes of distributions?  What exactly does it take to convince an auditor that an allegedly consistent model is, in fact, (mostly) consistent?  We leave these questions to future work.

\section{Technical overview}
This section describes the ideas behind our two main results, the two verification settings of \Cref{fig:two-settings}; the full proofs occupy the remainder of this paper. \Cref{sec:overview:np} outlines our proof that consistency admits a \emph{sparse} witness: if a CPC of $m$ claims is consistent, then some distribution supported on only $m+1$ points is consistent with it as well, and in fact one whose exact description takes polynomially many bits, placing $\EXPCONSIST$ in $\NP$. \Cref{sec:mip:overview} uses this witness as its starting point: to verify the consistency of a predictive model, whose implicit claims are exponentially many, the Prover encodes a sparse witnessing distribution as a proof oracle, and the Verifier tests the model's inconsistency against this encoding.
\subsection{\texorpdfstring{$\EXPCONSIST \in \NP$}{Explicit-Consistency in NP} and a low-precision witness for the gapped case}\label{sec:overview:np}
We start by describing the techniques used to establish the $\NP$-verifier for $\EXPCONSIST$.
Here the instance is explicit: a CPC $\PP$ of $m$ probabilistic claims $(x, y, p)$ over $n$ variables, together with a tolerance $\tau \ge 0$, and the question is whether some distribution $\mu \in \Delta\bool^n$ satisfies
\begin{equation*}
	\D_\PP(\mu) = \sqrt{\frac{1}{m} \sum_{(x,y,p) \in \PP} \bigl(\mu(x \cap y) - p \cdot \mu(x)\bigr)^2} \;\le\; \tau.
\end{equation*}
The natural witness is the distribution itself, and therein lies the difficulty: \emph{a priori}, $\mu$ is an arbitrary vector of $2^n$ nonnegative reals summing to one.

Because the degree of inconsistency can be viewed as an $\ell_2$ norm of an $m$-dimensional constraint residual vector (the vector of the $m$ summands above), Carath\'eodory's theorem guarantees the existence of another distribution $\mu'$ supported on only $m+1$ points. We call these points $z$, and the distribution over them $\alpha$. Thus we refer to a pair $(z, \alpha) \in (\bool^n)^{m+1} \times [0,1]^{m+1}$ that has the same inconsistency as $\mu$ as a \emph{Carath\'eodory Witness}. In symbols: for every $\mu \in \Delta\bool^n$ there exist support points $z = (z_1, \dots, z_{m+1})$ with $z_j \in \bool^n$ and weights $\alpha \in \Delta[m+1]$ such that
\begin{equation*}
	\D_\PP(\mu_{z,\alpha}) = \D_\PP(\mu), \qquad \text{where } \mu_{z,\alpha}(\omega) \coloneqq \sum_{j=1}^{m+1} \alpha_j\, \mathds{1}[z_j = \omega].
\end{equation*}

However, Carath\'eodory's theorem only guarantees the existence of a sparse vector of \emph{real} weights; the remaining work, which is where our analysis departs from the classical small-support arguments (see \Cref{sec:related}), is to obtain a distribution $\mu''$ representable exactly, by rational numbers of polynomial bit-length.

To do this, we start by fixing the identities $z$ of the support points in the Carath\'eodory witness, and formulate the problem as a convex optimization problem, specifically a quadratic program: writing $V_z \in \Z^{m \times (m+1)}$ for the matrix whose $j$th column lists the constraint residuals of the single point $z_j$, rescaled to integers, the best weights for the support $z$ are those minimizing $\norm{V_z\alpha}_2^2$ subject to $\alpha \ge 0$ and $\sum_j \alpha_j = 1$.
Therefore, the square in our $\ell_2$ measure of inconsistency plays a critical role: it turns the stationary conditions of the optimization problem into a system of linear equations. Concretely, on the support of an optimal $\alpha$ the only active constraint is that the weights sum to one, and setting the gradient of the Lagrangian to zero gives the linear system
\begin{equation*}
	\begin{bmatrix}
		2 V_z^{\sf T} V_z      & \mathbf{1} \\
		\mathbf{1}^{\sf T} & 0
	\end{bmatrix}
	\begin{bmatrix}
		\alpha \\ \lambda
	\end{bmatrix}
	=
	\begin{bmatrix}
		\mathbf{0} \\ 1
	\end{bmatrix}
\end{equation*}
over the support, where $\lambda$ is the multiplier of the sum-to-one constraint (\Cref{eq:kkt}).
We then argue that the size of the coefficients in this linear system of equations is small by Hadamard's inequality, which bounds the magnitude of the determinant of a matrix; finally, we bound the denominator of the solution via Cramer's Rule.
In fact, the Verifier need not receive the weights at all: given the support, it solves this very system itself, exactly, and the determinant bounds now control its running time rather than the certificate's length. The certificate for \Cref{thm:inNP} is thus the support, $O(mn)$ bits, together with a single prime of $O(\log m + \log B)$ bits, which lets the Verifier certify that the system is non-singular before it solves, and seeds the exact solve.

Finally, the gapped case, which our Interactive PCP requires because its proof oracle must store the weights explicitly: the exact rational weights above may need $O(m(B + \log m))$ bits each, but rounding each weight to $O(\log(1/\epsgap) + \log m)$ bits perturbs the squared inconsistency by at most $\epsgap^2$ (\Cref{lem:approx-gap}), and the completeness gap leaves exactly this much room: $(\tau - \epsgap)^2 + \epsgap^2 \le \tau^2$. This is \Cref{claim:sparse-witness}, which yields the logarithmic-precision witness of \Cref{prop:near-linear-witnesses}.

\subsection{An IPCP for \texorpdfstring{$\MODELCONSIST$}{Model-Consistency}}\label{sec:mip:overview}
We first briefly recall some probabilistic proof systems. In an \emph{interactive proof} \cite{GoldwasserMR89}, a polynomial-time Verifier tosses coins and exchanges messages with an all-powerful but untrusted Prover,\footnote{Throughout this paper, the Prover is always written out in words; the symbol $P$ is reserved for the model's probability circuit.} then accepts or rejects; the requirement is that some Prover convinces it of every true claim (completeness), and that no Prover strategy convinces it of a false one, except with small probability (soundness). In a \emph{probabilistically checkable proof} (PCP) \cite{AroraS98,FortnowRS94} there is no interaction: the proof is a fixed string, possibly far too long to read in full, of which the Verifier inspects a few randomly chosen locations. An \emph{Interactive PCP} ($\IPCP$) \cite{KalaiR08} combines the two: the Prover first commits to a proof string (the \emph{proof oracle}), and the interactive proof is then carried out, during which the Verifier may also query the oracle (\Cref{def:ipcp}). The essential distinction is that the oracle is stateless: its answer at each location is fixed in advance, whereas the Prover answers each message having seen those that preceded it.

Next, the cast (\Cref{fig:two-settings}, right). The instance is the model itself: the circuits $(P,Q)$ and a tolerance $\tau$, all held by the polynomial-time Verifier, which can afford to evaluate the circuits at a point but not to enumerate their claims. A witness is a distribution $\mu$ attaining $\D_{P,Q}(\mu) \le \tau$. As the model can be viewed as making $2^{\Theta(\ell d)}$ implicit probabilistic claims, even the sparse witness of \Cref{sec:overview:np} is an exponentially long object, with each of its support points being an assignment to all $n = 2^d$ variables. Obviously, no polynomial-time Verifier can read such a witness in full, and it is thus encoded in a \emph{proof oracle} $\pi$: an exponentially long string of which the Verifier reads a few locations of its choosing, while the untrusted Prover assists, interactively, in reducing statements about all of $\pi$ to those few locations. In contrast to the PCP tradition, in which the string behind the oracle encodes a proof, here it encodes a \emph{distribution}.

The Verifier's task boils down to confirming a single inequality: that the squared inconsistency $\D^2(\mu)$ of the model $(P,Q)$ stays below the threshold $\tau^2$. It must do so given only the access above: point queries to $\pi$, which \emph{allegedly} encodes $\mu$, and point evaluations of the model. Reducing this inequality to queries to $\pi$ and evaluations of the model is the work of \Cref{sec:succinct}; \Cref{fig:srr-dag} depicts the reductions involved. For readability, this overview omits all normalization factors; they are restored in the full treatment of \Cref{sec:mip:protocol}.

It will be instructive to present the protocol as a sequence of randomized reductions: reducing a statement the Verifier cannot afford to check (a sum of $2^{\Theta(\ell d)}$ terms compared against a threshold), to statements that it can (evaluations of $\hat P$, $\hat Q$ and $\pi$). The Verifier could not carry out such a reduction on its own; instead, each step is performed in interaction with the untrusted Prover, and is sound in the sense that if the statement before the step is false then, with high probability over the Verifier's coins, so is the statement after it. The reduction used for the sums is the famous \SumCheck protocol \cite{LundFKN92}, viewed as a randomized reduction from the claim that an exponential sum of a low-degree polynomial equals a given value to the claim that the polynomial takes a given value at a single random point (\cf\ \cite{Mei13,CFS17}). \Cref{fig:srr-dag} depicts the protocol in this light: each arrow is one reduction, and the leaves are the only statements the Verifier checks directly.

Why should the proof oracle be an \emph{encoding} of that table, rather than the table itself? For two reasons. The quantities that the model's claims refer to are \emph{marginals} of $\mu$, each a sum of exponentially many of its entries; storing the table as (the evaluations of) low-degree polynomials over a large field turns each marginal into a sum of polynomial values, exactly the kind of statement that \SumCheck reduces to a few point queries. The Verifier must check, by reading a few locations and with the Prover's interactive help, that it is close to a valid encoding of \emph{some} distribution in the first place. The encoding that accomplishes both, which we term ``Reed--$\mu$ller,'' is presented in \Cref{sec:codecheck} as a self-contained ``library''; for this overview, we summarize its two capabilities:
\begin{itemize}
	\item $\VerifyEncoding$: an interactive protocol verifying that an alleged codeword is close to a valid one. Validity is defined so that every valid codeword encodes a genuine distribution, with nonnegative weights summing to one (\Cref{def:code}); passing the test thus certifies at once proximity to the code and that the encoded object is a distribution.
	\item $\VerifyMarginal$: an interactive protocol verifying any \emph{marginal} of the underlying distribution, namely, the mass it places on a partial assignment to the variables.\footnote{The protocol proper calls the field-valued variant $\VerifyMarginoid$, whose costs this overview quotes; the distinction is deferred to \Cref{sec:codecheck}.}
\end{itemize}

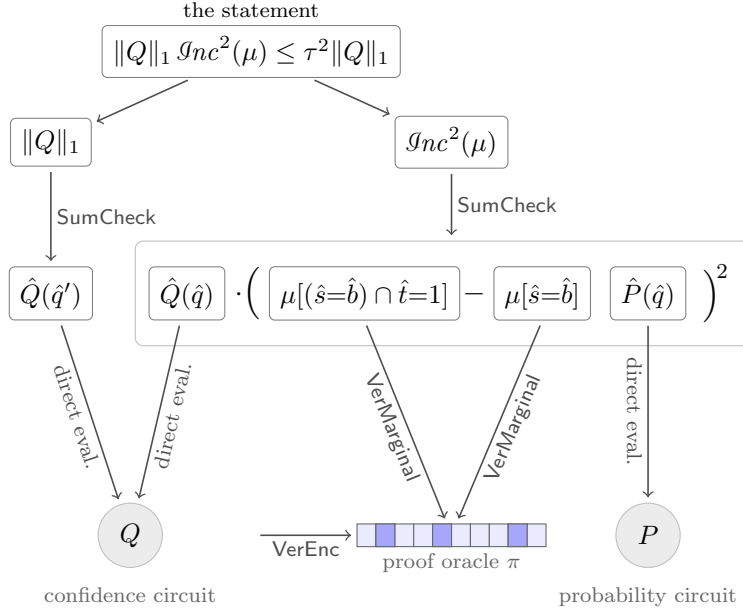
\begin{figure}[t]%
	\centering
	\begin{tikzpicture}[
			edge/.style={->, semithick, gray!60!black, shorten <=2pt, shorten >=2pt},
			elabel/.style={font=\scriptsize, fill=white, inner sep=1pt},
			resbox/.style={draw=gray!50, rounded corners=3pt, inner sep=5pt},
		]
		\node[srr, label={[font=\scriptsize, label distance=-1pt]above:the statement}] (STMT) at (3.15, 6.0) {$\norm{Q}_1 \D^2(\mu) \le \tau^2 \norm{Q}_1$};
		\node[srr] (SQ) at (0.5, 4.8) {$\norm{Q}_1$};
		\node[srr] (I) at (5.8, 4.8) {$\D^2(\mu)$};
		\draw[edge] (STMT) -- (SQ);
		\draw[edge] (STMT) -- (I);
		\node[srr] (Qhp) at (0.5, 2.8) {$\hat Q(\hat q')$};
		\node[srr] (Qh) at (2.3, 2.8) {$\hat Q(\hat q)$};
		\node (lp) at (3.1, 2.8) {$\cdot \Bigl( $};
		\node[srr] (MQ) at (4.65, 2.8) {$\mu[(\hat s{=}\hat b) \cap \hat t{=}1]$};
		\node (mi) at (6.1, 2.8) {$-$};
		\node[srr] (MC) at (7, 2.8) {$\mu[\hat s{=}\hat b]$};
		\node[srr] (Ph) at (8.4, 2.8) {$\hat P(\hat q)$};
		\node (rp) at (9.3, 2.8) {$\Bigr)^{\!2}$};
		\node[resbox, fit=(Qh)(lp)(MQ)(mi)(MC)(Ph)(rp)] (RES) {};
		\node[veval] (Q) at (1.55, -0.4) {$Q$};
		\node[minimum width=2.5cm, minimum height=0.3cm, inner sep=0] (pi) at (5.8, -0.4) {};
		\foreach \i in {0,...,9} \draw[black!60, fill=blue!8] ({4.55+\i*0.25},-0.55) rectangle ++(0.25,0.3);
		\foreach \i in {1,4,8} \draw[black!60, fill=blue!35] ({4.55+\i*0.25},-0.55) rectangle ++(0.25,0.3);
		\node[veval] (P) at (8.4, -0.4) {$P$};
		\node[font=\scriptsize, color=black!60, anchor=north] at (1.55, -0.95) {confidence circuit};
		\node[font=\scriptsize, color=black!60, anchor=north, align=center] at (5.8, -0.5) {proof oracle $\pi$};
		\node[font=\scriptsize, color=black!60, anchor=north] at (8.4, -0.95) {probability circuit};
		\draw[edge] (SQ) -- node[elabel, right=1pt] {$\SumCheck$} (Qhp);
		\draw[edge] (I) -- node[elabel, right=1pt] {$\SumCheck$} (I |- RES.north);
		\draw[edge] (Qhp) -- node[elabel, sloped, below=1pt] {direct eval.} (Q);
		\draw[edge] (Qh) -- node[elabel, sloped, below=1pt] {direct eval.} (Q);
		\draw[edge] (MQ.south) -- node[elabel, sloped, below=1pt] {$\VerifyMarginal$} (pi);
		\draw[edge] (MC.south) -- node[elabel, sloped, below=1pt] {$\VerifyMarginal$} (pi);
		\draw[edge] (Ph) -- node[elabel, sloped, below=1pt] {direct eval.} (P);
		\draw[edge] (3.2, -0.4) -- node[elabel, sloped, below=1pt] {$\VerifyEncoding$} (pi);
	\end{tikzpicture}
	\caption{\small The protocol as a chain of reductions: an arrow $E \to E'$ means that verifying $E$ reduces to verifying $E'$, interacting with the Prover via the sub-protocol labeling the arrow. The root is the statement the Verifier sets out to check; rounded boxes are intermediate expressions that get reduced further; and the bottom row holds the only objects accessed directly: the input circuits $P,Q$ (gray), which the Verifier evaluates itself, and the proof oracle $\pi$ (blue), written by the untrusted Prover and read at a few locations. The sub-protocols labeling the arrows ($\SumCheck$, $\VerifyEncoding$, $\VerifyMarginal$) are interactive, carried out with the Prover's help; ``evaluate'' the Verifier does alone. Hats mark field-valued objects.}\label{fig:srr-dag}
\end{figure}

Conceptually, the protocol has three parts: it checks that $\pi$ is close to a valid encoding of \emph{some} distribution ($\VerifyEncoding$); it obtains the marginals of that distribution to which the model's claims refer ($\VerifyMarginal$); and it verifies, by \SumCheck, that the model's aggregate inconsistency against those marginals stays below the threshold. \Cref{fig:srr-dag} depicts how the parts fit together. Reading it top-down, we proceed with a strawman that gets progressively closer to the full protocol.

\emph{Reducing the threshold $\norm{Q}_1$.} The Verifier compares $\D^2(\mu)$ against the threshold $\tau^2$, and the former requires dividing by $\norm{Q}_1 = \sum_q Q(q)$. Note that this is a sum over the set of all possible queries $\Qcal$, which is of exponential size in the model's input length.

The first \SumCheck reduces $\norm{Q}_1$ to an evaluation at a random point $\hat q' \in \F^{\log|\Qcal|}$, where $\F$ is a large finite field over which all of the protocol's arithmetic runs. Here and throughout, hats mark objects over $\F$: $\hat Q$ is the polynomial extending $Q$ to $\F$ (\Cref{fact:arith}), and the random point $\hat q'$ is a vector of field elements, so it has the ``shape'' of a query but need not describe an actual one:
\begin{equation*}
	\inlinesrr{\norm{Q}_1} = \sum_{q \in \Qcal} Q(q) \;\xmapsto{\SumCheck}\; \inlinesrr{\hat Q(\hat q')}.
\end{equation*}

\emph{Verifying the inconsistency $\D^2$.} The root of the figure is the squared inconsistency $\D^2(\mu)$, a sum over the query universe $\Qcal$ of an expression depending on the model $(P,Q)$ and the marginals of the distribution $\mu$.
Verifying such a sum is precisely what the \SumCheck protocol is designed for \cite{LundFKN92}, so one \SumCheck randomly-reduces it to a single uniformly random point $\hat q \in \F^{\log|\Qcal|}$:
\begin{gather*}
	\inlinesrr{\D^2(\mu)} = \sum_{q \in \Qcal} Q(q)\bigl(\mu[(s{=}b) \cap (t = 1)] - \mu[s{=}b] P(q)\bigr)^2 \\
	\Big\downarrow \SumCheck \\
	\inlinesrr{\hat Q(\hat q)}\Bigl(\inlinesrr{\mu[(\hat s{=}\hat b) \cap \hat t{=}1]} - \inlinesrr{\mu[\hat s{=}\hat b]}\inlinesrr{\hat P(\hat q)}\Bigr)^2.
\end{gather*}
Here $\hat s, \hat b, \hat t$ denote the blocks of the point $\hat q$ that, in a Boolean query $q = (s_1, b_1, \dots, s_\ell, b_\ell, t)$, would hold the context variable descriptions, the context bits, and the target. Four quantities are left at $\hat q$: the marginals $\mu[\hat s{=}\hat b], \mu[(\hat s{=}\hat b) \cap (\hat t{=}1)]$ and the model values $\hat P(\hat q), \hat Q(\hat q)$. We account for each in turn.

\emph{The marginals $\mu[\hat s{=}\hat b], \mu[(\hat s{=}\hat b) \cap (\hat t{=}1)]$.} The first is the mass the encoded distribution places on the context read off $\hat q$, and the second extends that context by requiring the target be true; that is, each is a marginal of $\mu$. The marginal check $\VerifyMarginal$ (the second capability of the encoding) verifies each, reducing it to $O(\ell)$ point queries on the proof oracle:
\begin{equation*}
	\inlinesrr{\mu[\hat s{=}\hat b]},\ \inlinesrr{\mu[(\hat s{=}\hat b) \cap (\hat t{=}1)]} \;\xmapsto{\VerifyMarginal}\; \text{$O(\ell)$ point queries to } \inlineoracle{\pi}.
\end{equation*}

\emph{The model values $\hat P(\hat q), \hat Q(\hat q), \hat Q(\hat q')$.} These are the input circuits $P, Q$, viewed as the polynomials they compute (\Cref{fact:arith}), at the \SumChecks' random points. The Verifier evaluates each itself, by running the circuit over $\F$ at a single point---no interaction with the Prover:
\begin{equation*}
	\inlinesrr{\hat P(\hat q)},\ \inlinesrr{\hat Q(\hat q)},\ \inlinesrr{\hat Q(\hat q')} \;\xmapsto{\Cref{fact:arith}}\; \text{evaluations of the circuits } \inlineveval{P}, \inlineveval{Q}.
\end{equation*}

\emph{Validity of the encoding.} Every reduction above reads the proof oracle $\inlineoracle{\pi}$ as a genuine $\Code$-encoding of a witnessing distribution, but a cheating Prover could instead supply an arbitrary string. The \emph{encoding check} ($\VerifyEncoding$), run once at the outset, certifies that $\pi$ is close to a genuine codeword. It is the first capability of the encoding, and the one global check in the protocol: it reduces to nothing further.

All together, the protocol runs the encoding check once; two \SumChecks for $\D^2$ and $\norm{Q}_1$; two marginal checks for the marginals at $\hat q$; and three direct circuit evaluations for $\hat P(\hat q), \hat Q(\hat q), \hat Q(\hat q')$ (\Cref{fact:arith}).\footnote{We think of these as two evaluations of the model, one at $\hat q$ and the other at $\hat q'$.} Each sub-protocol runs with its own fresh randomness, and the oracle $\pi$ is read only through $\VerifyEncoding$ and $\VerifyMarginal$. The honest proof oracle is constructed in \Cref{sec:mip:witness}, and the full protocol is given as \Cref{alg:full-protocol} and analyzed in \Cref{sec:mip:protocol}; \Cref{sec:delegation} then improves the efficiency of the direct evaluations by delegating them to the Prover. This gives an explicit route by which a model could, in principle, prove its own probabilistic consistency; \Cref{rem:spmip} elaborates.

\begin{remark}[Learning the Prover?]\label{rem:spmip}
	An $\IPCP$ can be viewed as a \emph{stateless} prover (the proof oracle, whose answers depend only on the query location) and a \emph{stateful} prover (the interactive Prover). These roles can be realized by a single model in a scenario where a user (Verifier) attempts to ascertain the probabilistic consistency of a (large language) model with which it is interacting. Namely, when training the model to act as a Prover \cite{AmitGPR25}, a single Self-Proving model $P_\theta$ can play both roles through its state or ``context'': the Verifier implements the proof oracle by resetting $P_\theta$'s context before each query (\cf\ \emph{backtracking} \cite{backtracking2,backtracking3}), so that an answer cannot depend on the queries that preceded it, and implements the interactive Prover by letting the context persist across rounds. And of course, a proof oracle is equivalent to non-communicating provers \cite{FortnowRS94}, which is how an $\IPCP$ \cite{KalaiR08} formally is a type of multi-prover interactive proof system \cite{BabaiFLS91}.
\end{remark}

\section{Related work}\label{sec:related}
\paragraph{Verifying distribution properties.} Verifying (or testing) consistency of predictive models can be viewed as an inverse problem to that of verifying properties of distributions. This is because in the latter settings, the Verifier is given access to samples from a distribution and is tasked with certifying whether a global property of this distribution holds. In the case that the property can be cast as a (possibly huge) set of probabilistic claims (constraints on the distribution's marginals), this setting can be viewed as inverse to ours: in that setting access is given to the distribution and the task is to determine whether a set of claims hold, whereas in ours access is given to the claims (or a predictive model) and the task is to determine whether there exists a distribution consistent with them.

In more detail, Interactive Proofs (IPs, \cite{GoldwasserMR89}) have been adapted to verify properties of distributions \cite{IPD0}, and this setting has since been studied extensively, showing that a rich family of properties can be verified, including, \eg, label-invariant properties \cite{IPD1,IPD2,IPD3,IPD4,IPD5,IPD6}. In the context of machine learning \cite{IPPAC0} introduced the IP-PAC systems, in which the Verifier certifies the accuracy of a given learned model with respect to a distribution (see also \cite{IPPAC1,IPPAC2}).

Without the aid of a prover, \cite{DistributionTesting1,TestingDistributions2} studied property testing \cite{RubinfeldS96,GoldreichGR98} for distribution properties, wherein a Tester is given access to samples and determines a property of the distribution from which they were drawn (see survey \cite{Canonne22}). Of this rich line of work, particularly relevant is when the tester can draw samples from marginals of the distribution in what is known as \emph{subcube-conditioned sampling} \cite{Subcube1,Subcube2,Subcube3,Subcube4,Subcube5}: in the Boolean setting, a \emph{context} $x \in \{0,1,*\}^n$ specifies a subcube, and the claim ``$\Pr[y =1\mid x] = p$'' is the (alleged) expectation of the $y$th coordinate conditioned on this subcube.

\paragraph{Probabilistic Consistency.}
The pursuit of probabilistic consistency has been salient since the development of probability itself \cite{Keynes1921,Ramsey1931}.
\cite{Boole1854} posed it as the search for the ``conditions of possible experience'',
and \emph{Dutch book arguments}, beginning with the coherence theorem of de Finetti \cite{deFinetti37}, contend that a bettor who fails to be probabilistically consistent can be easily exploited.
The standard axioms of probability, due to Kolmogorov \cite{kolmogorov50}, are intended to ensure that a probability measure---which without constraints would be simply a function that assigns numbers in $[0,1]$ to events in a $\sigma$-algebra---behaves as one would expect.

The complexity of probabilistic satisfiability has been studied in several settings \cite{Nilsson86,GeorgakopoulosKP88,FaginHM90,Lukasiewicz01,KollerM94}.
The question of whether $m$ explicit unconditional probability claims can be jointly satisfied (\textsf{PSAT}) was posed by \cite{Nilsson86}, who recast it as a linear program; 
\cite{GeorgakopoulosKP88} demonstrate $\NP$-hardness already at two literals per clause,
    and argue containment in $\NP$ is a consequence of the fact that every linear program with $m$ constraints (corresponding to $m$ unconditional claims) has a solution supported on at most $m+1$ points.
\cite{FaginHM90} study a variant with inequalities, and bound the binary representation of the witnesses, formally placing the exact consistency problem for rational weights in $\NP$ under the standard model of computation (their Sec.~2.4).
\cite{GeorgakopoulosKP88,FaginHM90} both remark in passing that the construction can be extended to handle explicit conditional claims of the form $\Pr[y\mid x] = p$ that we use (essentially our $\EXPCONSIST$ problem restricted to the case of $\tau=0$).
Our proof of \Cref{thm:inNP} is related to this classical approach, except that we operate in an inconsistency-tolerant setting ($\tau \ge 0$) whose objective is not merely a constraint residual of a linear program, and at a precision that requires us to be more careful about certificate length and introduce a gapped analogue.

Tolerance of inconsistency ($\tau > 0$) has received less attention, at least until recently. 
Kilian and Naor \cite{KilianN95} call a set of statistics $\eps$-inconsistent if every distribution violates one of them by at least $\eps$, and prove that distinguishing consistent from $\eps$-inconsistent statistics is $\NP$-hard for a constant $\eps$; however, they do not upper bound its complexity. 
\cite{potyka2014linear} proposes a family of $\ell_p$-based inconsistency measures, including our $\D$ for $p=2$, but only places its complexity in $\EXP$. This kind of inconsistency minimization is useful for database consolidation \cite{potyka2014consolidation}.
These $\ell_p$ measures are surveyed in \cite{debona2018inconsistency}, where they are connected to the degree of vulnerability to Dutch books; the entropy-based divergences of \cite{RichardsonH21} are a separate line.
\cite{richardson2023inference} show that approximating a notion of inconsistency closely related to ours is \textsf{\#P}-hard.

In machine learning, probabilistic consistency appears not as a property to be decided but as one to be trained for \cite{xu2018hierarchical,ZhuPIE17,jeong2019consistency,kim2020structured,richardson2026local}.
Systems are trained to ensure consistency between different representations, in both generative \cite{ZhuPIE17} and semi-supervised settings \cite{kim2020structured}.
Ensuring consistency between forward and backward probabilistic trajectories is a key idea behind Generative Flow Networks \cite{bengio2023gflownet}. This has led to the \emph{Scientist AI} project, which aims to generalize the techniques used to train GFlowNets for consistency to systems of exponentially more variables \cite{ScientistAI,fornasiere2026scientist}---which is precisely the primary problem we tackle in this paper.
It has been argued that all of these results fit together as part of a bigger theory of machine learning as probabilistic modeling that is tolerant to inconsistency
\cite{richardson2024unified,Richardson22}.

\paragraph{The Integer Carath\'eodory theorem.} Previous work studied ``Integer Carath\'eodory theorems'' \cite{CookFonluptSchrijver1986,Sebo90,EisenbrandS06}, in which any integer point in the conic hull of a Hilbert basis (or a finite set of integer vectors) is representable as an integer conic combination supported on linearly many basis vectors. Our bounds (\cref{sec:NP}) can be viewed as complementary in that they give a sparsity bound for \emph{rational convex} combinations of Boolean vectors, with explicit control on the bit-length of the coefficients. To the best of our understanding, our results fundamentally cannot be derived from the existing literature because, essentially, clearing denominators to reduce our setting to theirs would blow up coefficient magnitudes exponentially in the input precision, which would defeat our goal of a polynomial-length certificate. Also related is literature on ``approximate Carath\'eodory theorems'' \cite{Barman15,MirrokniLVW17}, however these results bound the $\ell_p$-approximation error of a sparse convex combination at a fixed support size, rather than the bit-complexity of an (almost-)exact stochastic representation.

\paragraph{Interactive Proofs and AI.} IPs have been demonstrated to hold theoretical and empirical promise as a way of establishing trust in learned models (``AI Safety''). \cite{irving2018ai} proposes using Debate Systems (with two competing provers, \cf\ \cite{CondonFLS93,FeigeK97}), later extended to the double-efficient setting in \cite{DEDebate1} and studied further in \cite{DEDebate2,DEDebate3}. \cite{AnilZWG21} introduced \emph{Prover--Verifier Games}, in which a non-interactive ($\NP$-)Verifier is learned jointly with a Proof-generating model. This setting was used in large-scale experiments by OpenAI in \cite{KirchnerCELMB24}, and generalized to that of Interactive Proofs (including MIPs and Zero Knowledge proof systems) in \cite{HammondD24}. Relatedly, \cite{MAClassifier2} studied a variant (\emph{MA-Classifiers}) towards obtaining more interpretable model outputs.

\cite{AmitGPR25} introduced \emph{Self-Proving models}, which are trained to act as Provers in a prescribed IP, that is, generate proofs that convince a given Verifier. This approach, which formalizes earlier empirical work (\eg, \cite{PoluS20}), can be viewed as complementary to our main result: We show how a proof oracle and a prover can convince a verifier of the consistency of a model; future work can explore whether both roles can be learned as Self-Proving models themselves. The two roles differ only in statefulness---the proof oracle is a \emph{stateless} prover, the interactive prover a \emph{stateful} one---so a single ``generalist'' model can play both: its ``memory'' is reset before each oracle query and persists across the interactive rounds (\cf\ \Cref{rem:spmip} and recent empirical works \cite{ParkOCMLB23,Du00TM24}).

\paragraph{Sum-of-squares and pseudodistributions.} The Sum-of-Squares (SoS) hierarchy of \cite{Nesterov00,Lasserre01,Parrilo03}, and the pseudo-distribution viewpoint \cite{BarakBHKSZ12,BarakS14} share a conceptual resemblance to our setting, in that their goal is to certify that a system of constraints admits a (perhaps approximate) probability-distribution solution. There are significant technical differences, however: A degree-$d$ pseudo-distribution is a linear functional on degree-$d$ polynomials with no guarantee of extending to a genuine distribution, and the relevant parameter is the polynomial degree $d$ with a semidefinite program of size exponential in $d$; more importantly, the SoS hierarchy prescribes algorithms for \emph{deciding} approximate consistency, whereas we offer algorithms for \emph{verifying} a (different notion of) approximate consistency. Namely, our setting instead measures inconsistency by $\ell_2$ distance from a genuine distribution, parameterizes by a tolerance $\tau$ (and an additive gap $\eps$), and our verification methods utilize an auxiliary proof or interactive proof system.

\section{Conventions and definitions}
For a natural number $m \ge 1$, we use $[m]:= \{1, \ldots, m\}$ for the set of positive integers at most $m$. All logarithms are base $2$. We write $\poly(\cdot)$ for an unspecified polynomial in its arguments, and $\tilde O(\cdot)$ suppresses factors polylogarithmic in the argument.

For a finite set $X$, we write $\Delta X$ for the set of all distributions $\nu$ over the elements of $X$, and use the same symbol ($\nu$, in this case) for both:
\begin{enumerate}[label=(\alph*)]
	\item the simplex point (or probability mass function) $\nu : X \to [0,1]$ satisfying $\sum_{x \in X} \nu(x) = 1$, which uniquely corresponds to
	\item the probability measure $\nu : 2^{X} \to [0,1]$ that assigns a numerical probability between zero and one to events $E \subseteq X$,
	      satisfies the Kolmogorov Axioms of probability
	      (\ie, $\nu(X) = 1$, and $\nu(A \cup B) = \nu(A) + \nu(B)$ for all disjoint $A,B \subseteq X$),
	      agrees with the simplex form (a) on singletons, and is extended to larger subsets by $\nu(E) = \sum_{x \in E} \nu(x)$.
\end{enumerate}
In particular, the set of joint distributions on the $n$-dimensional hypercube
is denoted $\Delta\bool^n = \Delta2^n$. Here and throughout, $2^n$ for a \emph{number} $n$ denotes the hypercube $\bool^n$, whereas $2^X$ for a \emph{set} $X$ denotes its power set, as in (b) above.

\paragraph{Variables and partial assignments.}
Throughout this paper, $n \in \N$ will be used to refer to a number of Boolean variables.
A (joint) \emph{assignment} to all variables is therefore a point $\omega \in \bool^n$ on the $n$-dimensional hypercube.
We are especially interested in the restricted set of events on the hypercube, namely those with that can be described through partial assignments to variables.

A \emph{partial assignment} is $x \in \{0,1,\ast\}^n$, where $x_i = \ast$ should be interpreted as having the $i$th variable be unassigned in $x$.
Assignment $\omega \in \bool^n$ \emph{agrees} with $x$, written $w \models x$, if, for all $i \in [n]$, $x_i \in \{w_i,\ast\}$; that is, $w$ and $x$ agree on all non-$\ast$ entries of $x$.
Thus, we allow ourselves to implicitly regard the partial assignment $x$ as the event $\{ \omega \in 2^n : \forall i \in [n].~x_i \in \{ \omega_i, *\}\}$ when discussing its probability.
We will often be interested in a \emph{target event} of the form $y=1$ for some variable index $y\in[n]$, which we formally regard as the event
$\{ \omega \in 2^n  : \omega_y = 1\}$.

\paragraph{Machines and oracles.}
Our model of computation is the Turing machine, and the running times of the combinatorial algorithms of \Cref{sec:NP} are counted in machine steps, that is, bit operations. An \emph{oracle machine} $M^f$ is a machine $M$ with query access to a function $f$: it may write down an input $x$ and read back $f(x)$, at unit cost per query. This is the access model behind every oracle in this paper, the proof oracle of the Interactive PCP (\Cref{def:ipcp}) in particular. The superscript names the oracle: $\VerifyEncoding_\F^{Z,A}$, for example, has query access to the pair of functions $(Z,A)$.

Circuits are Boolean unless stated otherwise, and the \emph{size} $|C|$ of a circuit $C$ is its number of gates, inputs and outputs included. When a circuit is an explicit \emph{input} to an algorithm, as the predictive model $(P,Q)$ is to the Verifier of \Cref{sec:succinct}, the input length is $\Theta(|P|+|Q|)$; a polynomial-time Verifier may therefore run in time polynomial in $|P|+|Q|$, which in particular dominates the circuits' own input length $\ell(d+1)+d$ and output precision $B$.

\paragraph{Precision of rational numbers.}
Except where noted otherwise, all numbers in the paper are assumed to be rational. A non-trivial part of our techniques is devoted to controlling the errors introduced when doing arithmetic over rational numbers, and so we must explicitly deal with how these numbers are represented.

Throughout this paper we will say that a rational number $q \in [0,1]$ is \emph{specified to precision $B \in \N$} to mean that $q$ is represented by a numerator $q^\dagger$ as a $B$-bit string, so that $q = q^\dagger / 2^B$. For example, if we say that a Turing machine takes as input a rational number $q \in [0,1]$ to precision $B$, then it takes the $B$-bit string $q^\dagger$. The convention extends to collections: when several rationals $q_1,\dots,q_k$ are specified to a common precision $B$ (as are, \eg, the claims of a CPC and the weights of our witnessing distributions), the machine takes just the numerators $(q^\dagger_1,\dots,q^\dagger_k)$, the shared denominator $2^B$ being given once. In any case, when analyzing running time or proof lengths, the description will be made explicit.

\paragraph{Finite fields and polynomials.}
For a field $\F$ of odd characteristic we let $2^{-1}$ denote the inverse of $1+1$ in $\F$. For an $n$-variate polynomial $\varphi \in \F[X_1,\dots,X_n]$, the \emph{total degree} of $\varphi$ is the maximum, over its monomials, of the sum of the variables' exponents, and its \emph{individual degree} is the largest exponent with which any single variable appears in it.

Working over a prime field, we identify the elements of $\F$ with the integers $\{0,\dots,|\F|-1\}$. An (in)equality asserted \emph{over $\Z$} compares these integer representatives as integers, rather than as field elements; unlike its over-$\F$ counterpart, such an assertion is sensitive to wraparound, and controlling it is what the field-size hypotheses of \Cref{sec:succinct} are for.

Running times in \Cref{sec:codecheck,sec:succinct} are counted in \emph{field operations} (additions, multiplications, inversions, and samples of a uniformly random element of $\F$), each of which costs $\poly(\log |\F|)$ bit operations. Finally, we say that a random vector is \emph{marginally uniform} on a set of coordinates if each of those coordinates, viewed on its own, is uniformly distributed; nothing is implied about their joint distribution.

\paragraph{Pseudocode.} We write $x \gets e$ for the assignment of the value $e$ to $x$; inside a procedure call, $f(p \gets e)$ runs $f$ with the named parameter $p$ set to $e$. We write $a \qeq b$ for an equality the Verifier asserts, rejecting if it fails. ``Expect $v$ from the Prover'' receives a message from the (untrusted) Prover and names it $v$; nothing is assumed of $v$ beyond its type. A sub-protocol invoked by the Verifier runs on the Verifier's own fresh randomness, and ``reject if it rejects'' propagates its verdict.

\subsection{Collections of Probabilistic Constraints (CPCs) and their inconsistency}
We now formally define our object of study: a collection of probabilistic claims (CPC). When specified explicitly, this is the key data of an  $\EXPCONSIST$ instance; when specified succinctly, it is the key data of a $\MODELCONSIST$ instance.

\begin{definition}[CPC specified to precision $B$]
	Let $B, n \in \N$. A \emph{probabilistic claim} specified to precision $B$ is a tuple $(x,y,p)$ where $x \in \{0,1,*\}^n$ is the \emph{context}, $y \in [n]$ is the \emph{target}, and $p \in [0,1] \cap \frac{1}{2^B} \Z$ can be represented as $p^\dagger / 2^B$ for some integer $p^\dagger \in \{0,\dots,2^B-1\}$.
	A \emph{Collection of Probabilistic Constraints (CPC)} specified to \emph{precision $B \in \N$} is a multi-set
	\unskip\footnote{We will still use $P \in \PP$ to denote membership in the multi-set, and use $|\PP|$ to denote its size (sum of multiplicities).}
	$\PP$ of probabilistic claims all specified to precision $B$. We will commonly use $m = |\PP|$ to refer to the number of claims, \ie, elements in the multi-set $\PP$. The length of $\PP$ in bits is
	$m (n\log_2(3) + \lceil \log_2 n \rceil + B) \in O(mn + mB)$.
\end{definition}

\begin{definition}[Approximate consistency and $\EXPCONSIST$]
	Let $\PP$ be a CPC of $m$ claims over $n$ variables, and let $\mu \in \Delta\bool^n$. We define the (normalized, $\ell_2$-)\emph{inconsistency} of $\mu$ with $\PP$ to be
	\begin{equation}
		\D_\PP(\mu) \coloneqq
		\sqrt{\frac{1}{|\PP|}\sum_{(x,y,p) \in \mathcal P} \big(
			\mu(x \cap y) - p \cdot \mu(x)
			\big)^2},
	\end{equation}
	and the \emph{inconsistency} of the CPC instance $\PP$ to be $\D(\PP) \coloneqq \inf_{\mu \in \Delta\bool^n} \D_\PP(\mu)$.
	For a given tolerance $\tau \ge 0$, we say that $\PP$ is \emph{$\tau$-consistent} if $\D(\PP) \leq \tau$.
	We let $\EXPCONSIST$ denote the set of all pairs $(\PP,\tau)$ such that $\PP$ is $\tau$-consistent.
\end{definition}

\subsection{Predictive models and their consistency}\label{sec:def:model}
In this paper, a \emph{(predictive) model} is a computer program that takes as input a description of events $x$ and $y$ and outputs a number to be interpreted as the conditional probability of $y$ given $x$. A common example is neural predictive models, which are neural network architectures instantiated with a set of weights, and claims described by a sequence of natural language sentences.

It is useful to allow the model to output a \emph{confidence score} in addition to its prediction. In particular, a zero confidence score allows the model to abstain from making a prediction all together, a capability whose utility is well-studied \cite{GoldwasserKKM20a,GoldwasserKKM20b,KalaiK21a,KalaiK21b,GoelHMS23}.

\begin{definition}[Predictive model]
	Fix a \emph{variable description length} $d \in \N$, a \emph{context length} $\ell \in \N$, and a \emph{precision} $B \in \N$, and write $n \coloneqq 2^d$ for the number of variables these parameters address, indexed by their $d$-bit descriptions: $\omega_z$ for $z \in \bool^d$. A \emph{predictive model} is a pair of Boolean circuits $(P,Q)$, where $P$ is the \emph{probability circuit} and $Q$ is the \emph{confidence circuit}, each taking a query as input and outputting $B$ bits. A \emph{query} is a string
	\begin{equation*}
		q = (s_1, b_1, \dots, s_\ell, b_\ell, t) \in \Qcal, \qquad \Qcal \coloneqq \bool^{\ell(d+1)+d},
	\end{equation*}
	where $\Qcal$ is the \emph{query universe}. Its components are the \emph{variable descriptions} $s_1,\dots,s_\ell \in \bool^d$ of the context, the corresponding \emph{context bits} $b_1,\dots,b_\ell \in \bool$, and the \emph{target variable description} $t \in \bool^d$. On a query $q$, the probability circuit outputs $P(q) \in \{0,\dots,2^{B}-1\}$, interpreted as an (alleged) probability $P(q) / 2^{B}$, and the confidence circuit outputs an integer $Q(q) \in \{0,\dots,2^{B}-1\}$.
\end{definition}

Fix a distribution $\mu$ over $\bool^{n}$. For a context $(s_1,b_1,\dots,s_\ell,b_\ell)$ we abbreviate by $s = b$ the event that every context literal holds, that is, $\omega_{s_k} = b_k$ for all $k \in [\ell]$, and use square brackets $\mu[s = b]$ for its probability under $\mu$---note the square brackets:
\begin{equation*}
	\mu[s = b] \coloneqq \mu\left(\left\{
	\omega \in \bool^{n} : \omega_{s_k} = b_k \text{ for all } k \in [\ell]
	\right\}\right).
\end{equation*}
Writing $t = 1$ for the target literal $\omega_t = 1$, we similarly write $\mu[(s = b) \cap (t = 1)]$ for the probability of the context event together with $\omega_t = 1$.

\begin{definition}[$\tau$-consistent predictive models]
	Let $(P,Q)$ be a model with variable description length $d$ and context length $\ell$, and let $\mu$ be a distribution over $\bool^{n}$. Let $\norm{Q}_1 = \sum_{q \in \Qcal} Q(q)$ denote the total confidence the model places over all queries.\footnote{We treat $Q$'s outputs as unsigned (nonnegative) integers, so $\norm{Q}_1 = \sum_{q}|Q(q)|$ is the $\ell_1$-norm of $Q$ over its inputs.}
	\begin{align*}
		\D_{P,Q}(\mu) & \coloneqq \sqrt{
			\frac{1}{\norm{Q}_1}
			\sum_{q = (s_1,b_1,\dots,s_\ell,b_\ell,t) \in \Qcal}
			Q(q)
			\left(
			\mu[(s = b) \cap (t = 1)] - \frac{P(q)}{2^{B}} \mu[s = b]
			\right)^2
		}
		\\
		\D_{(P,Q)}    & \coloneqq \inf_{\mu \in \Delta\bool^{n}} \D_{P,Q}(\mu).
	\end{align*}
	For a tolerance \emph{$\tau$}, we say that $(P,Q)$ is \emph{$\tau$-consistent} if $\D_{(P,Q)} \leq \tau$, and define the set $\MODELCONSIST$ comprising all $(P,Q,\tau)$ such that $(P,Q)$ is $\tau$-consistent.
\end{definition}

\subsection{Aside: Models as implicit CPCs}
The two inconsistency operators defined above deliberately share the symbol $\D$: a predictive model is a succinct description of a collection of claims, and its inconsistency is the inconsistency of that collection. We formalize this next.

\begin{definition}[Implicit CPC]\label{def:implicit-cpc}
	Let $(P,Q)$ be a predictive model. Call a query $q = (s_1, b_1, \dots, s_\ell, b_\ell, t)$ \emph{conflicting} if there are indices $k \ne k' \in [\ell]$ with $s_k = s_{k'}$ yet $b_k \ne b_{k'}$. The \emph{implicit CPC} $\PP_{P,Q}$ is the multiset of claims over $n = 2^d$ variables obtained as follows. For each $q \in \Qcal$: if $q$ is non-conflicting, let $x_q \in \{0,1,\ast\}^{n}$ be the partial assignment with $(x_q)_{s_k} = b_k$ for all $k \in [\ell]$ and $\ast$ in all other entries, and add $Q(q)$ copies of the claim $\bigl(x_q,\omega_t,\ P(q)/2^B\bigr)$; if $q$ is conflicting, add $Q(q)$ copies of the \emph{tautological} claim $\bigl(x_q, \omega_t, 0\bigr)$, where $x_q$ sets $(x_q)_t = 0$ and $\ast$ in all other entries.
\end{definition}

\begin{claim}\label{claim:implicit-cpc}
	For every distribution $\mu$ over $\bool^n$, $\D_{\PP_{P,Q}}(\mu) = \Incof(P,Q;\mu)$: the inconsistency of the implicit CPC, as a CPC, equals the inconsistency of the model, as a model. In particular, $(P,Q)$ is $\tau$-consistent if and only if $\PP_{P,Q}$ is.
\end{claim}
\begin{proof}
	By construction, $|\PP_{P,Q}| = \sum_q Q(q) = \norm{Q}_1$, and the claims group by the query that produced them: each $q$ contributes $Q(q)$ copies of one claim. For a conflicting $q$ the context event is empty, so $\mu[s=b] = \mu[(s=b) \cap (t=1)] = 0$ and $q$ contributes $0$ to the model sum; its tautological claim contributes $\mu(x_q \cap \omega_t) - 0 \cdot \mu(x_q) = 0$ to the CPC sum, since $(x_q)_t = 0$. Both are weighted $Q(q)$, so $|\PP_{P,Q}| = \sum_q Q(q) = \norm{Q}_1$ is unchanged and the two sums agree term by term. For any query $q = (s_1,b_1,\dots,s_\ell,b_\ell,t)$, the claim $(x_q, \omega_t, P(q)/2^B)$ has $\mu(x_q) = \mu[s = b]$ and $\mu(x_q \cap \omega_t) = \mu[(s = b) \cap (t = 1)]$. Therefore
	\begin{align*}
		\D_{\PP_{P,Q}}(\mu) & = \sqrt{\frac{1}{|\PP_{P,Q}|}\sum_{(x,t,p) \in \PP_{P,Q}} \bigl(\mu(x \cap y) - p \cdot \mu(x)\bigr)^2}                                                      \\
		                    & =\sqrt{\frac{1}{|\PP_{P,Q}|}\sum_{(x_q,\omega_t,P(q)/2^B) \in \PP_{P,Q}}Q(q) \bigl(\mu(x_q \cap \omega_t) - \frac{P(q)}{2^B} \cdot \mu(x_q)\bigr)^2}         \\
		                    & = \sqrt{\frac{1}{\norm{Q}_1}\sum_{q \in \Qcal} Q(q) \Bigl(\mu[(s = b) \cap (t = 1)] - \frac{P(q)}{2^B} \cdot \mu[s = b]\Bigr)^2} = \Incof(P,Q;\mu). \qedhere
	\end{align*}
\end{proof}

In light of \Cref{claim:implicit-cpc}, we use $\D$ for both operators without further comment.

\section{Verifying consistency of probabilistic claims}\label{sec:NP}

This section is concerned with establishing the following two results.

\begin{proposition}
	\label{prop:near-linear-witnesses}
	The Verifier $V$ of \Cref{alg:verifier}, which takes as input a collection of probabilistic claims $\PP$ specified to precision $B$, a tolerance $\tau\ge 0$ and gap $\epsgap > 0$, has the properties of:
	\begin{itemize}

		\item \emph{Soundness}:
		      if $\D_\PP(\mu) > \tau$ for all $\mu \in \Delta\bool^n$, then $V(\PP,\tau,\pi) = \rej$ for all proofs $\pi$.
        \item \emph{Completeness} (up to $\epsgap)$: If $\D(\PP) \le \tau - \epsgap$, then there exists a proof certificate $\pi$ such that $V(\PP, \tau, \pi) = \acc$. Moreover, $\pi$ satisfies:
        \item \emph{Efficiency}:
		      $V$ runs in time $O(m^2(n + B \log(\nicefrac m \epsgap)))$.
		\item \emph{Near-linear length and logarithmic certificate precision}:
        accepting proofs $\pi$ consist of at most $m+1$ pairs $(z,\alpha)$ with $z \in \bool^n$ and $\alpha \in [0,1]$ specified in binary to precision $B_{\epsgap} \in O(\log(\nicefrac m \epsgap))$---thus
        satisfying $\mathrm{length}(\pi) \in O(m(n + \log(\nicefrac m \epsgap)))$.
		
	\end{itemize}
\end{proposition}

Our main construct (\Cref{thm:explicit-iop}) relies on \Cref{prop:near-linear-witnesses}. Crucially, the weights of the witnessing distribution are rounded so as to be specifiable in $B_{\epsgap} \in O(\log(\nicefrac m \epsgap))$ bits.

As noted in \Cref{sec:related}, the complexity of verifying approximate consistency is of independent interest within its related literature, and so we then adapt \Cref{prop:near-linear-witnesses} to obtain a gap-free variant:

\begin{proposition}[$\EXPCONSIST \in \NP$]
	\label{thm:inNP}
	There exists a Verifier that, when run on input $\PP,\tau$, satisfies:
	\begin{itemize}
        
	    \item \emph{Completeness}:
            If $\D(\PP) \le \tau$, there exists a proof certificate such that the Verifier accepts.
		\item \emph{Soundness}:
		      if $\D(\PP) > \tau$, the Verifier rejects regardless of the given proof certificate.
        \item \emph{Efficiency}:
		      The Verifier runs in $O(m^3(B+\log m)^2 + m^2 n)$ time.
		\item \emph{Certificate length}:
		      proof certificates are of length $O(mn + \log B)$.

	\end{itemize}
\end{proposition}

To compare the two: In \cref{prop:near-linear-witnesses} the Prover sends the weights, rounded, and the Verifier only checks them; the rounding introduces the gap, but it buys weights whose precision scales logarithmically with $m / \epsgap$. In \cref{thm:inNP} only the support vectors (and a prime number) are sent: the Verifier solves for the exact optimal weights, which removes the gap at the price of a precision polynomial in $m$ and $B$, and a running time to match.
The gap-free variant cannot be directly used for our main construct (\Cref{thm:explicit-iop}), as elaborated in \Cref{rem:exact-witness}.

\subsection{The Verifier}
Consider a finite collection of probabilistic claims $\PP
	= \{ (x_i, y_i, p_i) \}_{i=1}^m
$ specified to precision $B \in \N$, \ie, each $p_i \in \mathbb Q$ is represented by the integer $p^\dagger_i \coloneqq p \cdot 2^B \in \{0,\dots, 2^B-1\}$.
Define the \emph{constraint residuals}, a vector $\phi : \bool^n \to \Z^m$ over the space of joint assignments to variables, according to:
\begin{align*}
	\phi(\omega) \coloneqq \Big[
		(2^B \omega_{y_i} - p_i^\dagger)
		\mathds{1}[\omega \models x_i]
		\Big]_{i=1}^m,
\end{align*}
whose components each have magnitude at most $2^B$.
The key feature of this representation is that
\begin{equation}
	\D^2_\PP(\mu)
	= \frac{1}{ m \cdot 2^{2B}}
	\norm[\Big]{ \E_{\omega \sim \mu}[\phi(\omega)] }_2^2
	=
	\frac1m
	\sum_{i=1}^m (\mu(x_i,y_i) - \mu(x_i) p_i)^2
	\quad  \text{  for all }\mu \in \Delta \bool^n
	.
\end{equation}

Carath\'eodory's Theorem \cite{Caratheodory07} states that any point in the convex hull of a set $S \subset \R^d$ can be expressed as a convex combination of at most $d+1$ points from $S$.
We will use this to demonstrate that, if there exists a distribution $\mu^*$ with small constraint violation $\D_\PP(\mu^*) \le \tau$, then there also exists one supported on only $m+1$ points. That sparsely supported distribution (or rather, a discretization of it) will ultimately be the proof certificate.

To be more precise, start with the observation that
$\{ \E_\mu[\phi] : \mu \in \Delta \bool^n \} = \mathrm{conv}\{\phi(\omega) : \omega \in \bool^n\}$
is convex by definition of the expectation.
Therefore, by Carath\'eodory's Theorem, for every $\mu$, the point $\E_{\mu}[\phi] \in \mathbb R^m$ can be expressed as a convex combination $\E_{\mu}[\phi] = \sum_{j=1}^{m+1} \alpha_j\, \phi(z_j)$ of some joint assignments $z_1, \ldots, z_{m+1} \in \bool^n$.
This yields a (typically different) distribution $\mu' = \mu_{z,\alpha} \in \Delta \bool^n$
with the same constraint residual as $\mu$, but supported on only $m+1$ points. Stated more formally:
\begin{align*}
	\forall \mu \in \Delta(\bool^n).~
	\exists \alpha = (\alpha_j )_{j=1}^{m+1} \in \Delta[m+1],~z \in (\bool^n)^{m+1}.\quad
	\E\nolimits_{\mu}[\phi] = \sum_{j=1}^{m+1} \alpha_j \phi(z_j) = \E\nolimits_{\mu_{z,\alpha}}[\phi],
	\\[-1ex]
	\text{ where } \mu_{z,\alpha}(\omega) \coloneqq \sum_{j = 1}^{m+1}  \alpha_j\, \mathds{1}[ z_j = \omega]
	~~\text{ has }~~ \big|\,\mathrm{supp}(\mu_{z,\alpha})\big| \le m+1.
\end{align*}

We remark that such sparsification arguments have been used in the literature: the same fact appears as a statement about basic feasible solutions of a linear system in \cite{GeorgakopoulosKP88}\cite[Theorem 2.3]{KollerM94}, and is what enables the small-model theorem in \cite[Theorem 2.6]{FaginHM90}, both in the service of placing exact ($\tau = 0$) consistency of unconditional claims in $\NP$. The other key ingredient is a bound on the lengths of the rational coefficients of this witness, which \cite{FaginHM90} repurposed from the non-negative integer linear programming literature \cite{chvatal1983linear}. The bound that we prove in \Cref{lem:rationalopt} serves the same role for us, although we must do it directly, since our $\ell_2$ inconsistency formulation cannot be cast as a linear program, and we also need to very precisely bound the bit complexity of the weights of an approximate witness (c.f., \Cref{rem:exact-witness}).

This leads us to \cref{alg:verifier}. Its certificate specifies the weights to precision $B_{\epsgap}$: per our convention, only the numerators $\alpha^\dagger_j = 2^{B_{\epsgap}}\alpha_j$ are sent.

\begin{algorithm}
	\caption{Verification of Approximate Probabilistic Consistency}
	\label{alg:verifier}
	\KwInput{Collection $\PP$ of $m$ probabilistic claims to precision $B$, tolerance $\tau$, gap $\epsgap > 0$.}
	\KwInput{A proof certificate comprising: an integer $k \in [m+1]$, assignments $z_1,\dots, z_k \in \bool^n$, and weights $\alpha^\dagger_1,\dots,\alpha^\dagger_k \in \{0,\dots,2^{B_{\epsgap}}\}$ for $B_{\epsgap} := \lceil\log_2\!\big( 2(m+1)^3 / (\epsgap^2 m) \big) \rceil$.}
	\BlankLine

    Check that the certificate is of the form declared above, and that $\sum_j\alpha^\dagger_j = 2^{B_{\epsgap}}$, else $\rej$.\;
	Calculate
	$\displaystyle
		\mathit{inc}^2 \coloneqq \sum_{(x,y,p) \in \PP} \Big( \sum_{j=1}^{k} \alpha_j^\dagger \, (2^B z_{j,y} - p^\dagger)\mathds{1}[z_j \models x]  \Big)^2.
	$ \;
	\KwOutput{\acc\ if $\mathit{inc}^2 \le m\cdot2^{2B + 2B_{\epsgap}} \cdot  \tau^2$; else \rej}
\end{algorithm}

Write $V_z \coloneqq \big[\phi(z_1),\, \phi(z_2),\, \cdots,\, \phi(z_k)\big] \in \Z^{m \times k}$ for the \emph{residual matrix} of the support, whose $j$th column lists the constraint residuals of the single point $z_j$. It is what ties the algorithm's arithmetic back to the inconsistency it is meant to measure: for any weights $\alpha \in \Delta[k]$ over the support,
\begin{equation}
	\norm{V_z\alpha}_2^2 = \norm[\Big]{\sum_{j=1}^{k}\alpha_j \,\phi(z_j) }_2^2 = m \cdot 2^{2B} \cdot \D^2_{\PP}(\mu_{z,\alpha}),
	\label{eq:V-inc}
\end{equation}
so that $\mathit{inc}^2 = 2^{2B_{\epsgap}}\norm{V_z\alpha}_2^2 = m (2^{B + B_{\epsgap}})^2 \D^2_\PP(\mu_{z,\alpha})$. Therefore, the final check of \cref{alg:verifier} holds exactly when $\D_\PP(\mu_{z,\alpha}) \le \tau$, and the best weights for a given support are those minimizing $\norm{V_z \alpha}_2^2$ over the simplex.

\paragraph{Runtime Analysis.}
Let's take a closer look at the calculation on the penultimate line of \cref{alg:verifier}. In the inner sum, we:
\begin{enumerate}
	\item calculate whether or not $z_j \models x$, which involves comparing each of the $n$ bits of $x$ against the $n$ bits of $z_j$, and hence takes $n$ steps.
	\item Subtract $p^\dagger$, a non-negative integer at most $2^B$, from either $2^B$ or zero, depending on the bit $z_y^j \in \{0,1\}$, which takes $B$ steps, and produces a $(B+1)$-bit number.
	\item Multiply this number by $\alpha_j^\dagger$, an integer at most $2^{B_{\epsgap}}$. This can be done in  $1+B B_{\epsgap}$ steps (since the additional bit is a sign that can be copied in one step) and produces a $(B+B_{\epsgap})$-bit number.
\end{enumerate}
We repeat these three steps $k \le m+1$ times, which takes at most $(m+1) (1+n + B + B B_{\epsgap}) \in O(m(n+B B_{\epsgap}))$ steps, and results in an integer that can be represented in $1+B+B_{\epsgap}+\log_2(m+1)$ bits.
Then we square it, resulting in a number that can be represented in $2(1+B+B_{\epsgap}+\log_2(m+1))$ bits, and sum over all $m$ probabilistic claims in $\PP$, for a total of
\[
	m \Big(\,
	\underbrace{(m+1)(1 + n+B+B B_{\epsgap})}_{\text{inner loop}}\, + \underbrace{4(1+B_{\epsgap}+B+2\log(m))^2\vphantom{\big|}}_{\text{accumulation}}\,\Big)
	\in O(m^2 (n+B B_{\epsgap}) )
	~~\text{steps}.
\]
Since $B_{\epsgap} \in O(\log(\nicefrac m \epsgap))$, this is $O(m^2(n + B\log(\nicefrac m \epsgap)))$ steps. The modified Verifier of \cref{sec:NP:exact} runs the same accounting at a larger denominator, and pays for its solve on top.

\paragraph{Soundness.}
The proof $\pi = (z,\alpha^\dagger)$ sent by the Prover exactly encodes a distribution $\mu_{z,\alpha}$, and by \eqref{eq:V-inc} the quantity computed on the last line satisfies $\mathit{inc}^2 = m (2^{B + B_{\epsgap}})^2 \D^2_\PP(\mu_{z,\alpha})$, so the final check holds if and only if $\D_\PP(\mu_{z,\alpha}) \le \tau$. Therefore, if the Verifier accepts, $\mu_{z,\alpha}$ itself is a distribution with $\D_\PP(\mu_{z,\alpha}) \le \tau$. The procedure is thus perfectly sound: the gap will be paid only in completeness.

\subsection{\texorpdfstring{Logarithmic-precision weights, at the cost of a gap (\Cref{prop:near-linear-witnesses})}{Logarithmic-precision weights, at the cost of a gap}}

The honest Prover sends an approximate witness $(z,\hat\alpha)$: a support $z$, and weights $\hat\alpha$ rounded to $B_{\epsgap}$ bits, satisfying $\norm{\alpha-\hat\alpha}_\infty \le \delta$ for the exact optimum $\alpha$ on that support. Two lemmas make the rounding precise: the first bounds the error it incurs, and the second shows that weights of any prescribed precision exist within any rounding radius. \Cref{claim:sparse-witness} then puts them together, and \cref{prop:near-linear-witnesses} follows.

\begin{lemma}
	\label{lem:approx-gap}
	If $\alpha, \hat\alpha \in \Delta[m+1]$ are such that $\norm{\alpha - \hat\alpha}_\infty \le \delta$, \ie, $\forall j \in [m+1].~|\alpha_j - \hat\alpha_j| \le \delta$,
	then for all $z = (z_j)_{j=1}^{m+1}$,
	$\big|\D^2_\PP(\mu_{z,\hat \alpha}) - \D^2_\PP(\mu_{z,\alpha})\big| \le 2 \delta (m+1)^3 /m$.
\end{lemma}

\begin{proof}
	First, we show that, since $\alpha, \hat \alpha \in \Delta[m+1]$ are probabilities over $[m+1]$ that are close in the sense that $\max_i |\alpha_i - \hat \alpha_i| \le \delta$, then $|\alpha_i\alpha_j - \hat\alpha_i \hat\alpha_j| \le 2\delta - \delta^2$ for all $i,j \in [m+1]$.
	This is evident in the following picture:
	
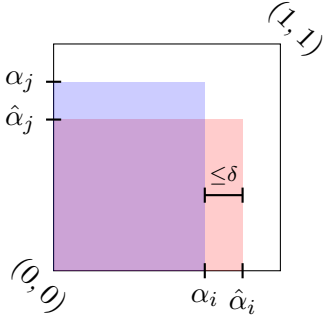
\begin{center}
		\begin{tikzpicture}
			\draw (0,0) rectangle (3,3);
			\fill[red,opacity=0.2] (0,0) rectangle (2.5,2);
			\fill[blue,opacity=0.2] (0,0) rectangle (2,2.5);
			\draw[thick] (2,0.1) -- (2,-0.1) node[below]{$\alpha_i$};
			\draw[thick] (2.5,0.1) -- (2.5,-0.1) node[below]{$\hat\alpha_i$};
			\draw[thick] (0.1,2.5) -- (-0.1,2.5) node[left]{$\alpha_j$};
			\draw[thick] (0.1, 2) -- (-0.1,2) node[left]{$\hat\alpha_j$};

			\draw[thick] (2,1) -- node[above]{$\scriptstyle\le\delta$} (2.5,1);
			\draw[thick] (2,0.9) -- (2,1.1) (2.5,0.9) -- (2.5, 1.1);
			\node[rotate=-45] at (-.2,-.2) (Z) {$(0,0)$};
			\node[rotate=-45] at (3.2,3.2) (Z) {$(1,1)$};
		\end{tikzpicture}
	\end{center}
	In this unit square, the difference between the areas of the red and blue rectangles can be at most $2\delta$.

	We now compute
	\begin{align*}
		m \cdot 2^{2B}\cdot \big|\D^2_\PP(\mu_{z,\hat \alpha}) - \D^2_\PP(\mu_{z,\alpha})\big| & =
		\bigg| \, \norm[\Big]{\sum_{j=1}^{m+1} \hat \alpha_j \phi( z_j )}_2^2 -
		\norm[\Big]{\sum_{j=1}^{m+1} \alpha_j \phi( z_j )}^2_2 \,\bigg|                                                                                            \\
		=
		\Big| \sum_{i=1}^m \sum_{k=1}^{m+1}\sum_{j=1}^{m+1}                        & \hat\alpha_j \hat\alpha_k \phi_i(z_j) \phi_i(z_k)-
		\sum_{(x,y,p) \in \PP} \sum_{k=1}^{m+1}\sum_{j=1}^{m+1} \alpha_j \alpha_k \phi_i(z_j) \phi_i(z_k)  \Big|                     \\
		                                                                           & =
		\Big| \sum_{i=1}^m \sum_{k=1}^{m+1}\sum_{j=1}^{m+1} ( \hat\alpha_j \hat\alpha_k -  \alpha_j \alpha_k) \phi_i(z_j) \phi_i(z_k) \Big| \\
		                                                                           & \le
		\sum_{i=1}^m\sum_{k=1}^{m+1}\sum_{j=1}^{m+1} | \hat\alpha_j \hat\alpha_k -  \alpha_j \alpha_k|\, |\phi_i(z_j) \phi_i(z_k)|          \\
		                                                                           & \le
		(2 \delta - \delta^2) \sum_{i=1}^m \sum_{k=1}^{m+1}\sum_{j=1}^{m+1} |\phi(z_j) \phi(z_k) |
		\\
		                                                                           & \le 2 \delta  (m+1)^3 \cdot (2^{B})^2
	\end{align*}
	Canceling the $2^{2B}$ scaling factor and dividing by $m$ gives the desired result.
\end{proof}

We will also need the following straightforward fact about the density of fixed-precision binary numbers on the simplex.

\begin{lemma}
	\label{lem:density}
	For all $\alpha \in \Delta[m+1]$ and $B \in \mathbb N$, there exists an $\alpha' \in \Delta[m+1]$ satisfying $\norm{\alpha - \alpha'}_\infty \le 2^{-B}$ whose components are rational numbers of the form $\alpha'_i = \alpha^\dagger_i / 2^B$ for some $\alpha^\dagger_i \in \Z_{\ge 0}$.
\end{lemma}

\begin{proof}[Proof of \cref{lem:density}]
	Given $\alpha\in \Delta[m+1]$, start by truncating the first $B$ bits of its binary representation of each component, \ie,
	define $\alpha'' \in [0,1]^{m+1}$ to be the vector whose $i^\text{th}$ component is $\lfloor 2^B \alpha_i \rfloor / 2^B$.
	Clearly $\norm{\alpha-\alpha''}_\infty \le 2^{-B}$, but $\alpha''$ will not be a distribution since $\sum_i \alpha''_i < 1$ unless $2^B\alpha$ is already a vector of integers.

	Let $\beta := 1-\sum_{i}\alpha''_i$ be total mass missing from $\alpha''$ required for it to be a distribution. Since the number $1$ and each $\alpha''_i$ is an integer multiple of $2^{-B}$, so too is $\beta$. Let $\beta^\dagger := \beta \cdot 2^B \in \Z$ be that multiple.
	Since each $(\alpha_i-\alpha''_i) < 2^{-B}$ by the definition of rounding, we know that $\beta$, which is the sum across all $m+1$ components of the vector must be at most $2^{-B}(m+1)$, \ie, $\beta^\dagger < m+1$.

	Finally, choose any vector $v\in \{0,1\}^{m+1}$ with Hamming weight $\sum_i v_i = \beta^\dagger$, and let $\alpha' := \alpha'' + v \cdot 2^{-B}$ be the result of re-assigning all $\beta^\dagger$ units of missing mass to distinct components. Clearly $\alpha'$ is non-negative, sums to one, and its components are rational with common denominator $2^B$.
\end{proof}

Now to put the pieces together. The point is that the gap in the completeness hypothesis is exactly the room the rounding needs.

\begin{claim}[Sparse witnesses at bounded precision]
	\label{claim:sparse-witness}
	Let $\tau \ge \epsgap > 0$ and let $B_{\epsgap} = \big\lceil \log_2\!\big( 2(m+1)^3/(\epsgap^2 m) \big) \big\rceil \in O(\log(\nicefrac1\epsgap) + \log m)$ as in \cref{alg:verifier}. If $\D_\PP(\mu) \le \tau - \epsgap$ for some $\mu \in \Delta(\bool^n)$, then there are assignments $z = (z_j)_{j=1}^{m+1}$ and integer weights $\alpha^\dagger \in \Z_{\ge 0}^{m+1}$ with $\sum_j \alpha^\dagger_j = 2^{B_{\epsgap}}$ such that $\hat\mu \coloneqq \mu_{z, \alpha^\dagger 2^{-B_{\epsgap}}}$ satisfies $\D_\PP(\hat\mu) \le \tau$.
\end{claim}

\begin{proof}
	By Carath\'eodory's theorem, there are $z \in (\bool^n)^{m+1}$ and $\alpha \in \Delta[m+1]$ with $\E_{\mu_{z,\alpha}}[\phi] = \E_\mu[\phi]$, hence $\D_\PP(\mu_{z,\alpha}) = \D_\PP(\mu) \le \tau - \epsgap$. Set $\delta \coloneqq \epsgap^2 m / (2(m+1)^3)$, so that $2^{-B_{\epsgap}} \le \delta$. By \cref{lem:density}, applied at precision $B_{\epsgap}$, there is $\hat\alpha = \alpha^\dagger 2^{-B_{\epsgap}} \in \Delta[m+1]$ with $\norm{\alpha - \hat\alpha}_\infty \le 2^{-B_{\epsgap}} \le \delta$, and by \cref{lem:approx-gap} the rounding perturbs the squared inconsistency by at most $2\delta(m+1)^3/m = \epsgap^2$. Therefore
	\[
		\D^2_\PP(\hat\mu) \;\le\; (\tau - \epsgap)^2 + \epsgap^2 \;=\; \tau^2 - 2\epsgap(\tau - \epsgap) \;\le\; \tau^2,
	\]
	using $\tau \ge \epsgap$ in the last step. %
\end{proof}

\begin{proof}[Proof of \cref{prop:near-linear-witnesses}]
	Soundness is the soundness paragraph above: the final check of \cref{alg:verifier} is exact, so any accepted certificate exhibits a distribution $\mu_{z,\alpha}$ with $\D_\PP(\mu_{z,\alpha}) \le \tau$.

	For completeness, suppose $\D_\PP(\mu) \le \tau - \epsgap$ for some $\mu$; in particular $\tau \ge \epsgap$, since otherwise no such $\mu$ exists. The certificate $(z, \alpha^\dagger)$ of \cref{claim:sparse-witness} has weights summing to $2^{B_{\epsgap}}$, so it passes the first check, and the exact final check passes because $\D_\PP(\hat\mu) \le \tau$. Its length is at most $(m+1)(n + B_{\epsgap}) \in O(m(n + \log(\nicefrac m \epsgap)))$.

	Efficiency is the runtime analysis above.
\end{proof}

In particular, if we care about numbers only up to the same precision as the problem instance, \eg, to 32-bit fixed point precision ($\epsgap = 2^{-B}$), then witness weights of length $(m+1)(2B + 1 + 3 \log_2(m+1)) \in O(mB + m \log m)$ suffice.

\subsection{\texorpdfstring{Polynomial-precision weights, and no gap (\Cref{thm:inNP})}{Polynomial-precision weights, and no gap}}\label{sec:NP:exact}

To decide $\EXPCONSIST$ rather than its gapped version, the Verifier needs the exact optimal weights for the support, and those are not short. If the claims are represented to 16 bits, and we have a meager $m=100$ claims, the bound of \cref{lem:rationalopt} below already asks for some 4600 bits per weight; we expect the AI systems and probabilistic models that motivate our interest in the $\EXPCONSIST$ problem to make many more claims, and for $m$ equal to one million (2 megabytes of data) the figure is above $7 \times 10^7$ bits per weight, against the 105 that suffice at a gap of $2^{-32}$. The remedy is not to send them at all. 

The Verifier  can solve for the weights itself, incurring the polynomial cost of precision in the running time rather than in length of the certificate.

\paragraph{The modified Verifier.} The certificate components $\alpha_1^\dagger,\dots,\alpha_k^\dagger$ are replaced by a single integer $q$. The Verifier constructs
\begin{equation*}
	M_z \coloneqq \begin{bmatrix}
		2 V_z^{\sf T} V_z  & \mathbf{1} \\
		\mathbf{1}^{\sf T} & 0
	\end{bmatrix} \in \Z^{(k+1)\times(k+1)}
\end{equation*}
from the residual matrix $V_z$ of \cref{alg:verifier}, and solves
\begin{equation}
M_z \begin{bmatrix} \alpha \\ \lambda\end{bmatrix} = \begin{bmatrix}\mathbf 0 \\ 1\end{bmatrix}
	\label{eq:kkt}
\end{equation}
for $(\alpha,\lambda)$ as follows:
\begin{itemize}
\item
reject unless $q$ is prime and $M_z$ is non-singular modulo $q$;
\item solve \eqref{eq:kkt} over $\Q$ as described below, and reject unless the returned $(\alpha,\lambda)$ satisfies \eqref{eq:kkt} under exact substitution over $\Z$ (concretely, writing $\alpha = \alpha^\dagger / D$ and $\lambda = \lambda^\dagger / D$ over the returned common denominator $D > 0$, it checks $M_z (\alpha^\dagger, \lambda^\dagger)^{\sf T} = (\mathbf 0, D)^{\sf T}$ over $\Z$);
\item 
reject unless $\alpha_j > 0$ for every $j \in [k]$;
\item write $\alpha = \alpha^\dagger/D$ over the returned common denominator $D$, and runs the last two lines of \cref{alg:verifier} with $D$ in place of $2^{B_{\epsgap}}$. 
\end{itemize}

\paragraph{Solving \eqref{eq:kkt} over $\Q$.} We describe how to use the standard technique of \cite{Dixon82} to solve \eqref{eq:kkt} without incurring unnecessary precision. 
Simplify the notation by dropping the subscript on $M_z =: M$, and write $b \coloneqq (\mathbf 0, 1)^{\sf T}$ for the right-hand side.
As we will show in \cref{lem:rationalopt} below,
the entries of the solution of \eqref{eq:kkt} are fractions whose numerators and common denominator are at most $2^{B_M}$, where $B_M \coloneqq 2(m+2)(B + \lceil \log_2 (m+2) \rceil) \in O(m \log m + Bm)$.
Rather than carry numbers of this size through Gaussian elimination, choose a prime $q$ and calculate the solution one base-$q$ digit at a time, working modulo $q$ throughout. 
The key requirement on $q$ is that $M$ must be invertible mod $q$, \emph{i.e.,} $C:= M^{-1} \bmod q\ne 0$.
Starting from an initial remainder $r_0 \coloneqq b$, the Verifier computes, for $i = 0, \dots, \ell-1$ with $\ell \coloneqq \lceil (2 B_M + 2) / \log_2 q \rceil$, the $(k+1)$-dimensional integer vectors
\begin{equation}
	x_i \coloneqq C r_i \bmod q
	\qquad\text{and}\qquad
	r_{i+1} \coloneqq \frac{r_i - M x_i}{q},
	\label{eq:lifting}
\end{equation}
The pointwise division in the second case is exact, as $M x_i \equiv M C r_i \equiv r_i \pmod q$, and by induction $\norm{r_i}_\infty \le (k+1) \cdot 2m 2^{2B}$, so each round costs $O(m^2)$ multiplications of $O(B + \log m)$-bit integers by $\log_2 q$-bit digits. Telescoping, $X \coloneqq \sum_{i < \ell} q^i x_i$ satisfies $MX = b - q^\ell r_\ell \equiv b \pmod{q^\ell}$.
The solution $u/d$ (with $d = \det M$ and $u \in \Z^{k+1}$) satisfies $M(dX - u) \equiv 0$, hence $u \equiv dX \pmod{q^\ell}$, as $M$ is invertible modulo $q^\ell$. 
Since $q^\ell > 2 \cdot 2^{2B_M}$, each residue class modulo $q^\ell$ contains at most one fraction with numerator and denominator at most $2^{B_M}$.
The $\ell \in O(B_M / \log q)$ rounds cost $O(m^2 B_M (B + \log m))$ bit operations in total, and the inverse modulo $q$ and the $k+1$ reconstructions no more, so the solve costs $O(m^3 (B + \log m)^2)$ bit operations \cite{Dixon82}.

\paragraph{Soundness of the modified Verifier.} The argument of the previous subsection holds verbatim: the Verifier ends holding a support and weights, so \eqref{eq:V-inc} makes its final check exact. No matter what the Prover supplies, the Verifier accepts only a pair $(\alpha,\lambda)$ satisfying \eqref{eq:kkt} under its own exact substitution check over $\Z$, the last row of which reads $\mathbf 1^{\sf T}\alpha = 1$, and only when every $\alpha_j$ is strictly positive---so the accepted $\alpha$ is a distribution whether or not $q$ was prime and whether or not it divided $\det M_z$. The prime $q$ bears on completeness and on the running time, never on soundness.

For completeness, suppose there exists a distribution $\mu \in \Delta(\bool^n)$ with $\D_\PP(\mu) \le \tau$. By Carath\'eodory's Theorem, we know there exists a $\mu'$ that is supported on only $m+1$ points with the same value; here, we show that in particular there exists one whose weights are pinned down by the support alone, and are rational numbers with $O(m B + m \log m)$ bits. For this, we need the following Lemma.

\begin{lemma}
	\label{lem:rationalopt}
	If $m+1 \ge k$ and $V \in \mathbb Z^{m \times k}$ is an integer matrix whose entries are bounded in magnitude by $2^B$, then the optimization problem
	\begin{equation}
		\min_{\alpha \in \mathbb R^{k}}  \norm[\big]{V \alpha}^2_2
		\quad\text{subject to}~~ \alpha \ge 0, ~\sum_{i=1}^k \alpha_i = 1
		\label{eq:quad-prog}
	\end{equation}
	admits a rational solution $\alpha^*$ whose components have a common denominator representable in
	$O(m \log m + Bm)$
	bits. Moreover, writing $S$ for the support of such a solution of minimal support, $s \coloneqq |S|$, and $M$ for the matrix \eqref{eq:kkt} formed from the columns $V_S$:
	\begin{enumerate}[nosep,itemsep=0.4ex, topsep=0.4ex]
		\item $M$ is non-singular, so that $(\alpha^*_S, \lambda)$ is the \emph{unique} solution of $M(a,\lambda)^{\sf T} = (\mathbf 0, 1)^{\sf T}$;
		\item $\alpha^*_S > 0$ entrywise; and
		\item $|\det M| \le \big(2^B(m+2)\big)^{2(m+2)}$.
	\end{enumerate}
\end{lemma}
\begin{proof}
	Let $\alpha \in \Delta[k]$ be a solution to \eqref{eq:quad-prog} with minimal support $S = \{ i \in [k] : \alpha_i > 0 \}$, and let $s \coloneqq |S|$.  We will show that this solution has the desired rational property.
	Write $\alpha_S$ for the restriction of $\alpha$ to its support, and $V_S \in \Z^{m \times s}$ for the restriction of $V$ to the corresponding columns.
	Since $\alpha_S > 0$ is positive in this subspace, $\alpha_S \in \Delta[s]$ is the solution to the modified problem
	\[
		\min_{a \in \mathbb R^s} \norm[\big]{V_S a }_2^2
		\quad \text{subject to}\quad
		\sum_{i \in S} a_i = 1.
	\]
	The stationary point can be found with a single Lagrange multiplier $\lambda$ corresponding to this one constraint. More precisely, the Lagrangian is given by $\mathcal L(a,\lambda) = \norm{V_S\, a}_2^2 +\lambda(-1 +\sum_i a_i)$, and a solution $a$ must satisfy $\nabla\mathcal L = 0$. Since $\nabla_a \norm{V_S\,a}_2^2 = 2 V_S^{\sf T} V_S\,a$,
	the stationary condition gives us the system of equations
	\begin{equation}
		0  =
		\nabla_{\!a,\lambda} \,\mathcal L(a,\lambda)  =
		\begin{bmatrix}
			2 V_S^{\sf T} V_S \, a + \lambda \mathbf 1 \\
			\mathbf{1}^{\sf T} a - 1
		\end{bmatrix}
		=
		\begin{bmatrix}
			2 V_S^{\sf T} V_S  & \mathbf{1} \\
			\mathbf{1}^{\sf T} & 0
		\end{bmatrix}
		\begin{bmatrix}
			a \\
			\lambda
		\end{bmatrix}
		- \begin{bmatrix}
			0 \\1
		\end{bmatrix},
		\label{eq:matrix-eq}
	\end{equation}
	which is the system \eqref{eq:kkt} formed from $V_S$.
	Since it's a linear system of equations with integer coefficients, it has a rational solution. To bound its magnitude, we use the form on the right-hand side of \eqref{eq:matrix-eq}.
	Let $M \coloneqq \begin{bmatrix}
			2 V_S^{\sf T} V_S  & \mathbf{1} \\
			\mathbf{1}^{\sf T} & 0
		\end{bmatrix} \in \Z^{(s+1) \times (s+1)}$
	be the key matrix on the right.
	We claim that $M$ is non-singular.

	In search of contradiction, suppose that $M y = 0$ for some nonzero $y = (d,\gamma)^{\sf T} \in \R^{s+1}$.
	This means that (a)
	$0 = 2 V_S^{\sf T} V_S d + \gamma \mathbf 1$ and (b) $\mathbf 1^{\sf T} d = 0$. Left-multiplying equation (a) by $d^{\sf T}$ and substituting equation (b), we find that
	\begin{align*}
		0 = 2 d^{\sf T} V_S^{\sf T} V_S d  + \gamma d^{\sf T}\mathbf 1
		\;=\; 2 d^{\sf T} V_S^{\sf T} V_S d
		\;=\; 2 \norm{ V_S d }_2^2,
	\end{align*}
	from which it follows that $V_S d = 0$. Plugging this back into equation (a) above, we find that $0 = \gamma \mathbf 1$, meaning $\gamma = 0$. Since $y =(d,\gamma) \ne 0$ by assumption, this implies $d \ne 0$.
	Altogether, this means we have a vector $d \in \R^s$ with the properties
	\[
		d \ne 0,
		\qquad
		V_S \,d = 0,
		\qquad\text{and}
		\quad
		\mathbf{1}^{\sf T}d = 0,
	\]
	meaning that $\norm{V_S(a + d)}_2^2 = \norm{V_S\,a}_2^2$ and $\mathbf{1}^{\sf T}(a + d) = \mathbf{1}^{\sf T}a = 1$.
	In other words, if $a$ is a solution to the optimization problem, so too is $a+td$, provided $a + td \ge 0$. Since $\mathbf 1^{\sf T} d = 0$ and $d \ne 0$, some component of $d$ is negative, so $t^* \coloneqq \min_{i : d_i < 0} (-a_i / d_i)$ is well defined and positive; $a + t^* d$ is then non-negative, sums to one, and is zero in some component, so it is a solution to the optimization problem with strictly smaller support, contradicting minimality.

	Now that we know $M \in \Z^{(s+1)\times(s+1)}$ is non-singular, we can apply Cramer's rule, which tells us that the unique solution $(\alpha^*,\lambda)^{\sf T}$ at component $i$ is equal to $(\det M_i) / (\det M)$, where $M_i$ is the matrix obtained by replacing the $i^{\text{th}}$ column of $M$ by $[0, \ldots, 0, 1]^{\sf T}$.
	Since $M$ and $M_i$ have integer components, their determinants are integers, and thus $\alpha$ is a rational number whose common denominator divides $|\det M|$.
	All that remains is to bound $\det M$, for which we turn to Hadamard's inequality \cite{hadamard1893resolution}, which states
	that the determinant of a matrix is at most the product of the L2 norms of its columns.
	Every entry of the upper-left block of $M$ is twice a sum of $m$ numbers, each a product of two integers of magnitude at most $2^B$, and is thus at most $2m \cdot 2^{2B}$ in magnitude. The L2 norm of each column is therefore at most $2m\sqrt{s+1} \cdot 2^{2B} +1$, and thus the common denominator is of the order
	\begin{align*}
		\det M & \le \big(2m\sqrt{s+1} \cdot 2^{2B} +1 \big)^{s+1}            \\
		       & \le \big(2^B(m+2)\big)^{2(m+2)} \;\le\; 2^{2(m+2) (B + \lceil \log_2(m+2) \rceil)} \\
		       & \in 2^{O(m \log m + Bm)}
	\end{align*}
	since $s \le k \le m+1$.
	Therefore, it can be represented with $O(m \log m + Bm)$ bits.
	Finally, the non-singularity of $M$ bounds the support: were $s > m+1$, the $s$ columns of $\begin{bsmallmatrix}V_S \\ \mathbf 1^{\sf T}\end{bsmallmatrix} \in \R^{(m+1)\times s}$ would be linearly dependent, giving a nonzero $d$ with $V_S d = 0$ and $\mathbf 1^{\sf T} d = 0$, which the argument above has just ruled out.
\end{proof}

Let's now apply \cref{lem:rationalopt} to finish off the proof of completeness.
Recall that we are in the case where there exists a distribution $\mu$ with $\D_\PP(\mu) \le \tau$, and that Carath\'eodory's theorem guarantees the existence of a witness $(z,\alpha)$ of $m+1$ points.
Applying \cref{lem:rationalopt} to $V = V_z$ and taking $S$ to be the support of a minimal-support optimum $\alpha^*$, the Carath\'eodory weights $\alpha$ are feasible for \eqref{eq:quad-prog}, so $\norm{V_z\alpha^*}_2 \le \norm{V_z\alpha}_2$, and \eqref{eq:V-inc} turns this into $\D_\PP(\mu_{z_S,\alpha^*_S}) \le \D_\PP(\mu_{z,\alpha}) = \D_\PP(\mu) \le \tau$. Its support satisfies $s \le k = m+1$.
The honest Prover sends (only) this support set $z_S$. The Verifier forms $M_{z_S}$, which is non-singular by clause 1, so the solve returns $\alpha^*_S$ and nothing else; these weights are strictly positive by clause 2, so the Verifier does not reject; and the final check passes because $\D_\PP(\mu_{z_S,\alpha^*_S}) \le \tau$. The weights themselves are never transmitted, so the certificate is $s \cdot n \le (m+1)n$ bits of support, together with the prime, which we bound next.

It remains to exhibit the auxiliary prime $q$. Clause 3 of \cref{lem:rationalopt} bounds $\det M_{z_S}$ in magnitude by $2^{B_M}$. Being a nonzero integer of magnitude at most $2^{B_M}$, it has at most $B_M$ distinct prime divisors, so among the first $B_M+1$ primes there is one, call it $q$, that does not divide it; and by Chebyshev's bound the $(B_M+1)$st prime is $O(B_M \log B_M)$, so $q$ is representable in $O(\log B_M) = O(\log m + \log B)$ bits.\footnote{Concretely, at $m$ equal to one million and $B = 32$, a $31$-bit number suffices.} Elimination modulo $q$ then certifies that $\det M_{z_S} \not\equiv 0 \pmod q$, hence that $M_{z_S}$ is non-singular over $\Q$, and the lifting \eqref{eq:lifting} returns the exact rational solution of \eqref{eq:kkt} in $O(m^3 (B+\log m)^2)$ bit operations.
The certificate is therefore of length $(m+1)n + O(\log m + \log B) \in O(mn + \log B)$, as claimed.

Observe that the prime $q$ is not needed for correctness, only for speed, and it costs next to nothing: at $O(\log m + \log B)$ bits it does not affect the certificate length above, and at $m$ equal to one million and $B = 32$ it is a $31$-bit number. Sending it as part of the proof certificate spares Verifier the search, and simplifies our analysis.

For the runtime, the accounting of the previous subsection applies with $B_M$ in place of $B_{\epsgap}$: its inner loop costs $O(m^2(n + B B_M)) \subseteq O(m^2 n + m^3 B(B + \log m))$ and its accumulation costs $O(m B_M^2) = O(m^3 (B + \log m)^2)$, so the final check costs $O(m^3(B+\log m)^2 + m^2 n)$; forming $2 V_z^{\sf T} V_z$ and the exact substitution check cost $O(m^3(B+\log m)^2)$ as well, and adding the $O(m^3(B+\log m)^2)$ of the solve leaves the total unchanged, giving the $O(m^3(B+\log m)^2 + m^2n)$ of \cref{thm:inNP}. This completes its proof.

\section{Hardness of approximating consistency in (deterministic) CPCs}\label{sec:hardness}
In this section we observe that CPCs can be used to express CNF formulas in a straightforward and ``lossless'' fashion.
\unskip\footnote{A related reduction has been documented in the theory probabilistic dependency graphs \cite{richardson2023inference}, but the one we present here is more direct, and the inapproximability lifts more directly in the normalized $\ell_2$ inconsistency measure.}
Consequently, we can use tight inapproximability results for $\EkSAT$ to automatically obtain inapproximability of $\EXPCONSIST$ to the same parameters, and \Cref{cor:mip} lifts this to $\NEXP$-hardness of $\MODELCONSIST$.

The reduction (\Cref{claim:ksat}) preserves many properties of the source formula, such as bounded variable occurrence, arity,\footnote{Formally,
	the \emph{arity} of the probabilistic claim $(x,y,p)$ is equal to $k = k' + 1$, where $k'$ is the number of non-$\ast$ entries in $x$. The arity of a CPC $\PP$ is the maximal arity among its claims.} and the fraction of violated clauses/claims by any assignment. However, $\EXPCONSIST$ allows \emph{distributions} over assignments as ``solutions.'' To argue that the reduction is gap-preserving, we (tightly) bound $\D(\PP)$ of the reduced CPCs in terms of the fraction of claims violated by an assignment in \Cref{lem:dirac}.
\subsection{Violated claims and inconsistency in deterministic CPCs}
The CPCs derived from CNFs are \emph{deterministic}, in the sense that all claims are about events that occur (or do not occur) \emph{with certainty}.
\begin{definition}[Deterministic CPC]
	A CPC $\PP$ is \emph{deterministic} if for all $(x,y,p) \in \PP$ it holds that $p \in \bool$.\footnote{\emph{Tedious comment:} the precision-$B$ numerator convention $p = p^\dagger/2^B$ with $p^\dagger \in \{0,\dots,2^B-1\}$ excludes $p = 1$; for deterministic claims we admit the numerator $p^\dagger = 2^B$ as well, so that $p \in \bool$ is representable at every precision. This adds one numerator value ($B+1$ bits per claim) and affects no bound.} We say that assignment $w \in \bool^n$ \emph{violates} the claim $(x,y,p) \in \PP$ if $w$ agrees with the context $x$, and $w_y = 1-p$. We let $\Viol_\PP(w)$ denote the fraction of claims in $\PP$ violated by $w$,
	\begin{equation*}
		\Viol_\PP(w) \defeq \frac{\left| \left\{P \in \PP : w\textrm{ violates }P\right\} \right|}{|\PP|},\quad
		\Viol(\PP) \coloneqq\min_{w \in \bool^n}\Viol_\PP(w).
	\end{equation*}
\end{definition}
\begin{lemma}\label{lem:dirac}
	Let $\PP$ be a deterministic CPC. Then
	\begin{equation*}
		\Viol(\PP) \leq \D(\PP) \leq \sqrt{\Viol(\PP)}.
	\end{equation*}
\end{lemma}
\begin{proof}[Proof of \Cref{lem:dirac}]
	Let $n$ be the number of variables of $\PP$. First, observe that for any $(x,y,p) \in \PP$ and any $w \in \bool^n$, letting $\delta_w$ denote the point mass distribution supported on $w$,
	\begin{equation}\label{eq:dirac-userful}
		|\delta_w(x \cap y) - p\delta_w(x)| = \begin{cases}
			1\quad \mathrm{if }\ w\text{ violates } (x,y,p) \\
			0\quad \mathrm{otherwise.}
		\end{cases}
	\end{equation}
	Because if $w$ violates $(x,y,p)$ then, by definition, $w$ agrees with $x$ which means that $\delta_w(x) = 1$, but $w_y = 1-p$. By case analysis, this means that $\delta_w(x\cap y) = 1-p$, therefore $(\delta_w(x \cap y) - p\delta_w(x))^2 = 1$. Conversely, if $w$ does not violate $(x,y,p)$, then there are two possibilities: Either $w$ does not extend $x$ in which case $\delta_w(x)=\delta_w(x \cap y)=0$, or $w$ agrees with $x$ and $w_y = p$, in which case $\delta_w(x \cap y) = p=p\delta_w(x)$.

	We now prove the lemma, starting from the leftmost inequality. Let $n$ denote the number of variables of $\PP$, and let $\mu$ be a distribution over $\bool^n$. We will show that $\D_\PP(\mu) \geq \Viol(\PP)$.

	Let $S = \{w_1,\dots,w_s\}$ denote the support of $\mu$. Let $\delta_i \coloneqq \delta_{w_i}$ denote the point mass distribution supported on $w_i$, and so we write $\mu = \sum_{i=1}^s \alpha_i \delta_{i}$ where $\alpha_i > 0$ for all $i \in [s]$ and $\sum_i \alpha_i = 1$. Then we can write
	\begin{align*}
		\D^2_\PP(\mu) & = \frac{1}{m}\sum_{(x,y,p) \in \PP}\left|\mu(x\cap y) - p\mu(x) \right|^2.
		\\
		              & =
		\frac{1}{m}\sum_{(x,y,p) \in \PP}\left(\sum_{i=1}^s\alpha_i \left|\delta_i(x\cap y) - p\delta_i(x) \right| \right)^2
	\end{align*}
	Which, by the Cauchy--Schwartz inequality, is
	\begin{align*}
		\D^2_\PP(\mu) & \geq \frac{1}{m} \cdot \frac{1}{m} \left(\sum_{(x,y,p) \in \PP}\sum_{i=1}^s\alpha_i \left|\delta_i(x\cap y) - p\delta_i(x) \right| \right)^2
		\\
		              & =
		\left(\sum_{i=1}^s\alpha_i \cdot \frac{1}{m}\sum_{(x,y,p) \in \PP} \left|\delta_i(x\cap y) - p\delta_i(x) \right| \right)^2,
	\end{align*}
	By \Cref{eq:dirac-userful} and definition of $\Viol(\PP)$ as $\min_w \Viol_{\PP}(w)$ we have
	\begin{equation*}
		\D^2_\PP(\mu) \geq \left( \sum_{i=1}^s\alpha_i \Viol_\PP(w_i) \right)^2 \geq \left( \Viol(\PP) \cdot \sum_{i=1}^s \alpha_i  \right)^2 = \Viol(\PP)^2.
	\end{equation*}
	Taking the square root of both sides gives the lower bound.

	For the upper bound, let $w \in \bool^n$, we will show that $\Viol_\PP(w) = \D^2_\PP(\delta_w)$ and so in particular $\Viol_\PP(w) \geq \D^2(\PP)$. By \Cref{eq:dirac-userful} and definition of $\Viol_\PP(w)$ as the fraction of claims violated by $w$,
	\begin{equation*}
		\D^2_\PP(\delta_w) = \frac{1}{m}\sum_{(x,y,p) \in \PP}(\delta_w(x \cap y) - p\delta_w(x))^2 = \Viol_\PP(w),
	\end{equation*}
	using the fact that the summands are Boolean (\Cref{eq:dirac-userful}), therefore each summand's square coincides with its absolute value.

\end{proof}

\paragraph{Tightness of \Cref{lem:dirac}.} The lower and upper bounds are of course attained for consistent deterministic CPCs, because all quantities are zero. Let us briefly detour to discuss tightness for inconsistent instances. The lower bound can still be an equality: consider a CPC containing just two contradictory claims $\{(\ast, y,0), (\ast, y,1)\}$; then $\Viol(\PP) = 1/2$, and $\D_\PP(\U) = (1/4 + 1/4) = 1/2$ as well for the uniform distribution over its single variable.

For the upper bound, we first note that it becomes an equality if \emph{and only if} the CPC is consistent. To see this, let $\PP$ be an inconsistent deterministic CPC with $|\PP| = m$ claims, and let $w^\ast \in \bool^n$ be a maximally satisfying assignment. Take any claim $P \in \PP$ violated by $w^\ast$, and take $w_P \in \bool^n$ to be any assignment that satisfies the claim $P$. Consider the distribution $\mu$ that assigns mass $(1-1/m)$ to $w^\ast$ and $1/m$ mass to $w_P$. Then any claim satisfied by $w^\ast$ is violated by $\mu$ with probability at most $1/m$, but now $P$ is violated by $\mu$ with probability $(1-1/m)$, giving
\begin{equation}\label{eq:uppertight}
	\D_\PP(\mu)^2 \leq \Viol_\PP(w^\ast) - 1/|\PP|^2 < \Viol_\PP(w^\ast) = \Viol(\PP)
\end{equation}

That said, the inequality in \Cref{eq:uppertight} (and the upper-bound of \Cref{lem:dirac} more generally) can approach equality as the number of claims $m$ grows by simply adding more consistent claims: Consider a single-variable CPC $\PP_m$ containing a single copy of $\{*,y,0\}$ and $m-1$ copies of $\{*,y,1\}$. Then $\Viol(\PP) = 1/m$, but similarly to the argument before \Cref{eq:uppertight}, a distribution $\mu$ that assigns $y=0$ with probability $1/m$ will have $\D_\PP(\mu)^2 = \Viol(\PP) - 1/m^2$, so the ratio $\D_\PP(\mu)^2 / \Viol(\PP) \to 1$ as $m \to \infty$.

\subsection{Reduction from $\EkSAT$.}
Let $F$ be an Exact $k$CNF formula over $n$ Boolean variables. We let $|F|$ denote its number of clauses, and for any assignment $w \in \bool^n$ we let $\Viol_F(w)$ denote the number of clauses violated by $w$ in $F$. We  write
\begin{equation*}
	\Viol(F) = \frac{1}{|F|}\min_{w \in \bool^n}\Viol_F(w).
\end{equation*}
\begin{claim}\label{claim:ksat}
	There is a polynomial-time computable reduction mapping Exact $k$CNF formulas $F$ to deterministic CPCs $\PP_F$, such that $\Viol(F) = \Viol(\PP_F)$.
\end{claim}
\begin{proof}
	Let $F = \{F_j\}_{j=1}^m$ be a Exact $k$CNF formula with $m$ clauses over $n$ variables. We write $F_j = (j_1=b_1 \vee \dots \vee  j_k=b_k)$ for variables $j_1,\dots,j_k \in [n]$ and bits $b_1,\dots,b_k \in \bool$.\footnote{Formally, $b_i=0$ indicates that the variable $j_i$ is negated in the clause, and we assume the $j_1,\dots,j_k$ are sorted.} For each clause $F_j$, the reduction adds the Probabilistic Claim $P_j = (x^j,j_k,b_k)$ where $x^j_{j_i} = (1-b_i)$ for each $i \in [k-1]$, and the rest of $x^j$'s entries are set to $\ast$. The output of the reduction is the instance $\PP = \{P_j\}_{j=1}^m$.

	Note that $|F| = |\PP_F|$, and so to prove the claim it suffices to show that for any assignment $w \in \bool^n$, the number of clauses of $F$ that $w$ violates equals the number of claims of $\PP$ that $w$ violates; since $|F| = |\PP_F|$, the violated \emph{fractions} then coincide as well. We will show a correspondance between violated clauses and claims: $w$ violates the clause $F_j = (j_1=b_1 \vee \dots \vee  j_k=b_k)$ if and only if $w_{j_i} = 1 - b_i$ for all $i \in [k]$. By construction, this is equivalent to having $w$ extend $x^j$ but $w_{j_k} = 1 - b_k$, \ie, $w$ violates $P_j$.

\end{proof}
\begin{claim}\label{claim:hardnessofapx}
	For all $\delta \in (0, 2^{-k})$ and $k \geq 3$, the following promise problem is $\NP$-hard. The instance is a (deterministic) CPC $\PP$ such that $\D(\PP) \leq 2^{-k}$, and the goal is to distinguish
	\begin{itemize}
		\item YES: $\D(\PP) = 0$, from
		\item NO: $\D(\PP) \geq 2^{-k} - \delta$,
	\end{itemize}
\end{claim}
\begin{proof}
	Fix $k \geq 3$ and $\delta \in (0, 2^{-k})$. Theorems 6.5 and 6.14 of \cite{Hastad01} give a polynomial time reduction from any $L \in \NP$ to Exact $k$CNF formulas $F$ mapping instances of $L$ to \begin{itemize}
		\item YES: The set of satisfiable Exact $k$CNF formulas $F$, \ie, $\Viol(F) = 0$,
		\item NO: The set of Exact $k$CNF formulas $F$ such that $\Viol(F) \geq 2^{-k} - \delta$.
	\end{itemize}
	For each Exact $k$CNF $F$, we apply $\Cref{claim:ksat}$ to obtain the corresponding deterministic CPC $\PP_F$.

	We first observe that $\D(\PP_F) \leq 2^{-k}$. This follows from the fact that a clause is violated by a random assignment with probability $2^{-k}$: letting $\U$ denote the uniform distribution over $\bool^n$,
	\begin{equation*}
		\D(\PP_F) \leq \D_{\PP_F}(\U) = \sqrt{\frac{1}{|F|}\sum_{(x^j,j_k,b_k) \in \PP_F} \left( \U(x^j \cap j_k) - b_k \U(x^j) \right)^2} = \sqrt{\frac{1}{|F|} \cdot |F| \cdot 2^{-2k}} = 2^{-k},
	\end{equation*}
	where the rightmost equality uses the fact that $x^j$ has exactly $k-1$ non-$\ast$ entries.

	Completeness and soundness of the reduction follow from combining \Cref{claim:ksat} with \Cref{lem:dirac}. If $\Viol(F) = 0$ then by \Cref{claim:ksat}, $\Viol(\PP_F) = 0$, and so by \Cref{lem:dirac}, $\D(\PP)$ is sandwiched between zeros. Conversely, if $\Viol(F) \geq 2^{-k} - \delta$ then by the lower bound of the lemma, $\D(\PP) \geq 2^{-k} - \delta$.
\end{proof}

\subsection{$\NEXP$-completeness of $\MODELCONSIST$}\label{sec:nexp-completeness}
Recall the implicit CPC $\PP_{P,Q}$ of a predictive model (\Cref{def:implicit-cpc}): the exponentially long collection that has, for each query $q \in \Qcal$, $Q(q)$ copies of the claim $(x_q, \omega_t, P(q)/2^B)$, and whose consistency is that of the model (\Cref{claim:implicit-cpc}). With this translation in hand, $\NEXP$-completeness is short to derive.

\begin{corollary}\label{cor:mip}
	$\MODELCONSIST$ is $\NEXP$-complete (under polynomial-time many-one reductions).
\end{corollary}
\begin{proof}
	For membership, let $V_\NP$ be the verifier of \Cref{thm:inNP}, and let $V_\NEXP$ be the verifier that, on input $(P,Q,\tau)$ and proof $\pi$, constructs $\PP_{P,Q}$ explicitly---in time $2^{O(\ell d + B)} \cdot O(S)$, where $S$ bounds the sizes of the circuits $P$ and $Q$---and accepts if and only if $V_\NP(\PP_{P,Q}, \tau, \pi)$ accepts. By \Cref{claim:implicit-cpc}, $(P,Q)$ is $\tau$-consistent if and only if $\PP_{P,Q}$ is, so completeness and soundness are those of \Cref{thm:inNP}. Using the fact that $||Q|| \in 2^{O(\ell d + B)}$, \Cref{thm:inNP} implies that running $V_\NP$ on $\PP_{P,Q}$ takes time $||Q||^3(B + \log ||Q||)^2 + ||Q||^2 n \in 2^{O(\ell d + B)}$ with proofs of length $O(||Q|| \cdot n + \log B) \in 2^{O(\ell d + B)}$, so all together $V_\NEXP$ runs in time exponential in its input length. For hardness, $\MODELCONSIST$ is the succinct version of $\EXPCONSIST$, so the reduction of \Cref{claim:ksat} applies as a reduction from the $\mathsf{SUCCINCT}$-$\EkSAT$ problem (see, \eg, \cite{PapadimitriouY86,Goldreich08}).
\end{proof}

\begin{remark}[A pragmatic aside on implementation]\label{rem:implementation}
	Producing a faithful implementation of the black-box construction underlying \Cref{cor:mip} (specification, code, and checking that the two match) is, in the authors' own experience and despite considerable effort, well beyond the scope of an empirical pilot: we speculate it would occupy a graduate student fluent in the relevant literature for a semester or two: disentangling the Cook--Levin compilation and the BFLS arithmetization layers from one another. The remainder of this paper is dedicated to constructing an Interactive PCP that admits a self-contained specification (\Cref{alg:full-protocol}) whose proof oracle is simply a specialized encoding of the witnessing distribution (\Cref{sec:codecheck}); the same hypothetical graduate student should 
    be able to get experiments running on top of our \Cref{alg:full-protocol}, developed over the next two sections, in a matter of weeks. That said, even in our explicit IPCP, the Prover's job remains difficult enough that deploying the full proof system at scale remains, for now, out of reach.
\end{remark}

\section{\texorpdfstring{Reed--$\mu$ller}{Reed-muller}: a marginal-verifiable distribution encoding}\label{sec:codecheck}
The PCP component of our IPCP will be a locally-verifiable encoding of distributions which enables efficient (interactive) verification of marginals. This is essentially an encoding of distributions that enables \emph{delegation of marginal computation}; we present it as a self-contained ``library'' consisting of the encoding method (\Cref{def:code}), codeword-validity verifier (\Cref{alg:encoding}) and marginal verifier (\Cref{alg:marginal-pp}), which may be of independent interest.

To keep the exposition self-contained, we open with a standalone motivating story, in the storytelling tradition of \cite{Babai85,Babai90}.

\begin{rulequote}
	All-powerful Merlin (the Prover) holds a distribution $\mu$ over $\bool^n$, the truth-assignments to $n$ propositions. Each proposition is named by a short address of $d = \log n$ bits; picture a sentence king Arthur (the Verifier) could scrawl on one line, a hundred bits or so, so that $n = 2^d \approx 2^{100}$. Merlin's $\mu$ is the verdict of a finite engine, vastly smaller than the $2^n$ worlds it opines on, so it can deem only a modest number $m$ of assignments possible. He describes it by its support: points $z_1,\dots,z_m \in \bool^n$ with weights $\mu(z_1),\dots,\mu(z_m) \in [0,1]$. (In the application this section serves, Carath\'eodory pins $m = \poly(n)$; see \Cref{sec:NP}. In this section, $m$ is a free parameter.)

	Arthur, his computationally bounded sovereign, wants two assurances. First, that $\mu$ is a distribution at all: that the weights are nonnegative and that
	\begin{equation*}
		\sum_{j=1}^{m}\mu(z_j) = 1 .
	\end{equation*}
	Second, he wants to interrogate it. To learn, say, the chance that the first proposition holds and the second fails, Arthur must total the mass of every possible world that agrees:
	\begin{equation*}
		\mu(\omega_1=1,\ \omega_2=0) \;=\; \sum_{j   :   z_{j,1}=1,\ z_{j,2}=0} \mu(z_j) .
	\end{equation*}
	Each such question is a sum over Merlin's $m$ possible worlds, costing $\Omega(m)$ to answer by hand. And $m$ is the least of Arthur's troubles: his entire budget is $\poly(d) = \poly(\log n)$, so he cannot even tour the $n = 2^d$ propositions, let alone audit a sum over worlds.

	This section is the remedy. Merlin encodes $\mu$ once, as $\pi(\mu)$, and hands it over; thereafter, with a few rounds of interaction and his own coin flips, Arthur convinces himself that $\mu$ is a genuine distribution in time $O(\log m + \log n)$, and verifies any marginal $\mu(\omega_1=1,\omega_2=0)$ just as cheaply. More generally, he verifies the marginal on any $\ell$ propositions in time $O(\ell\log m + \log n)$.
\end{rulequote}

The encoding of $\mu$ is a Reed--Muller encoding of the support vectors and their weights (in binary); as such, we denote it $\Code$.\footnote{The notation reads as ``Reed--$\mu$ller,'' a Reed--Muller (RM) encoding of the distribution $\mu$.}

Before we define it, we first recall the notion of a \emph{multilinear extension}.
\begin{definition}[$i$-linear, multilinear extension]
    Fix $k \in \N$ and $i \in [k]$, we say that a function $f \colon \F^k \to \F$ is \emph{$i$-linear} if, for all $x_1,\dots,x_{i-1},x_{i+1},\dots,x_k \in \F$, the univariate function $f' \colon \F \to \F$ defined by $f'(\xi) \defeq f(x_1,\dots,\xi,\dots,x_k)$ is linear. We say that a function $f$ is \emph{multilinear} if it is $i$-linear in each of its coordinates.

     For a function $g_0 \colon \bool^k \to \bool$, its \emph{multilinear extension} $\hat g \colon \F^k \to \F$ is defined by $\hat g(x) \defeq \sum_{x' \in \bool^k} g_0(x') \prod_{i=1}^k \bigl(x_i x'_i + (1- x_i)(1-x'_i)\bigr)$. It is the only multilinear function agreeing with $g_0$ on the entirety of $\bool^k$.
\end{definition}

\begin{definition}[\Code]\label{def:code}
	Let $n,m,B \in \N$ and let $\mu$ be a distribution over $\bool^n$. Assume that $\mu$ is supported on $m$ points with weights expressible as $B$-bit precision rationals; that is, there are $z_1,\dots,z_{m} \in \bool^n$, and $\alpha^\dagger_1,\dots,\alpha^\dagger_{m} \in \{0,\dots,2^{B}-1\}$ such that $\mu$ is supported on $z_j$ with weight $\mu[z_j] = \alpha^\dagger_j / 2^{B}$.

	For $j \in [m], i \in [n]$ we let $z_{j,i}$ denote the $i$th coordinate of $z_j$. Similarly, for $j \in [m]$ and $\beta \in \{0,\dots,B-1\}$ let $\alpha^\dagger_{j,\beta}$ denote the $\beta$th bit of $\alpha^\dagger_j$, such that $\alpha^\dagger_j = \sum_{\beta=0}^{B-1}2^{\beta}\alpha^\dagger_{j,\beta}$. We denote
	\begin{align*}
		Z_0 \colon \bool^{\log{m}} \times \bool^{\log n} \to \bool \quad & Z_0(j,i) = z_{j,i}                      \\
		A_0 \colon \bool^{\log{m}} \times \bool^{\log B} \to \bool \quad & A_0(j,\beta) = \alpha^\dagger_{j,\beta}
	\end{align*}
	For a finite field $\F$, let $\tilde Z$ (resp. $\tilde A$) denote the multilinear extension of $Z_0$ (resp. $A_0$). A $\Code_\F(n,m,B)$-encoding of $\mu$ is the pair $(\tilde Z, \tilde A)$. We use $\Code_\F(n,m,B)$ to denote the set of all such valid encodings. Namely, $(Z,A) \in \Code_\F(n,m,B)$ if and only if
	\begin{enumerate}
		\item \emph{Multilinearity:} Both $Z$ and $A$ are multilinear.
		\item \emph{Booleanity:} For any $j \in \bool^{\log m}$, $i \in \bool^{\log n}$, and $\beta \in \bool^{\log B}$ it holds that $\{Z(j,i)$, $A(j,\beta) \} \subseteq \bool$.
		\item \emph{Unit measure:} For $j \in \bool^{\log m}$, let $a^\dagger_j \in \{0,\dots,2^B-1\}$ be the integer represented by the $B$-bit string $(A(j,\beta))_{\beta \in \bool^{\log B}}$. Then $\sum_{j \in \bool^{\log m}}a^\dagger_j = 2^B$.
	\end{enumerate}
\end{definition}

We note that $\Code$ is not an error-correcting code in the standard sense, because a distribution $\mu$ supported on $m$ vectors $z_1,\dots,z_m$ may be encoded in one of $m!$ ways (one for each ordering of its support vectors). One could require the vectors $z_1,\dots,z_m$ to be ordered lexicographically, but then one would also have to test that an alleged codeword respects this order---an added complication which is entirely unnecessary for our purposes. Hereafter when we consider an encoding of $\mu$, we let it be an encoding in an arbitrary order.

\subsection{Prerequisites: the sum-check protocol, and four multilinear gadgets}
We recall the sum-check protocol \cite{LundFKN92}, modified slightly for convenience.

\begin{algorithm}[h!]
	\caption{$\SumCheck_\F(\val, \rounds,\degree)$.}\label{alg:sumcheck}
	\KwInput{$\rounds,\degree \in \N$, a finite field $\F$ of odd characteristic, and $\val \in \F$.}
	\KwOutput{$\hat y \in \F^\rounds$ and $\hat v \in \F$.}
	\BlankLine
	\Verifier{Initialize the running value $u \gets v$.}\;
	\For{$k = 1, \dots, \rounds$}{
		\Prover{Send $a_1, a_2, \dots, a_\degree \in \F$. }\;
		\Verifier{Compute $a_0 \coloneqq 2^{-1} \cdot (u - a_1 - a_2 - \cdots - a_\degree)$ where $2^{-1}$ is over $\F$. Sample $\hat y_k \in \F$ uniformly at random and send $\hat y_k$ to the Prover. Update $u \gets \sum_{i=0}^\degree a_i \hat y_k^i$.}\;
	}
	\Verifier{Output $\hat v \gets u$ and $\hat y \gets (\hat y_1, \dots, \hat y_\rounds)$.}
\end{algorithm}
\begin{fact}[\cite{LundFKN92}]\label{fact:sumcheck}
	Let $r, d \in \N$, a finite field $\F$ of odd characteristic, and $v \in \F$. Let $\varphi \in \F[Y_1,\dots,Y_r]$ be an $r$-variate polynomial that has individual degree at most $d$ in each of its variables. $(\hat v, \hat y)$ denote the output of $\SumCheck_{\F}(\val \gets v,\rounds \gets r,\degree \gets d)$, noting that this is a random variable determined by the sampled $\hat{y}$. Then,
	\begin{itemize}
		\item \emph{Completeness.} If $v = \sum_{y \in \bool^r}\varphi(y)$ over $\F$ then there exists an honest Prover such that $\hat v = \varphi(\hat y)$ with probability $1$ over $\hat y$.
		\item{Soundness.} Otherwise, $\hat v = \varphi(\hat y)$ with probability at most $rd / |\F|$.
	\end{itemize}
	Furthermore, \Cref{alg:sumcheck} has $r$ many rounds, in which the Prover (resp. Verifier) sends a total of $r\cdot d$ (resp. $r$) field elements. The Verifier makes a total of $O(rd)$ field operations. Lastly, $\hat y$ is distributed marginally uniformly in $\F^{\rounds}$.
\end{fact}

\Cref{alg:sumcheck} is the sum-check protocol of \cite{LundFKN92}, presented with two cosmetic changes that streamline its use as a subroutine. First, the per-round equality $g_k(0)+g_k(1) = u_{k-1}$ that the standard Verifier checks is instead enforced by construction: In the standard protocol the Prover sends, in round $k$, an individual degree-$d$ polynomial $g_k(Y) = \sum_{i=0}^{d} a_i Y^i$, and the Verifier checks $g_k(0)+g_k(1) = u_{k-1}$ before sampling $\hat y_k$ and updating $u_k \gets g_k(\hat y_k)$; whereas in \Cref{alg:sumcheck} the Prover sends only the $d$ non-free coefficients $a_1,\dots,a_d$, and the Verifier sets the free coefficient to $a_0 \gets 2^{-1}(u_{k-1} - a_1 - \cdots - a_d)$, the unique value for which $g_k(0)+g_k(1) = u_{k-1}$.\footnote{This is the one place odd characteristic of $\F$ is used}
Second, the final test $u_r \qeq \varphi(\hat y)$ that the standard Verifier asserts is deferred: \Cref{alg:sumcheck} outputs the pair $(\hat y, \hat v) = (\hat y, u_r)$ and leaves the comparison $\hat v \qeq \varphi(\hat y)$ to its caller.

We find this presentation useful for two reasons: The Verifier of \Cref{alg:sumcheck} never rejects; it only samples randomness and maintains the running value $u$, so its actions depend on $v$, $r$, and $d$ alone and never on $\varphi$. And the round polynomials $g_k$ never enter the Verifier's computation: they are objects of the analysis, not of the protocol. We still name the parties Prover and Verifier because, facing forward, each invocation of \Cref{alg:sumcheck} in \Cref{sec:mip:protocol} will be run by the main protocol's Verifier.

The protocols accompanying $\Code$ below rely on four multilinear polynomials ``gadgets.'' We collect them here.

\paragraph{$\eq$:}
For $t \in \N$ and $\hat y, y' \in \F^{t}$,
\begin{equation}\label{eq:eqdef}
	\eq(\hat y, y') \defeq \prod_{i=1}^{t}\hat y_i y'_i + (1-\hat y_i)(1-y'_i).
\end{equation}
On Boolean $\hat y, y'$ it equals $1$ if $\hat y = y'$ and $0$ otherwise, so it is the (unique) multilinear extension of the equality indicator on $\bool^t$.

\paragraph{$\exptwo$:} For $\beta \in \F^{\log B}$,
\begin{equation}\label{eq:exp2def}
	\exptwo(\beta) = \exptwo(\beta_1, \dots, \beta_{\log B}) \defeq \prod_{k=1}^{\log B} 1 + \beta_k (2^{2^{k-1}} - 1).
\end{equation}
It is multilinear, and on a Boolean $\beta \in \bool^{\log B}$ it equals $2$ raised to the integer value of the bit string $\beta$ (with $\beta_1$ the least significant bit). These are exactly the coefficients that decode an integer from its binary expansion: $\sum_{\beta \in \bool^{\log B}} \exptwo(\beta)   a_\beta$ is the integer with bits $(a_\beta)_\beta$. The Verifier can compute $\exptwo(\hat\beta)$ in $O(\log B)$ field operations at any $\hat\beta \in \F^{\log B}$.

\paragraph{$\weight$:}
For an oracle $A \colon \F^{\log m + \log B} \to \F$ and $j \in \F^{\log m}$, write
\begin{equation}\label{eq:Wdef}
	\weight^{A}(j) \defeq \sum_{\beta \in \bool^{\log B}} \exptwo(\beta)\cdot A(j,\beta),
\end{equation}
the integer that $\exptwo$ decodes from the bits $A$ stores in its $j$-th row. On a codeword, $A$ holds the weight bits, so $\weight^{A}(j) = \alpha^\dagger_j$ is the integer weight numerator of the $j$-th support point and $\weight^{A}(j)/2^B = \mu[z_j]$ its probability mass; for an arbitrary multilinear $A$ it is the same $\F$-linear combination of the rows $A(\cdot,\beta)$, hence multilinear in $j$.

\paragraph{$\match$:}
Fix an oracle $Z \colon \F^{\log m + \log n} \to \F$, points $s = (s_1, \dots, s_\ell)$ with $s_k \in \F^{\log n}$, and bits $b = (b_1, \dots, b_\ell) \in \F^\ell$. For $j \in \F^{\log m}$,
\begin{equation}\label{eq:satdef}
	\match^{Z}_{b,s}(j) \defeq \prod_{k=1}^{\ell} b_k Z(j,s_k) + (1-b_k)(1-Z(j,s_k)).
\end{equation}
Each factor is the equality indicator between the bit $b_k$ and the value $Z(j,s_k)$, so $\match^{Z}_{b,s}(j)$ is just $\eq$ of the bit-vector $b$ against the slice-values $(Z(j,s_k))_{k=1}^{\ell}$:
\begin{equation*}
	\match^{Z}_{b,s}(j) = \eq(b, (Z(j,s_k))_{k=1}^{\ell}).
\end{equation*}
For multilinear $Z$ the map $j \mapsto \match^{Z}_{b,s}(j)$ has individual degree at most $\ell$ in $j$, and on a codeword it indicates, at Boolean $j$, whether the $j$-th support point satisfies every context literal $z_{j,s_k} = b_k$.

\subsection{Verifying proximity to the code}

Proximity is measured in relative Hamming distance: for functions $f, g \colon \F^k \to \F$, $\dist(f,g) \defeq \Pr_{x \sim \F^k}[f(x) \neq g(x)]$ is the fraction of inputs on which they disagree, and $f$ is \emph{$\delta$-far} from a set of functions if $\dist(f,g) \geq \delta$ for every $g$ in the set (otherwise \emph{$\delta$-close} to it). For the pairs at hand, the worst part governs: the distance between $(Z,A)$ and a codeword $(\tilde Z, \tilde A)$ is $\max\{\dist(Z,\tilde Z), \dist(A,\tilde A)\}$.

\Cref{alg:encoding} tests that a pair of oracles $(Z,A)$ is close to a valid codeword: that both $Z$ and $A$ are multilinear, that each is Boolean on the hypercube, and that the masses $A$ encodes sum to one.

\begin{fact}[\cite{BabaiFL91,FeigeGLSS96}]\label{fact:fglss}
	Let $f \colon \F^k \to \F$. Consider a tester as follows: (1) Sample $i \in [k]$ uniformly at random; (2) Sample three points $x,y,z \in \F^k$ that differ only on the $i$th coordinate, uniformly at random; (3) Accept if and only if $f(x),f(y),f(z)$ agree with an $i$-linear function.
	Then, for any $\delta \in (0,1)$ such that $6k / |\F| < \delta < 1/2$,
	\begin{enumerate}
		\item \emph{Completeness:} If $f$ is multilinear, the tester accepts with probability $1$.
		\item \emph{Soundness:} If $f$ is $\delta$-far from multilinear in relative Hamming distance, the tester accepts with probability at most $1-\delta / k$.
	\end{enumerate}
\end{fact}
We note that although \cite{FeigeGLSS96} state the soundness of their tester only for $\delta \geq 0.1$, \Cref{fact:fglss} follows immediately from their intermediate argument in Lemma 4.3.1.3. Namely, their lemma guarantees rejection with probability at least $3(1-\delta)\delta/k - 3/|\F|$. We then use the additional assumptions made in \Cref{fact:fglss} to obtain
\begin{equation*}
	\frac{3(1-\delta)\delta}{k} - \frac{3}{|\F|} \geq \frac{3 \delta}{2k} - \frac{3}{|\F|} \geq \frac{3\delta}{2k} - \frac{\delta}{2k} = \frac{\delta}{k},
\end{equation*}
using first the assumption that $\delta < 1/2$, then the assumption that $6k/|\F| < \delta$.

\begin{algorithm}[h!]
	\caption{$\VerifyEncoding_\F^{Z,A}(\delta, \eps)$}\label{alg:encoding}
	\KwInput{Finite field $\F$, a proximity parameter $\delta \in (0,1)$, and a soundness parameter $\eps \in (0,1)$. Oracle access to $Z \colon \F^{\log m + \log n} \to \F$ and $A\colon\F^{\log m + \log B} \to \F$ for $n,m,B \in \N$.}
	\BlankLine
	Repeat the multilinearity test of \Cref{fact:fglss} $\bigl\lceil(\log m + \log n)\max\{\ln(1/\eps)/\delta,\ 2/\eps\}\bigr\rceil$ times on oracle $Z$, and $\bigl\lceil(\log m + \log B)\max\{\ln(1/\eps)/\delta,\ 2/\eps\}\bigr\rceil$ times on oracle $A$. Reject if any test rejects.\;
	Sample a random $y' \in \F^{\log m + \log n}$ and send it to the Prover. Run
	\begin{equation*}
		(\hat v, \hat y) \defeq \SumCheck_\F(\val \gets 0,\rounds \gets \log m + \log n, \degree \gets 3)
	\end{equation*}
	Assert that
	\begin{equation}\label{eq:code:boolZ}
		\hat v \qeq \left(Z^2(\hat y) - Z(\hat y)\right) \cdot \eq(\hat y, y').
	\end{equation}
	\;
	Sample a random $x' \in \F^{\log m + \log B}$ and send it to the Prover. Run
	\begin{equation*}
		(\hat u, \hat x) \defeq \SumCheck_\F(\val \gets 0,\rounds \gets \log m + \log B, \degree \gets 3)
	\end{equation*}
	Assert that
	\begin{equation}\label{eq:code:boolA}
		\hat u \qeq \left(A^2(\hat x) - A(\hat x)\right) \cdot \eq(\hat x, x').
	\end{equation}
	\;
	Run $\left(\hat w, (\hat j, \hat \beta ) \right) \defeq \SumCheck_\F(\val \gets 2^B,\rounds \gets \log m + \log B, \degree \gets 2)$; here we mean that $\hat w \in \F$ is the first output, and $(\hat j, \hat \beta)$ is the second output, with $\hat j \in \F^{\log m}$ and $\hat \beta \in \F^{\log B}$. Assert that
	\begin{equation}\label{eq:code:normA}
		\hat w \qeq A(\hat j, \hat \beta) \cdot \exptwo(\hat \beta).
	\end{equation}\;
	If any assertion failed, reject. Otherwise, accept.
\end{algorithm}
\begin{claim}\label{claim:ipp}
	Suppose $\F$ is a prime field with $2^B < |\F|$, and $\VerifyEncoding_\F^{Z,A}(\delta, \eps)$ is run on $Z \colon \F^{\log m + \log n} \to \F$ and $A \colon \F^{\log m + \log B} \to \F$ such that
	\begin{equation}\label{eq:code:delta}
		m\cdot (2^B - 1)<|\F|,\quad 6\cdot \frac{\log m + \max(\log n, \log B)}{|\F|} < \delta < \tfrac12,\quad 10\cdot\frac{\log m + \max(\log n, \log B)}{|\F|} \le \eps
	\end{equation}
	Then,
	\begin{itemize}
		\item \emph{Completeness:} If $(Z,A) \in \Code_\F(n,m,B)$, there exists an honest Prover such that the Verifier accepts with probability 1.
		\item \emph{Soundness:} If
		      \begin{equation*}
			      \min_{(\tilde Z, \tilde A) \in \Code_\F(n,m,B)} \max\{\dist(Z,\tilde Z), \dist(A, \tilde A)\} \geq \delta,
		      \end{equation*}
		      then the Verifier accepts with probability at most $\eps$.
	\end{itemize}
	Furthermore, the Verifier runs in $O\bigl((\log m + \log n + \log B)\max\{\log(1/\eps)/\delta,\ 1/\eps\}\bigr)$ field operations.
\end{claim}
\begin{proof}
	At a high level, the Verifier tests that $Z$ and $A$ are close to multilinear (amplifying this test, by repetition, to soundness $\eps$), and then $\SumCheck$s that each is Boolean on the hypercube (\Cref{eq:code:boolZ,eq:code:boolA}) and that the masses encoded by $A$ sum to one (\Cref{eq:code:normA}). The three $\SumCheck$s are sound by the assumption on the field size, namely, $|\F| \geq \Omega((\log m + \max(\log n, \log B)) / \eps)$.

	We will first show that \Cref{eq:code:boolZ,eq:code:boolA,eq:code:normA} establish completeness and soundness under the assumption that $Z$ and $A$ were multilinear, and remove this assumption in the end. Rather than separating completeness and soundness, we will show that each check of the Verifier asserts a necessary property of $\Code$, and that if all properties are met then $(Z,A) \in \Code$.

	\pparagraph{Booleanity of a multilinear $Z$.}
	For any fixed $y' \in \F^{\log m + \log n}$, consider the polynomial
	\begin{equation*}
		G_{y'}(y) = \left(Z^2(y) - Z(y) \right)\cdot \eq(y,y').
	\end{equation*}
	It is a $(\log m + \log n)$-variate polynomial. By $Z$'s assumed multilinearity, $G_{y'}$ has individual degree at most $3$. Finally, if $Z$ was Boolean on $\bool^{\log(m) + \log n}$ then for Boolean-entry $y$, $Z^2(y) = Z(y)$, therefore for all $y' \in \F^{\log m + \log n}$
	\begin{equation*}
		\sum_{y \in \bool^{\log m + \log n}}(Z^2(y) - Z(y)) \eq(y, y') = \sum_{y \in \bool^{\log m + \log n}} 0 \cdot \eq(y, y') = 0
	\end{equation*}
	Therefore, \Cref{fact:sumcheck} applied to $\SumCheck(0,\log m + \log n, 3)$ implies that \Cref{eq:code:boolZ} is satisfied with probability 1 if the Prover participates honestly in the protocol.

	For soundness, suppose that $Z$ is not Boolean on the hypercube. Then there exists $y^\ast \in \bool^{\log m + \log n}$ such that $Z^2(y^*) - Z(y^*) \neq 0$. Consider the polynomial
	\begin{equation*}
		F(y') \defeq \sum_{y \in \bool^{\log m + \log n}}(Z^2(y) - Z(y)) \eq(y, y').
	\end{equation*}
	We have $F(y^*) = Z^2(y^*) - Z(y^*) \neq 0$, and $F$ has total degree at most $\log m + \log n$, therefore, by the Schwartz--Zippel lemma~\cite{Schwartz80,Zippel79},
	\begin{equation*}
		\Pr_{y' \sim \F^{\log m + \log n}}\left[ F(y') = 0 \right] \leq \frac{\log m + \log n}{|\F|}.
	\end{equation*}
	Assuming $F(y') \neq 0$ for the $y'$ sampled and sent by Verifier, by soundness of $\SumCheck$, the Verifier accepts with probability at most $3\cdot (\log m + \log n) / |\F|$. Combining with a union bound, we have that the Verifier accepts \Cref{eq:code:boolZ} with probability at most
	\begin{equation*}
		\frac{\log m + \log n}{|\F|} + \frac{3(\log m + \log n)}{|\F|} = \frac{4(\log m + \log n)}{|\F|} \le \frac{2}{5} \eps,
	\end{equation*}
	where the last inequality uses $10(\log m + \log n)/|\F| \le \eps$ from \Cref{eq:code:delta}.
	\pparagraph{Booleanity of a multilinear $A$.} By an identical argument, we have that if $A$ was Boolean on the hypercube then \Cref{eq:code:boolA} is satisfied with probability $1$; else, \Cref{eq:code:boolA} is satisfied with probability at most $4(\log m + \log B)/|\F| \le \tfrac{2}{5}\eps$.

	\pparagraph{Unit measure of a multilinear $A$.} Assume that $A$ was Boolean on its hypercube (otherwise, see previous case). The argument is similar to the previous one, but simpler, with $\exptwo$ (\Cref{eq:exp2def}) decoding the weight numerators from the bits stored in $A$.

	Suppose $A$ encodes probabilities that sum to 1; that is, letting $\alpha_{j}^\dagger = \sum_{\beta} \exptwo(\beta)  A(j,\beta)$, the assumption is that $\sum_j{\alpha_j^\dagger} = 2^B$. Then,
	\begin{equation*}
		\sum_{\substack{j \in \bool^{\log m} \\ \beta \in \bool^{\log B}}} A(j,\beta) \cdot \exptwo(\beta) = \sum_{j \in \bool^{\log m}} \alpha^\dagger_j = 2^B.
	\end{equation*}
	The summand on the left hand side above has individual degree at most $2$. Thus, \Cref{fact:sumcheck} applied to $\SumCheck(2^B, \log m + \log B, 2)$ implies that an honest Prover convinces the Verifier with probability 1.

	For soundness, suppose $A$ was Boolean on the hypercube but did \emph{not} have unit measure, $\sum_{j,\beta}A(j,\beta)  \exptwo(\beta) \neq 2^B$ over $\Z$. We claim the inequality holds over $\F$ as well; indeed, by Booleanity of $A$, the following holds over $\Z$,
	\begin{equation*}
		0 \leq \sum_{\substack{j \in \bool^{\log m} \\ \beta \in \bool^{\log B}}} A(j,\beta)  \exptwo(\beta) \leq 2^{\log m} \sum_{\beta \in \bool^{\log B}}\exptwo(\beta) = m (2^{B} - 1) < |\F|.
	\end{equation*}
	Therefore, \Cref{fact:sumcheck} implies that \Cref{eq:code:normA} holds with probability at most $2 (\log m + \log B) / |\F| \le \tfrac{1}{5}\eps$, again by $10(\log m + \log B)/|\F| \le \eps$ of \Cref{eq:code:delta}.

	\pparagraph{Non-multilinear $Z$ and $A$.} We now remove the assumption that $Z$ and $A$ are multilinear. Let $\tilde Z, \tilde A$ be the multilinear polynomials closest to $Z, A$, at relative Hamming distances $\delta_Z \defeq \dist(Z,\tilde Z)$ and $\delta_A \defeq \dist(A,\tilde A)$, and write $R \defeq \lceil(\log m+\log n)\max\{\ln(1/\eps)/\delta,  2/\eps\}\rceil$ for the number of multilinearity repetitions on $Z$ (the count for $A$ is analogous). We are in the soundness case, meaning that $(Z,A)$ is $\delta$-far from every codeword; therefore at least one of the next two situations holds:

	\emph{(i) $Z$ or $A$ is $\delta$-far from multilinear.} Say $\delta_Z \ge \delta$, so $Z$ is in particular $\delta$-far from multilinear. Since $6(\log m+\log n)/|\F| < \delta < 1/2$, \Cref{fact:fglss} applies at proximity $\delta$, and each of the $R$ multilinearity tests on $Z$ accepts with probability at most $1-\delta/(\log m+\log n)$; hence they all accept with probability at most
	\begin{equation*}
		\Bigl( 1-\tfrac{\delta}{\log m + \log n} \Bigr)^{R} \le e^{-R  \delta/(\log m + \log n)} \le e^{-\ln(1/\eps)} = \eps,
	\end{equation*}
	using $R \ge (\log m+\log n)\ln(1/\eps)/\delta$.

	\emph{(ii) Both $Z$ and $A$ are $\delta$-close to multilinear, but $(\tilde Z, \tilde A) \notin \Code$.} Here $\delta_Z, \delta_A < \delta$, yet $\tilde Z$ or $\tilde A$ is non-Boolean, or $\tilde A$ is not of unit measure---otherwise $(Z,A)$ would be $\delta$-close to the codeword $(\tilde Z,\tilde A)$, contrary to assumption. We treat the case that $\tilde Z$ is non-Boolean---the other two failures are handled identically by \Cref{eq:code:boolA,eq:code:normA}. By the multilinear case above applied to $\tilde Z$, together with the at-most-$\delta_Z$ chance that the marginally-uniform query into $Z$ for \Cref{eq:code:boolZ} lands where $Z \ne \tilde Z$, \Cref{eq:code:boolZ} accepts with probability at most $4(\log m+\log n)/|\F| + \delta_Z$; we split on whether the multilinearity test can see the defect.

	If $\delta_Z \le 6(\log m+\log n)/|\F|$, then $Z$ is below the detection threshold of \Cref{fact:fglss}, but \Cref{eq:code:boolZ} alone already suffices:
	\begin{equation*}
		\Pr[\text{accept}] \le \frac{4(\log m+\log n)}{|\F|} + \delta_Z \le \frac{10(\log m+\log n)}{|\F|} \le \eps,
	\end{equation*}
	by \Cref{eq:code:delta}. Otherwise $6(\log m+\log n)/|\F| < \delta_Z < \delta < 1/2$, so \Cref{fact:fglss} applies at proximity $\delta_Z$ and the multilinearity test on $Z$ accepts with probability at most $(1-\delta_Z/(\log m+\log n))^{R} \le e^{-2\delta_Z/\eps}$, using $R \ge 2(\log m+\log n)/\eps$. Acceptance requires both this test and \Cref{eq:code:boolZ} to pass, on independent randomness, so multiplying and using $4(\log m+\log n)/|\F| \le \tfrac{2}{5}\eps$,
	\begin{equation*}
		\Pr[\text{accept}] \le e^{-2\delta_Z/\eps}\bigl(\tfrac25 \eps + \delta_Z\bigr) = \tfrac25 \eps   e^{-2\delta_Z/\eps} + \delta_Z   e^{-2\delta_Z/\eps} \le \tfrac25 \eps + \frac{\eps}{2e} < \eps,
	\end{equation*}
	using $e^{-2\delta_Z/\eps} \le 1$ for the first term and $\max_{x \ge 0} x   e^{-2x/\eps} = \eps/(2e)$ (attained at $x = \eps/2$) for the second.

	In every situation the Verifier accepts with probability at most $\eps$, which is the soundness claim.
\end{proof}
\begin{remark}\label{remark:fingerprinting}
	The field size $|\F| \geq m2^B$ is needed only for the unit-measure check, where the sum-check on \Cref{eq:code:normA} verifies the integer identity $\sum_j \alpha_j^\dagger = 2^B$ and the partial sums, which can reach $m(2^B-1)$, must not wrap around in $\F$. It could be avoided in an IP, by fingerprinting \cite{Freivalds77}: have the Verifier issue a random prime $p$ of size $\poly(B,n,\log m)$ and check the identity modulo $p$. However, the field $\F$ is then determined only during the protocol's runtime, thereby losing the clean conceptual cut of having a distribution encoded as $(Z,A)$. We prefer the cleaner presentation. On a technical level, our full protocol (\Cref{alg:full-protocol}) pays an exponential field size regardless, so the saving would not help our overarching goal.
\end{remark}

\subsection{Verifying a succinctly-specified marginal}
We give two protocols. The first, $\VerifyMarginoid$ (\Cref{alg:marginal}), is the one our main construction (the IPCP of \Cref{sec:mip:protocol}) actually calls. It runs on multilinear oracles and verifies $\sum_j \weight^{A}(j)  \match^Z_{b,s}(j)$ over $\F$ for a given $b,s$. We call this a \emph{marginoid} rather than a marginal because its context is specified over $\F$: the variable descriptions $s_1,\dots,s_\ell \in \F^{\log n}$ and the truth values $b_1,\dots,b_\ell \in \F$ range over the field, so the pair $(s,b)$ need not describe any actual event $\{\omega_{s_k} = b_k\}_k$. When $s$ and $b$ are Boolean and $(Z,A)$ is a codeword the marginoid is a genuine marginal mass; in general it is only the marginal-shaped algebraic value that the sum-check manipulates. This is all that \Cref{sec:mip:protocol} needs.

For completeness, and because it may be of independent interest, we then package $\VerifyMarginoid$ into the self-contained primitive $\VerifyMarginal$ (\Cref{alg:marginal-pp}). It verifies a genuine marginal up to tolerance $\tau$ using \emph{self-correction}: recovering the value of the nearby codeword at a point by querying the given (possibly corrupted) oracle at correlated random points (\Cref{fact:selfcorrect}). It will not be used downstream in this paper, however.

\begin{algorithm}[h!]
	\caption{$\VerifyMarginoid_\F^{Z,A}(v,   s_1,b_1,\dots,s_\ell,b_\ell)$}\label{alg:marginal}
	\KwInput{Finite field $\F$, $v,b_1,\dots,b_\ell \in \F$, and $s_1,\dots,s_\ell \in \F^{\log n}$. Oracle access to $Z \colon \F^{\log m + \log n} \to \F$ and $A\colon\F^{\log m + \log B} \to \F$.}
	\BlankLine
	Run
	\begin{equation*}
		(\hat v, \hat j) \defeq \SumCheck_\F(\val \gets v,\rounds \gets \log m, \degree \gets \ell+1)
	\end{equation*}\;
	Expect $w \in \F$ to be sent from the Prover. Run
	\begin{equation*}
		(\hat w, \hat \beta) \defeq \SumCheck_\F(\val \gets w,\rounds \gets \log B, \degree \gets 2)
	\end{equation*}\;
	Accept if and only if
	\begin{align}
		\hat w & \qeq \exptwo(\hat \beta) \cdot A(\hat j, \hat \beta), \ \textrm{and }\label{eq:marginal:w} \\
		\hat v & \qeq w \cdot \match^{Z}_{b,s}(\hat j). \label{eq:marginal:v}
	\end{align}
\end{algorithm}

\begin{claim}\label{claim:marginal}
	Fix a finite field $\F$, $v,b_1,\dots,b_\ell \in \F$, and $s_1,\dots,s_\ell \in \F^{\log n}$, and suppose $\VerifyMarginoid$ is run with \emph{multilinear} oracles $Z \colon \F^{\log m + \log n} \to \F$ and $A\colon\F^{\log m + \log B} \to \F$. Recall the gadgets $\exptwo$ and $\match^{Z}_{b,s}(j)$ from \Cref{eq:exp2def,eq:satdef}.
	\begin{itemize}
		\item \emph{Completeness:} There exists an honest Prover such that if
		      \begin{equation}\label{eq:marginal:claim}
			      v
                  = \sum_{\substack{j \in \bool^{\log m}\\ \beta \in \bool^{\log B}}} \exptwo(\beta)   A(j,\beta)   \match^{Z}_{b,s}(j)
		      \end{equation}
		      then the Verifier accepts with probability 1.
		\item \emph{Soundness:} If \Cref{eq:marginal:claim} does not hold then regardless of the Prover's messages, the Verifier accepts with probability at most $((\ell+1) \log m + 2\log B)/|\F|$.
	\end{itemize}
	Furthermore, the Verifier runs in $O((\ell+1) \log m + \log B)$ field operations, makes $\ell$ queries to $Z$ and one query to $A$. The query to $A$ is marginally uniform on $\F^{\log m + \log B}$, and each query to $Z$ is marginally uniform in its first $\log m$ coordinates.
\end{claim}
\begin{proof}
    We use the shorthand $\weight^A(\hat j) = \sum_\beta \exptwo(\beta)A(\hat j, \beta)$ throughout this proof (see \Cref{eq:Wdef})
	At a high level, the first $\SumCheck$ reduces the claimed value $v$ to the evaluation $\weight^{A}(\hat j)\match^{Z}_{b,s}(\hat j)$ at a random $\hat j$, and the second reduces the Prover's alleged value of $\weight^{A}(\hat j)$ to a single query to $A$. We first record the degrees of the two summands, then prove completeness and soundness.

	\pparagraph{Degrees of the two summands.} Denote
	\begin{equation*}
		\varphi(j) \defeq \weight^{A}(j)  \match^{Z}_{b,s}(j) \qquad \text{and} \qquad \psi_{\hat j}(\beta) \defeq \exptwo(\beta)   A(\hat j,\beta).
	\end{equation*}
	Because $Z$ is multilinear, for each fixed $s_k$ the map $j \mapsto Z(j,s_k)$ is multilinear. Therefore, each factor $b_k Z(j,s_k)+(1-b_k)(1-Z(j,s_k))$ has individual degree $1$ in each coordinate of $j$, and their product $\match^{Z}_{b,s}(j)$ has individual degree at most $\ell$. Since $\weight^{A}$ is multilinear in $j$ (\Cref{eq:Wdef}), the product $\varphi(\bullet) = \weight^{A}(\bullet) \cdot \match^{Z}_{b,s}(\bullet)$ has individual degree at most $\ell+1$ in each of its $\log m$ variables. As for $\psi_{\hat j}$: $\exptwo(\beta)=\prod_{k=1}^{\log B}(1+\beta_k(2^{2^{k-1}}-1))$ is multilinear in $\beta$, and so is $A(\hat j,\beta)$ for fixed $\hat j$, so $\psi_{\hat j}$ has individual degree at most $2$ in each of its $\log B$ variables. These are the degree parameters passed to the two invocations of \Cref{alg:sumcheck}.

	\pparagraph{Completeness.} Assume \Cref{eq:marginal:claim}; that is, $v=\sum_{j\in\bool^{\log m}}\varphi(j)$. Run the honest \Cref{fact:sumcheck} Prover for $\varphi$ in the first $\SumCheck$ (its individual degree is at most $\ell+1$, as required); by completeness of \Cref{fact:sumcheck}, $\hat v=\varphi(\hat j)=\weight^{A}(\hat j)\match^{Z}_{b,s}(\hat j)$ with probability $1$. Next, let the Prover send $w \defeq \weight^{A}(\hat j)=\sum_{\beta\in\bool^{\log B}}\psi_{\hat j}(\beta)$ and answer honestly as per \Cref{fact:sumcheck} for $\psi_{\hat j}$ in the second $\SumCheck$; then $\hat w=\psi_{\hat j}(\hat\beta)=\exptwo(\hat\beta)A(\hat j,\hat\beta)$ with probability $1$, so \Cref{eq:marginal:w} holds. Finally, $\hat v = \weight^{A}(\hat j)\match^{Z}_{b,s}(\hat j)=w\cdot\match^{Z}_{b,s}(\hat j)$, so \Cref{eq:marginal:v} holds as well, and the Verifier accepts with probability $1$.

	\pparagraph{Soundness.} Suppose \Cref{eq:marginal:claim} fails; that is, $v \ne \sum_{j\in\bool^{\log m}}\varphi(j)$. Fix any Prover, and distinguish two events according to the value $w$ it sends: either the Verifier accepts and $w=\weight^{A}(\hat j)$, or the Verifier accepts and $w\ne \weight^{A}(\hat j)$.

	Suppose the Verifier accepts and $w=\weight^{A}(\hat j)$. Then \Cref{eq:marginal:v} is $\hat v=w  \match^{Z}_{b,s}(\hat j)=\weight^{A}(\hat j)\match^{Z}_{b,s}(\hat j)=\varphi(\hat j)$. Since $v\ne\sum_j\varphi(j)$, by soundness of the first $\SumCheck$, the Verifier accepts with probability at most $(\ell+1)\log m/|\F|$.

	Now suppose $w\ne \weight^{A}(\hat j)$. Then $w\ne\sum_\beta\psi_{\hat j}(\beta)$, so by soundness of the second $\SumCheck$, \Cref{eq:marginal:w} holds with probability at most $2\log B/|\F|$.

    By a union bound, the Verifier accepts with probability at most
	\begin{equation*}
		\frac{(\ell+1)\log m + 2\log B}{|\F|}.
	\end{equation*}

	\pparagraph{Complexity and smoothness.}\footnote{Smoothness \cite{KT00}, \ie\ that the Verifier's oracle queries are (marginally) uniformly distributed over the proof, is a well-studied property of PCPs \cite{Par21,AC25,BHPT24}; we rely on it repeatedly, as it drives both the self-correction of \Cref{fact:selfcorrect} and the passage from the oracle to a nearby codeword in the soundness proof of \Cref{thm:explicit-iop}.} Both follow from \Cref{fact:sumcheck}. The two $\SumCheck$s cost $O((\ell+1)\log m)$ and $O(\log B)$ field operations, and their outputs $\hat j, \hat\beta$ are marginally uniform. To evaluate \Cref{eq:marginal:v}, the Verifier computes $\match^{Z}_{b,s}(\hat j)$ directly via the $\ell$ queries $Z(\hat j,s_k)$; to evaluate \Cref{eq:marginal:w}, it computes $\exptwo(\hat\beta)$ directly and makes the single query $A(\hat j,\hat\beta)$.
\end{proof}

We now package $\VerifyMarginoid$ as the standalone primitive $\VerifyMarginal$. Given oracles $(Z,A)$ close to a codeword for a distribution $\mu$, it certifies that an alleged value $v$ agrees, within a tolerance $\tau$, with the true marginal mass of a context---the denominator-cleared probability $2^B\mu[s_1,b_1,\dots,s_\ell,b_\ell]$. The Prover commits to that mass, the Verifier checks it against $v$ by an integer comparison exactly as the tolerance test of $\VerifyConsistency$ (\Cref{alg:full-protocol}), and the commitment is certified by running $\VerifyMarginoid$ on the self-corrected oracle.

\begin{fact}[\cite{BlumLR90}]\label{fact:selfcorrect}
	Fix a field $\F$, $k \leq |\F| - 2$, and $f \colon \F^k \to \F$.
	Consider the following procedure $\SelfCorrect_\F^f(x)$ for input $x \in \F^k$:
	\begin{enumerate}
		\item Sample $r \in \F^k$ uniformly at random.
		\item Query $f(x + r), f(x + 2r), \dots, f(x + (k+1) r)$.
		\item Compute $g \colon \F \to \F$ the unique polynomial of degree at most $k$ such that $g(i) = f(x + ir)$ for all $i \in [k+1]$.
		\item Output $g(0)$.
	\end{enumerate}
	For any $\delta \in (0,1)$, if $f$ is $\delta$-close to a multilinear $\tilde f \colon \F^k \to \F$, then $\SelfCorrect_\F^f(x)$ outputs $\tilde f(x)$ with probability at least $1 - (k+1) \delta$. $\SelfCorrect$ makes $k+1$ marginally uniform queries to $f$, and runs in time $O(k^2)$ field operations.\footnote{More precisely, given precomputed interpolation weights, the running time is $O(k)$ scalar field operations plus $O(k)$ vector additions.}
\end{fact}

\begin{algorithm}[h!]
	\caption{$\VerifyMarginal^{Z,A}(v,   s_1,b_1,\dots,s_\ell,b_\ell,   \tau)$}\label{alg:marginal-pp}
	\KwInput{Finite field $\F$, an alleged value $v \in \F$, a context $s_1,\dots,s_\ell \in [n]$ with bits $b_1,\dots,b_\ell \in \bool$, and a tolerance $\tau = \tau^\dagger/2^B$ given to precision $B$ by an integer numerator $\tau^\dagger \in \{0,\dots,2^B\}$. Oracle access to $Z \colon \F^{\log m + \log n} \to \F$ and $A\colon\F^{\log m + \log B} \to \F$ that are $\delta$-close to a codeword.}
	\BlankLine
	Expect $v^\ast \in \F$ from the Prover. Treating $v^\ast$ and $v$ as integers in $\{0,\dots,|\F|-1\}$, reject unless,
	\begin{equation}\label{eq:marginal-pp:threshold}
		\bigl|   v^\ast - v   \bigr| < \tau^\dagger \quad \text{over $\Z$}.
	\end{equation}\;
	Run $\VerifyMarginoid_\F^{Z,A}(v^\ast,   s_1,b_1,\dots,s_\ell,b_\ell)$, replacing each query $Z(\hat j, s_k)$ with $\SelfCorrect_\F^{Z}(\hat j, s_k)$ (\Cref{fact:selfcorrect}). Accept if and only if it accepts.
\end{algorithm}

In fact, only the $s_k$ argument to $Z$ needs to be self-corrected, but we self-correct the entire query $Z(\hat j, s_k)$ for simplicity of notation.

\begin{corollary}[Marginal certification within a tolerance]\label{cor:marginal}
	Fix $m,n,B \in \N$ and a prime field $\F$ with $\log m + \log n \le |\F| - 2$ and $2^B < |\F|$. Let $\pi = (Z,A)$ be functions that are $\delta$-close, in relative Hamming distance, to a codeword $(\widetilde Z,\widetilde A) \in \Code_\F(n,m,B)$ encoding a distribution $\mu$ over $\bool^n$. Fix a context $s_1,\dots,s_\ell \in [n]$ with bits $b_1,\dots,b_\ell \in \bool$, an alleged value $v \in \F$, and a tolerance $\tau = \tau^\dagger/2^B$ with integer numerator $\tau^\dagger \in \{0,\dots,2^B\}$. Then $\VerifyMarginal^{Z,A}(v, s_1,b_1,\dots,s_\ell,b_\ell, \tau)$ satisfies:
	\begin{itemize}
		\item \emph{Completeness:} if $\pi = (\widetilde Z, \widetilde A)$ is exactly the codeword and $|\mu[s_1,b_1,\dots,s_\ell,b_\ell] - v/2^B| < \tau$, there is an honest Prover with which the Verifier accepts with probability $1$.
		\item \emph{Soundness:} if $|\mu[s_1,b_1,\dots,s_\ell,b_\ell] - v/2^B| \ge \tau$, then regardless of the Prover's messages the Verifier accepts with probability
		      \begin{equation*}
			      O\!\left( \frac{\ell\log m + \log B}{|\F|} \;+\; \ell  (\log m + \log n)  \delta \right).
		      \end{equation*}
		\item \emph{Complexity:} the Verifier makes $O\bigl(\ell  \log m + \ell \log n\bigr)$ marginally-uniform queries to $Z$ and one to $A$, exchanges $O(\ell \log m + \log B)$ field elements with the Prover, and runs in $O\bigl(\ell  (\log m + \log n)^2 + \log B\bigr)$ field operations, plus $O(\log |\F|)$ bit operations for the integer comparison.
	\end{itemize}
\end{corollary}
\begin{proof}
	Write $\mu^\dagger \defeq 2^B \mu[s_1,b_1,\dots,s_\ell,b_\ell]$. Note that, for the $\Code$-encoding $(\widetilde Z, \widetilde A)$ of $\mu$, the value certified by \Cref{claim:marginal} (the right-hand side of \Cref{eq:marginal:claim}) is
    \begin{equation*}
        \sum_j \weight^{A}(j)  \match^{\widetilde Z}_{b,s}(j) = \sum_{j:\forall k\ z_{j,s_k} = b_k} \alpha^\dagger_j = \mu^\dagger,
    \end{equation*}
    indeed $\weight^{A}(j) = \alpha^\dagger_j$ is the integer weight of $z_j$ (\Cref{eq:Wdef}), and $\match^{\widetilde Z}_{b,s}(j)$ is the Boolean indicator that $z_{j,s_k} = b_k$ for all $k$ (\Cref{def:code,eq:satdef}). Because $0 \le \mu^\dagger \le 2^B < |\F|$, the integer comparison \Cref{eq:marginal-pp:threshold} against a Prover message claiming $\mu^\dagger$ is unambiguous.

	\emph{Completeness.} The honest Prover sends $v^\ast \defeq \mu^\dagger$. Since $|\mu^\dagger - v| < \tau^\dagger$, \Cref{eq:marginal-pp:threshold} holds. On the exact codeword every $\SelfCorrect$ returns $\widetilde Z(\hat j, s_k)$ with probability $1$, so the second step is $\VerifyMarginoid$ on $(\widetilde Z, \widetilde A)$ with claimed value $v^\ast = \mu^\dagger$; by \Cref{claim:marginal} it accepts with probability $1$.

	\emph{Soundness.} Suppose $|\mu^\dagger - v| \ge \tau^\dagger$ and fix any Prover sending some $v^\ast$ as its message. If the Verifier accepts then \Cref{eq:marginal-pp:threshold} holds, \ie, $|v^\ast - v| < \tau^\dagger$. Then $v^\ast \ne \mu^\dagger$ over $\F$. The second step is therefore $\VerifyMarginoid$ run with claimed value $v^\ast \ne \mu^\dagger$ and each $Z$-query self-corrected. By \Cref{fact:selfcorrect} (with $k = \log m + \log n \le |\F|-2$), each of the $\ell$ self-corrected queries returns $\widetilde Z(\hat j, s_k)$ except with probability $(\log m + \log n + 1)\delta$, and the single $A$-query is marginally uniform on $\F^{\log m + \log B}$ (\Cref{claim:marginal}), agreeing with $\widetilde A$ except with probability $\delta$. Conditioned on all of these agreements, the run is then exactly $\VerifyMarginoid$ against the multilinear codeword $(\widetilde Z, \widetilde A)$, whose true value is $\mu^\dagger$. Since $v^\ast \ne \mu^\dagger$, \Cref{eq:marginal:claim} fails and \Cref{claim:marginal} bounds acceptance by $((\ell+1)\log m + 2\log B)/|\F|$. Conclude with a union bound.

	\emph{Complexity.} Each $\SelfCorrect$ makes $\log m + \log n + 1$ marginally-uniform queries and $O((\log m + \log n)^2)$ field operations (\Cref{fact:selfcorrect}), and there are $\ell$ of them; the integer comparison \Cref{eq:marginal-pp:threshold} costs $O(\log |\F|)$ bit operations on the lifted representatives, and the remaining cost, including the $O(\ell \log m + \log B)$ field elements of Prover communication, is that of \Cref{claim:marginal}.
\end{proof}

\section{Verifying consistency of predictive models}\label{sec:succinct}

The main contribution of this section is an \emph{explicit} polynomial-time Interactive PCP ($\IPCP$) for $\MODELCONSIST$: a protocol in which the Verifier is given an (alleged) $\Code$-encoding (\Cref{sec:codecheck}) of a consistent distribution, and uses it to certify consistency of a given predictive model.

We note the consequence we do \emph{not} pursue: \Cref{cor:mip} composed with $\MIP = \NEXP$ \cite{BabaiFL91} gives that $\MODELCONSIST$ admits a $\MIP$, namely, polynomial-time Verifier interacting with two provers.\footnote{Equivalently, a polynomial time Verifier querying an exponential length PCP. We stick with the $\MIP$ strawman, but our informal argument would apply to a PCP just the same.} One could then hope to augment the training of $(P, Q)$ with the training of two Provers $P_1, P_2$ towards a predictive model that proves its own consistency (see \Cref{rem:spmip}). The practical obstacle to such a programme is not (only) the cost of running $P_1, P_2$: they may be worst-case exponential-time, but perhaps more efficient on average and may even be trained for average success by Reinforcement Learning from Verifier Feedback (RLVF) or Transcript Learning (TL) \cite{AmitGPR25}. However, before one trains an honest Prover by RLVF (resp. TL), one needs an implementation of the Verifier (resp. honest Prover) to learn against: code that, on input the witnessing distribution $\mu$ underlying $(P, Q)$'s consistency, produces the messages an honest $P_1, P_2$ would send. Attempting to follow the rationale of \Cref{cor:mip}, this code is the composition of two generic constructions. The first is the Cook--Levin \cite{Cook71,Levin73} reduction of $(P, Q, \tau, V_\NEXP)$ into a succinct CNF formula, which discards the convex-geometric semantics of the Carath\'eodory witness.
The second is the BFLS honest Prover strategy \cite{BabaiFLS91} which runs on top of the succinct CNF.

The IPCP we construct in this section is designed towards clearing this obstacle by explicitly describing the Verifier and honest Prover. Namely, the proof oracle is the $\Code$-encoding of the witnessing distribution, and the $\SumCheck$s that reduce the inconsistency inequality to queries on that oracle admit an arguably self-contained, implementable specification (\Cref{alg:full-protocol}). We reflect in \Cref{sec:conclusions} on what was natural about the resulting protocol: how it factors into a model-specific ``witness'' core and a model-agnostic ``proving'' overlay, which we hope opens a route to training self-proving predictive models in the spirit of \cite{AmitGPR25}.

\begin{definition}[Interactive PCP \cite{KalaiR08}]\label{def:ipcp}
	An \emph{Interactive PCP} ($\IPCP$) is an interactive proof in which the Verifier additionally has oracle access to a proof string $\pi$ fixed before the interaction begins (the PCP). It is measured by the \emph{Verifier runtime}; the \emph{PCP length}, the number of symbols of $\pi$; the \emph{query complexity}, the number of point queries the Verifier makes to $\pi$; and the \emph{communication} exchanged with the Prover.
\end{definition}

Importantly, the oracle cannot adapt to the interaction, whereas the Prover can; this is what separates an $\IPCP$ from an interactive proof, and since a proof oracle is equivalent to non-communicating provers \cite{FortnowRS94}, an $\IPCP$ is a special case of a multi-prover proof. As elsewhere, we specify only the Verifier, and describe the honest Prover and proof string when proving completeness.

\subsection{Setup: from a witnessing distribution to the proof oracle}\label{sec:mip:witness}

Fix a predictive model $(P,Q)$ with variable description length $d$, context length $\ell$, and precision $B$. We write $n \coloneqq 2^d$ for the number of (implicit) Boolean variables, and recall that $P$'s probability denominator is $2^{B}$.
For readability, we omit the subscript from $\D_{(P,Q)}$, writing simply $\D$ for the model inconsistency.

\paragraph{Queries.}
Recall from \Cref{sec:def:model} that a \emph{query} is a string $q = (s_1, b_1, \dots, s_\ell, b_\ell, t) \in \Qcal$, with $\Qcal \coloneqq \bool^{\ell(d+1)+d}$ the \emph{query universe}, $s_k \in \bool^d$ the description of the $k$th context variable, $b_k \in \bool$ the corresponding context bit, and $t \in \bool^d$ the target-variable description. For example, $q=(s_1, 0, s_5, 1, t)$ is asking for the probability that the variable $\omega_t$ is true conditioned on $\omega_{s_1} = 0 $ and $\omega_{s_5} = 1$. We abbreviate $\norm{Q} \coloneqq \norm{Q}_1 = \sum_{q \in \Qcal} Q(q)$, the number of (confidence-weighted) claims implicit in $(P,Q)$; note that $\norm{Q}_1 \le |\Qcal | \cdot 2^B \in 2^{O(\ell d + B)}$.

We now describe what the Prover of a consistent model must do to come up with the proof oracle: find a distribution $\mu^\star$ consistent with the model $(P,Q)$, \emph{sparsify} it, and \emph{encode} the sparse result.

\paragraph{Step 1: find any witnessing distribution.}
Suppose $(P,Q)$ is $(\tau - \epsgap)$-consistent for some gap $\epsgap > 0$: by definition, there is \emph{some} distribution $\mu^\star$ over $\bool^n$ with $\Incof(P,Q;\mu^\star) \le \tau - \epsgap$. Nothing more is assumed of $\mu^\star$; its support may well be all of $\bool^n$.

\paragraph{Step 2: sparsify.}
By the Carath\'eodory step of \Cref{sec:NP} applied to the implicit CPC $\PP_{P,Q}$ (\Cref{def:implicit-cpc,claim:implicit-cpc}), in the low-precision explicit form of \Cref{prop:near-linear-witnesses} (gap $\epsgap$), the witness $\mu^\star$ can be traded for a sparse distribution $\mu_{z,\alpha}$ with
\begin{itemize}
	\item support vectors $z_1,\dots,z_{m} \in \bool^{n}$, with $m = \norm{Q}_1 + 1$, \footnote{\Cref{prop:near-linear-witnesses} gives a  support of at most this size, and \Cref{def:code} permits the zero-weight padding. We assume, here and throughout, that $\norm{Q}_1+1$ is a power of two, so that $\log m$ is an integer as \Cref{def:code} requires; the decoupled variant described in \Cref{rem:variable-length} removes this assumption.} and
	\item weights $\alpha_j = \alpha^\dagger_j / 2^{B'}$ for integers $\alpha^\dagger_j \in \{0,\dots,2^{B'}-1\}$ with $\sum_{j=1}^{m} \alpha^\dagger_j = 2^{B'}$, where the denominator's bit-length $B'$ is the certificate precision $B_{\epsgap}$ of \Cref{prop:near-linear-witnesses} at this $m$, rounded up to a power of two as \Cref{def:code} requires, so that $B' \in O(\log \norm{Q} + \log(1/\epsgap)) \subseteq O(\ell d + B + \log(1/\epsgap))$,\footnote{\Cref{prop:near-linear-witnesses} also permits $\alpha^\dagger_j = 2^{B_{\epsgap}}$, a point mass, which the range above excludes. Such a witness is brought into range by duplicating its support vector.}
	\item such that $\D_{\PP_{P,Q}}(\mu_{z,\alpha}) \le \tau$.
\end{itemize}
The sparsification costs the additive $\epsgap$ slack in completeness. Throughout the rest of this section we work exclusively with the low-precision witness; \Cref{rem:exact-witness} explains why the exact witness of \Cref{lem:rationalopt} is not used here.

\paragraph{Step 3: encode.}
The Prover encodes $\mu$ using $\Code_\F(n, m, B')$, for sufficiently large field $\F$. Namely, the proof oracle is the codeword $\pi = (Z, A)$, where as in \Cref{def:code}, $Z$ is the multilinear extension of the support-vector table $(j,i) \mapsto z_{j,i}$ and $A$ that of the weight-bit table $(j,\beta) \mapsto (\alpha^\dagger_j)_\beta$.

From this point onward $\pi = (Z, A)$ denotes an \emph{alleged} $\Code_\F(n, m, B')$-codeword---a pair of functions over $\F$ of the types prescribed by \Cref{def:code}.

The predictive model $(P,Q)$ is given as a pair of circuits taking an input query $q$ given as a bit-string. The Verifier in our protocol will eventually want to compute the outputs of these circuits on $\F$-valued queroids $\hat q$; this is done by \emph{arithmetizing} each circuit, \ie, deriving polynomials $\hat P, \hat Q$ by replacing Boolean gates by their arithmetic form over $\F$: negation is replaced by $a \mapsto 1-a$, conjunction by $(a,b) \mapsto ab$, and disjunction by $(a,b) \mapsto a+b-ab$.

\begin{fact}[Naive]\label{fact:arith}
	Let $C$ be a Boolean circuit with $n$ inputs and $S$ many gates, computing a function $\bool^n \to \{0,\dots,M-1\}$ with $M < |\F|$. The polynomial $\hat C \in \F[X_1,\dots,X_n]$ is obtained by replacing each Boolean gate with its arithmetic form over $\F$. Then, for any Boolean input $x \in \bool^n$, $\hat C(x) = C(x)$, and for any $\hat x \in \F^n$, $\hat C(\hat x)$ is computable in $O(S)$ field operations.
\end{fact}
We let $\hat P, \hat Q$ denote the polynomials obtained by this naive arithmetization, and let $\Delta$ denote an upper bound on their degrees (a scalar, not to be confused with the simplex $\Delta X$, which always takes a set argument). The protocol presented in this section costs linearly in $\Delta$, which is itself exponential in the circuits' depth in the worst case; \Cref{sec:delegation} then shows how doubly-efficient interactive proofs \cite{GoldwasserKR15,Thaler13} replace the dependency on $\Delta$ by one on the depth, and improve the dependency on $S$.

\paragraph{Parameter glossary.} Denominators are written as explicit powers of two. We collect, for reference, all symbols used throughout the rest of \Cref{sec:succinct}:
\begin{itemize}
	\item $d$: variable description length, so that $n \coloneqq 2^d$ is the number of (implicit) Boolean variables.
	\item $\ell$: context length, the number of variables conditioned upon in each query.
	\item $B$: precision of the model $(P,Q)$, formally number of output gates of $P$ and $Q$.
	\item $\tau = \tau^\dagger / 2^{B}$: the consistency tolerance from $\MODELCONSIST$, given to precision $B$ with integer numerator $\tau^\dagger$.
	\item $\epsgap > 0$: the slack parameter from the gap variant of $\MODELCONSIST$.
	\item $\Qcal \coloneqq \bool^{\ell(d+1)+d}$: the query universe; each $q \in \Qcal$ is parsed as $q=(s_1, b_1, \dots, s_\ell, b_\ell, t)$.
	\item $\norm{Q}_1$: the number of (multiplicity-weighted) claims implicit in $(P,Q)$.\footnote{We may omit the subscript $1$ for readability}
	\item $m = \norm{Q}_1 + 1$: support size of the sparsified witness.
	\item $\PP_{P,Q}$: the implicit CPC of $(P,Q)$, as defined in \Cref{def:implicit-cpc}.
	\item $B'$: bit-length of the weights of the sparsified witnessing distribution.
	\item $\F$: the finite field over which all arithmetic and oracles are evaluated; specified in \Cref{sec:mip:protocol}.
	\item $\pi = (Z, A)$: the proof oracle, \ie, the PCP component of the IPCP. With $Z\colon \F^{\log m} \times \F^d \to \F$ and $A \colon \F^{\log m} \times \F^{\log B'} \to \F$.
\end{itemize}

\begin{remark}[Support size $m$ and the length of the proof $\pi$]\label{rem:variable-length}
	The honest oracle has length $m = \norm{Q}_1 + 1$, the length parameter of the codeword $\Code_\F(n, m, B')$, fixed the moment $\pi$ is given. The Verifier cannot compute $\norm{Q}_1 = \sum_q Q(q)$ on its own (it is a sum over all of $\Qcal$), but it can read $m$ off $\pi$'s type signature, recover $\norm{Q}_1 = m - 1$, and reject out of hand any oracle whose declared length exceeds the a-priori cap $m \le 2^B|\Qcal| + 1$. The honest length thus depends on the particular input, through $\norm{Q}_1$, not only on the input length. Instance-dependent proof length is common in the literature: \eg, the $\IPCP$ oracle of \cite{KalaiR08} is polynomial in the witness size, and the interactive oracle proofs of \cite{RonZewiR24} are made to approach the witness length. The instance-dependence may nonetheless be undesirable in our context, \eg, as it complicates preparing a uniform-length ``batch'' for feeding into a GPU; one regains uniformity by padding $\pi$ up to the maximal length $O(2^{\ell d})$ with further zero-weight dummy support vectors. We prefer the variable length for its granularity, and for the conceptual point it portrays: a more confident model, or one with more implicit claims (a larger $\norm{Q}_1$ in either case), requires a longer proof.

	We also note that $m = \norm{Q}_1 + 1$ tracks the model's total confidence rather than the true support size of the witness, which a resourceful Prover might make smaller still. A modest modification would enable this efficiency: Instead of identifying $\norm{Q}_1$ with $m-1$, the Verifier reads only the support size $m$ off $\pi$'s type signature and \emph{asks} the Prover for $\norm{Q}_1$ as a separate field element; it certifies this alleged value by the $\norm{Q}_1$-$\SumCheck$ (now seeded at the alleged value), uses the certified $\norm{Q}_1$ in the threshold \eqref{eq:full:threshold} in place of $m-1$, and checks over $\Z$ that $m \le 2(\norm{Q}_1 + 1)$ (the Carath\'eodory cap, doubled to absorb padding the support up to a power of two). The honest Prover could then encode the witness at its true support size (at most $\norm{Q}_1 + 1$, and often far smaller), and the proof would shrink accordingly. We keep the simpler protocol above, in which the oracle's length has $m=\norm{Q}_1 + 1$.
\end{remark}

\subsection{The full protocol and analysis}\label{sec:mip:protocol}

We now formalize the protocol sketched in \Cref{sec:mip:overview}.

\paragraph{Queroids.}
Recall that a query $q \in \Qcal$ is parsed into its components $q = (s_1, b_1, \dots, s_\ell, b_\ell, t)$, with $s_k \in \bool^d$, $b_k \in \bool$, and $t \in \bool^d$. The $\SumCheck$s below run over $\F$, so they manipulate the $\F$-valued analogue of a query:
\begin{equation}\label{eq:q-extension}
	\hat{q} = (\hat s_1, \hat b_1, \dots, \hat s_\ell, \hat{b}_\ell, \hat{t}) \quad\text{where}\quad \hat{s}_1,\dots, \hat s_\ell , \hat t \in \F^{d} \text{ and } \hat b_1,\dots, \hat b_\ell \in \F .
\end{equation}
We call such an $\hat q$ a \emph{queroid} (in the same spirit as the marginoid of \Cref{sec:codecheck}): it has the shape of a query, but its elements are in $\F$, so it may not describe any actual query. A genuine query asks for the mass of an event such as $\omega_{s_1} = b_1$; once the description $\hat s_1$ and the ``bit'' $\hat b_1$ are non-Boolean, no such event exists.

\paragraph{Gadgets.}
The protocol adapts three of the multilinear gadgets of \Cref{sec:codecheck}. The first is $\match^{Z}_{b,s}$ of \Cref{eq:satdef}, modified for covenience:
\begin{equation}\label{eq:match}
	\match^Z(\hat j, \hat q) = \match^Z(\hat j, \hat s_1, \hat b_1, \dots, \hat s_\ell, \hat{b_\ell}, \hat{t}) \coloneqq \prod_{k=1}^\ell \Bigl( \hat b_k\hat Z(\hat j, \hat s_k) + \left(1-\hat b_k\right) \left(1-\hat Z(\hat j, \hat s_k)\right) \Bigr).
\end{equation}
The only difference is that we promoted $b,s$ from subscript to an argument $q=(b,s,t)$, but the target variable description $t$ is ignored.

The second is $\exptwo$ of \Cref{eq:exp2def}, here in its $\log B'$-variate form $\exptwo(\hat \beta) = \prod_{k=1}^{\log B'}\bigl(1 + \hat \beta_k (2^{2^{k-1}} - 1)\bigr)$, which on a Boolean $\beta \in \bool^{\log B'}$ equals $2$ raised to the integer value of $\beta$ (as a field element). The Verifier evaluates $\exptwo(\hat\beta)$ in $O(\log B')$ field operations at any $\hat \beta \in \F^{\log B'}$.\footnote{The doubly-exponential constants $2^{2^{k-1}}$ for $k \le \log B'$ are precomputed once, each by $k-1$ repeated squarings, for $O((\log B')^2)$ field operations in total; each subsequent evaluation of $\exptwo$ then costs only the $O(\log B')$ multiplications of \Cref{eq:exp2def}.}

The third is the weight numerator $\weight^{A}(j) = \sum_{\beta \in \bool^{\log B'}} \exptwo(\beta) A(j,\beta)$ of \Cref{eq:Wdef}, here at precision $B'$: on the codeword $\pi$ it is the integer weight $\alpha^\dagger_j$ of the $j$-th support vector, so that $\weight^{A}(j)/2^{B'}$ is the probability mass placed on that vector.

\paragraph{Marginoids.}
The key quantities that connect the oracle $\pi=(Z,A)$ to the inconsistency are the two \emph{(denominator-cleared) marginoids}
\begin{align}\label{eq:masses}
	\Cmasshat & \defeq \sum_{j \in \bool^{\log m}} \weight^{A}(j)\match^Z(j,\hat q)
	\\
	\Qmasshat & \defeq \sum_{j \in \bool^{\log m}} \weight^{A}(j)\match^Z(j, \hat q)Z(j, \hat t),\nonumber
\end{align}
whose context $(\hat s, \hat b)$ and target $\hat t$ are the components of $\hat q$, parsed as in \Cref{eq:q-extension}.\footnote{In $\Qmass$ we spell the query event as $(s = b) \cap (t = 1)$; the $\mu[\cdot]$ notation of the preliminaries would abbreviate it as the context event intersected with the target, but spelling out that the target variable is required to be \emph{true} is clearer here.} They are \emph{marginoids} because $\hat q$ is over $\F$: on a Boolean query $q$ they are the genuine (denominator-cleared) marginals $2^{B'}\mu[s = b]$ and $2^{B'}\mu[(s = b) \cap (t = 1)]$. However, for a non-Boolean queroid $\hat q$ they are only marginal-shaped $\F$-vectors.

The quantity the protocol ultimately tests is the \emph{denominator-cleared inconsistency} $\Incdag$,
\begin{equation}\label{eq:Incdag}
	\Incdag \defeq \sum_{q \in \Qcal} Q(q)\bigl(2^B \Qmass - P(q) \Cmass\bigr)^2.
\end{equation}
The dagger follows the convention of $\tau^\dagger$ and $\alpha^\dagger_j$: $\Incdag$ is the integer obtained from the squared inconsistency by clearing all denominators (with the square superscript omitted for readability). As we soon prove, when $\pi$ is the $\Code$-codeword witnessing distribution $\mu$, $\Incdag$ has the closed form
\begin{equation}\label{eq:I-Dckt}
	\Incdag = 2^{2(B+B')} \cdot \norm{Q}_1 \cdot \Incof(P, Q; \mu)^2,
\end{equation}
so that checking $\Incdag \le \tau^22^{2(B+B')} \norm{Q}_1$ is equivalent to $\Incof(P, Q; \mu) \le \tau$.

\begin{algorithm}[h!]
	\caption{$\VerifyConsistency^{\pi}(P,Q;\tau,\epssound,\epsgap)$}\label{alg:full-protocol}
	\KwInput{Predictive model $(P,Q)$ with description length $d$, context length $\ell$ and precision $B$, given as arithmetic circuits of degree at most $\Delta$ (\Cref{fact:arith}); tolerance $\tau = \tau^\dagger / 2^{B} \ge 0$; soundness error $\epssound > 0$; and soundness gap $\epsgap > 0$.}
	\KwOracle{$\pi=(Z,A)$, with $Z \colon \F^{\log m} \times \F^{d} \to \F$ and $A \colon \F^{\log m} \times \F^{\log B'} \to \F$.}
	\BlankLine
	Read the dimensions $m, B' \in \N$ off the signature of the oracle $\pi$, and reject unless $m \le 2^B|\Qcal| + 1$, $2^{2(B+B')}(m-1) < |\F|$, and $m(2^{B'}-1) < |\F|$.\;
	Run $\VerifyEncoding_\F^{\pi}(\delta \gets \nicefrac{\epssound}{8\ell+12}, \eps \gets\epssound/4)$, and reject if it rejects.\;
	Expect $\vInc \in \F$ from the Prover, and reject unless, as integers,
	\begin{equation}\label{eq:full:threshold}
		\vInc \le (\tau^\dagger)^2 2^{2B'} \cdot (m - 1)
		\quad\text{over } \Z.
	\end{equation}\;
	Run
	\begin{align*}
		(\hvInc, \hat q)    & \defeq \SumCheck_\F(\val \gets \vInc,\ \rounds \gets \ell(d+1)+d,\ \degree \gets 3\Delta+2), \\
		(\hvNormQ, \hat q') & \defeq \SumCheck_\F(\val \gets m-1,\ \rounds \gets \ell(d+1)+d,\ \degree \gets \Delta),
	\end{align*}
	and parse $\hat q = (\hat s_1, \hat b_1, \dots, \hat s_\ell, \hat b_\ell, \hat t)$ for $\hat s_1,\dots, \hat s_\ell, \hat t \in \F^{d}$ and $\hat b_1,\dots,\hat b_\ell \in \F$.\;
	Expect $\vContext, \vQuery \in \F$ from the Prover, and run
	\begin{align*}
		 & \VerifyMarginoid_\F^{\pi}(\vContext;\ \hat s_1, \hat b_1, \dots, \hat s_\ell, \hat b_\ell)
		\quad\text{and}                                                                                         \\
		 & \VerifyMarginoid_\F^{\pi}(\vQuery;\ \hat s_1, \hat b_1, \dots, \hat s_\ell, \hat b_\ell, \hat t, 1).
	\end{align*}\;
	Compute the model values $\hat P(\hat q), \hat Q(\hat q), \hat Q(\hat q') \in \F$ directly, by evaluating the circuits $P, Q$ at $\hat q$ and $\hat q'$ over $\F$ (\Cref{fact:arith}).\;
	Accept if and only if $\VerifyEncoding$ and both $\VerifyMarginoid$ accepted, and
	\begin{equation}\label{eq:full:close}
		\hvInc \qeq \hat Q(\hat q) \cdot \left(2^B \vQuery - \hat P(\hat q) \cdot \vContext\right)^2
		\quad\text{and}\quad
		\hvNormQ \qeq \hat Q(\hat q')\quad \text{over}\ \F.
	\end{equation}
\end{algorithm}

\Cref{alg:full-protocol} is an $\IPCP$ in the sense of \Cref{def:ipcp}, with the PCP length counted in $\F$-symbols; \Cref{thm:explicit-iop} below bounds its four complexity measures.

\begin{theorem}[Explicit IPCP for $\MODELCONSIST$]\label{thm:explicit-iop}
	For every $\epssound, \epsgap > 0$, \Cref{alg:full-protocol} is an $\IPCP$ Verifier for $\MODELCONSIST$ with the following properties. Let $(P,Q)$ be a predictive model whose circuits have size $\leq S$, with arithmetizations $\hat P, \hat Q$ of degree $\leq \Delta$. Let $\tau = \tau^\dagger / 2^B$ be the consistency threshold for $\tau^\dagger \in \{0,\dots,2^{B}-1\}$. Let $\F$ be a prime field such that $|\F| \geq 2^{O(\ell d + B + \log(1/\epsgap))}\cdot \Delta/\epssound$.
	\begin{itemize}
		\item \emph{Completeness.} There exist a proof oracle and an honest Prover such that, if $(P,Q)$ is $(\tau - \epsgap)$-consistent, the Verifier accepts with probability $1$.
		\item \emph{Soundness.}  If $(P,Q)$ is not $\tau$-consistent, the Verifier accepts with probability at most $\epssound$, regardless of the proof oracle and the Prover.
		\item \emph{Verifier runtime.}  $\poly(\ell, d, B, \log(1/\epsgap), 1/\epssound) + O(|P| + |Q| + \Delta \ell d)$ field operations.
		\item \emph{PCP length and query complexity.}  The proof oracle consists of $|\F|^{O(\ell d + B)}$ symbols of $\F$, of which the Verifier reads $\poly(\ell, d, B, \log(1/\epsgap), \log(1/\epssound))/\epssound$.
		\item \emph{Communication.}  Prover messages altogether have total length $\poly(\ell, d, B,$ $\log(1/\epsgap), \log(1/\epssound))$ $+ O(\Delta \ell d)$.
	\end{itemize}
\end{theorem}
\begin{corollary}[Boolean IPCP]\label{cor:boolean-ipcp}
	\Cref{thm:explicit-iop} holds with the proof oracle over the Boolean alphabet, each $\F$-symbol written in binary. The bounds remain asymptotically the same, with a PCP length of
	\begin{equation*}
		2^{\poly(\ell,d,B)} \cdot \left( \frac{\Delta}{\epsgap \cdot \epssound} \right)^{O(\ell d + B)}
	\end{equation*}
\end{corollary}
\begin{proof}
	Evaluating each gate over $\F$ in topological order takes $O(S)$ field operations and returns $\hat C(\hat x)$, where $\hat C$ is the polynomial in $X_1,\dots,X_n$ that the output gate computes. The arithmetic gates agree with the Boolean connectives on $\bool$ ($\neg a = 1-a$, $a \wedge b = ab$, $a \vee b = a+b-ab$), so $\hat C$ extends $C$ off the cube; the hypothesis $M < |\F|$ guarantees the integer value is read back without reduction. For the degree, induct over the gates: an input and a constant compute polynomials of degree $1$ and $0$, an addition gate one of degree at most the maximum of its inputs', and a multiplication gate one of degree at most their sum---which is the syntactic degree, so $\deg \hat C \le \Delta$. Only multiplication gates raise the degree, and each at most doubles it, so along a path of length $D$ it at most doubles $D$ times: $\Delta \le 2^{D}$; and $D \le S$ gives $\Delta \le 2^S$.
\end{proof}

\begin{proof}[Proof of \Cref{thm:explicit-iop}]
	At a high level, the encoding test $\VerifyEncoding$ pins $\pi$ to a genuine codeword. The two $\VerifyMarginoid$ calls then return true marginoids of the encoded distribution, the three direct circuit evaluations (\Cref{fact:arith}) supply the true model values, and the $\SumCheck$s certify that the Prover's $\vInc$ is the actual inconsistency and that the total mass $\norm{Q}_1$ equals its known value $m-1$. At that point the integer threshold \eqref{eq:full:threshold} turns ``$\Incdag$ exceeds the bound'' into a deterministic rejection.

	Both directions of the proof rest on \Cref{eq:I-Dckt}, which we now prove. For any codeword $\widetilde\pi = (\widetilde Z,\widetilde A) \in \Code_\F(n,m,B')$ encoding a distribution $\mu$, the two marginoids of \Cref{eq:masses}, computed for $\widetilde\pi$, are, on a Boolean query $q=(s_1,b_1,\dots,s_\ell,b_\ell,t)$,
	\begin{align*}
		\Cmass & = \sum_{j} \weight^{A}(j)\match^Z(j,q) = 2^{B'}\mu[s = b],
		\\
		\Qmass & = \sum_{j} \weight^{A}(j)\match^Z(j,q)Z(j,t) = 2^{B'}\mu[(s = b) \cap (t = 1)].
	\end{align*}
	Indeed, $\weight^{A}(j) = \alpha^\dagger_j$ is the integer weight (\Cref{def:code}), $\match^Z(j,q) = \mathds{1}[\forall k\ z_{j,s_k} = b_k]$ indicates the context, and $Z(j,t) = z_{j,t}$. Substituting into \Cref{eq:Incdag} and pulling out $2^{B+B'}$,
	\begin{align*}
		2^B\Qmass - P(q)\Cmass        & = 2^{B+B'}\Bigl(\mu[(s = b) \cap (t = 1)] - \frac{P(q)}{2^B}\mu[s = b]\Bigr),
		\\
		\text{therefore}\quad \Incdag & = 2^{2(B+B')}\norm{Q}_1\D(\mu)^2,
	\end{align*}
	which is \Cref{eq:I-Dckt}. Therefore, the check in \Cref{eq:full:threshold} holds if and only if $\D(\mu) \le \tau$. %

	Before we turn to completeness and soundness, we argue that the field $\F$ is large enough to avoid wraparounds. By \Cref{claim:sparse-witness}, whenever $(P,Q)$ is $(\tau-\epsgap)$-consistent there is a witness $\mu$ with $\D(\mu) \le \tau$, supported on $z_1,\dots,z_{m}$ with integer weights summing to $2^{B'}$ where $B' \in O(\log m + \log(1/\epsgap))$. The honest proof oracle is exactly the codeword $\pi = \Code_\F(n,m,B')$ of this $\mu$ (\Cref{def:code}). (If a single point of $\mu$ carries weight $2^{B'}$, which \Cref{def:code} cannot store in $B'$ bits, the Prover lists that point twice in $z$ with the mass split evenly; the support is a multiset, so this encodes the same distribution with every stored weight in $\{0,\dots,2^{B'}-1\}$.) Write $M \coloneqq  2^{B + \ell(d+1)+d}$; as observed in \Cref{rem:variable-length}, $\norm{Q}_1 \leq M$ for any confidence circuit $Q$. By assumption,
	\begin{equation}\label{eq:full:fieldsize}
		|\F| > 2^{2(B+B')} M \ge 2^{2(B+B')} \norm{Q}_1.
	\end{equation}
	Since $\D(\mu) \le 1$ for every distribution $\mu$, the true inconsistency of any encoded $\mu$ satisfies
	\begin{equation*}
		\Incdag = 2^{2(B+B')}\norm{Q}_1\D(\mu)^2 \le 2^{2(B+B')} \norm{Q}_1 \le 2^{2(B+B')} M < |\F|,
	\end{equation*}
	and likewise $\norm{Q}_1 \le M < |\F|$. Therefore, whenever $\vInc$ equals $\Incdag$ over $\F$ and the $\norm{Q}_1$-$\SumCheck$ has pinned the mass to $m-1$, their integer lifts in \eqref{eq:full:threshold} are exactly $\Incdag$ and $\norm{Q}_1 = m-1$, no wraparounds. For an adversarial oracle the same holds with the \emph{declared} dimensions: the dimension check of \Cref{alg:full-protocol} enforces $2^{2(B+B')}(m-1) < |\F|$ and $m(2^{B'}-1) < |\F|$ for the $(m, B')$ read off $\pi$'s signature, so the lifted comparison in \eqref{eq:full:threshold} is wraparound-free for any oracle the Prover supplies, once the $\norm{Q}_1$-$\SumCheck$ has pinned the mass to $m-1$ (as in the case analysis below); the honest oracle passes the check by \Cref{eq:full:fieldsize}.

	Throughout the analysis below, each cited bound (\Cref{fact:sumcheck,claim:ipp,claim:marginal}) holds for every fixed claimed value and every Prover strategy, over the sub-protocol's own fresh coins; conditioning on the transcript prefix that fixes that value and averaging therefore preserves each bound, so the failure events may be union-bounded even though the Prover chooses $\vContext, \vQuery$ only after $\hat q$ is public.

	\pparagraph{Completeness.} Suppose $(P,Q)$ is $(\tau-\epsgap)$-consistent, and let $\mu$ be the aforementioned sparse distribution so $\D(\mu) \le \tau$. The proof oracle is the codeword $\pi$ (whose dimensions pass the check of \Cref{alg:full-protocol}, by $m \le M + 1$ and \Cref{eq:full:fieldsize}), and the honest Prover sends $\vInc = 2^{2(B+B')} \norm{Q}_1 \D(\mu)^2$ and the true marginoids $\vContext = \Cmasshat$ and $\vQuery = \Qmasshat$ (the total mass $\norm{Q}_1 = m-1$ is determined by $\pi$'s length, not sent), and plays the honest strategies of \Cref{claim:ipp,claim:marginal}. The threshold \eqref{eq:full:threshold} then holds over $\Z$: the lifts are exactly $\Incdag$ and $\norm{Q}_1$ (by \Cref{eq:full:fieldsize}, as recorded above), and $\Incdag = 2^{2B} 2^{2B'} \norm{Q}_1 \D(\mu)^2 \le (\tau^\dagger)^2 2^{2B'} \norm{Q}_1$ since $\D(\mu) \le \tau = \tau^\dagger/2^B$. $\VerifyEncoding$ accepts because $\pi$ is a genuine codeword; each $\VerifyMarginoid$ accepts with probability $1$ because $\vContext,\vQuery$ satisfy \Cref{eq:marginal:claim}; and the model values $\hat P(\hat q), \hat Q(\hat q), \hat Q(\hat q')$ are the Verifier's own evaluations of the circuits (\Cref{fact:arith}). Finally, both equalities of \eqref{eq:full:close} hold by construction: the $\Incdag$-$\SumCheck$ runs on the polynomial $q \mapsto \hat Q(q)\bigl(2^B\Qmass - \hat P(q)\Cmass\bigr)^2$, so its output is that polynomial at $\hat q$,
	\begin{equation*}
		\hvInc = \hat Q(\hat q)\bigl(2^B\Qmasshat - \hat P(\hat q)\Cmasshat\bigr)^2 = \hat Q(\hat q)\bigl(2^B\vQuery - \hat P(\hat q)\vContext\bigr)^2,
	\end{equation*}
	likewise $\hvNormQ = \hat Q(\hat q')$, since the $\norm{Q}_1$-$\SumCheck$ runs on $\hat Q$, which agrees with $Q$ on $\Qcal$ (\Cref{fact:arith}). Therefore, the Verifier accepts with probability $1$.

	\pparagraph{Soundness.} Now suppose $(P,Q)$ is not $\tau$-consistent. Fix any proof oracle $\pi$ and any Prover; we bound the acceptance probability by cases. Let $\delta = \nicefrac{\epssound}{8\ell+12}$.

	If $\pi$ is $\delta$-far from all $\Code$ codewords, then by \Cref{claim:ipp} (with soundness parameter $\epssound/4$) $\VerifyEncoding$ accepts with probability at most $\epssound/4$.

	Otherwise $\pi$ is $\delta$-close to a codeword $\widetilde\pi = (\widetilde Z,\widetilde A) \in \Code$ encoding a genuine distribution $\widetilde\mu$. The rest of the protocol reads the oracle at $2\ell+3$ points, each uniform over its oracle's domain: the $A$-reads by \Cref{claim:marginal}, and the $Z$-reads $(\hat j, \hat s_k), (\hat j, \hat t)$ because $\hat s_k, \hat t$ are coins of the $\Incdag$-$\SumCheck$, independent of the fresh $\hat j$. By a union bound, all of these reads agree with $\widetilde\pi$ except with probability $(2\ell+3)\delta \le \epssound/4$. Condition on all queries landing in the agreement region; every subprotocol now behaves as if run on $\widetilde\pi$. Now,

	\begin{itemize}
		\item If $\vContext \neq 2^{B'}\widetilde\mu[s = b]$, the corresponding $\VerifyMarginoid$ accepts with probability at most $\frac{(\ell+1)\log m + 2\log B'}{|\F|}$ by \Cref{claim:marginal}. Similarly for $\vQuery = 2^{B'}\widetilde\mu[(s = b) \cap (t = 1)]$, but with $\ell+2$ instead of $\ell+1$ in the numerator.
		\item If $\vInc \neq \Incdag$, the corresponding \SumCheck in \Cref{eq:full:close} passes with probability at most $(3\Delta+2)\log|\Qcal|/|\F|$. Here we use the assumption that $\hat P$ and $\hat Q$ have degree at most $\Delta$ and that the marginoids are multilinear in $q$.
		\item If $m \neq \norm{Q}_1 + 1$, the corresponding \SumCheck in \Cref{eq:full:close} passes with probability at most $\Delta\log|\Qcal|/|\F|$. Here we similarly use the degree bound on $\hat P$ and $\hat Q$.
	\end{itemize}

	If instead all the above equalities hold, then by the assumption on $|\F|$ and \Cref{eq:I-Dckt}, $\Incdag > (\tau^\dagger)^2 2^{2B'} \norm{Q}_1$ over $\Z$, so the threshold check fails and the Verifier rejects. %

	Collecting these bounds, the Verifier accepts with probability at most
	\begin{equation*}
		\Pr[\acc] \le \frac{\epssound}{4} + \frac{\epssound}{4} + \frac{(2\ell+3)\log m + 4\log B' + (4\Delta+2)\log|\Qcal|}{|\F|}.
	\end{equation*}
	The preceding bound gives $\Pr[\acc] \le \epssound$ as soon as $|\F|$ exceeds twice its numerator divided by $\epssound$. It remains to confirm that this, and the field hypotheses \eqref{eq:code:delta} of the encoding test (\Cref{claim:ipp}, at soundness parameter $\epssound/4$), hold for a field of the theorem's size---on top of the wraparound bound \eqref{eq:full:fieldsize} assumed above. The encoding test's no-wrap condition $m(2^{B'}-1) < |\F|$ is itself implied by \eqref{eq:full:fieldsize}, and every remaining condition is polynomial in the model parameters, $\Delta$, and $1/\epssound$. Each is monotone in $m$, so the Verifier certifies them at the a-priori cap $m \le M+1$ (\Cref{rem:variable-length}), never needing the true support size. As \eqref{eq:full:fieldsize} is $2^{O(\ell d + B)}\poly(1/\epsgap)$ and dominates, the theorem's field size $|\F| \ge 2^{O(\ell d + B)}\poly(1/\epsgap)\cdot\Delta/\epssound$ satisfies them all. %

	\pparagraph{The four complexity measures.} The \emph{PCP length} is the size of the oracle $\pi = (Z,A)$: it consists of $|\F|^{\log m+d} + |\F|^{\log m+\log B'}$ symbols (\Cref{def:code}), exponential in $\ell, d, B, \log(1/\epsgap)$, and $\log(1/\epssound)$. The other three measures are additive over the sub-protocols; the table below collects their contributions, writing $b \coloneqq \log m + d + \log B' = O(\ell d + B)$ and recalling $\delta = \epssound/(4(2\ell+3))$, $\log|\Qcal| = \ell(d+1)+d$, and $\log|\F| = O(\ell d + B + \log(1/\epsgap) + \log(1/\epssound))$ bits per element.
	\[
		\begin{array}{l ccc}
			                                            & \text{oracle queries}     & \text{communication } (\F\text{-elts}) & \text{Verifier ops}       \\
			\hline
			\VerifyEncoding                             & b\log(1/\epssound)/\delta & b\log(1/\epssound)                     & b\log(1/\epssound)/\delta \\
			\VerifyMarginoid\ (\times 2)                & \ell                      & \ell\log m + \log B'                   & \ell\log m + \log B'      \\
			\text{Direct }\SumCheck\text{s}\ (\times 2) & 0                         & \Delta\log|\Qcal|                      & \Delta\log|\Qcal|         \\
			\text{model evaluations}\ (\times 3)        & 0                         & 0                                      & |P|+|Q|                   \\
		\end{array}
	\]
	all entries $O(\cdot)$. Only $\VerifyEncoding$ and the two $\VerifyMarginoid$ read the oracle, $\VerifyEncoding$ dominating; the query total is $O\bigl(\ell(\ell d + B + \log(1/\epsgap))\log(1/\epssound)/\epssound\bigr)$. Summing the communication column and converting at $\log|\F|$ bits per element gives $\Delta\cdot\poly(\ell, d, B, \log(1/\epsgap), \log(1/\epssound))$, its only circuit-dependent term being the direct $\SumCheck$s. And at $\poly(\log|\F|)$ per field operation, the Verifier runs in $O(|P|+|Q|)\poly(\log|\F|) + \Delta\cdot\poly(\ell, d, B, \log(1/\epsgap), 1/\epssound)$, its only circuit-dependent term being the three circuit evaluations.
\end{proof}

\begin{remark}[On the $1/\epssound$ in the query complexity]\label{rem:query-eps}
	The $1/\epssound$ in the query complexity enters through the proximity parameter $\delta = \epssound/(4(2\ell+3))$ of the encoding test (\Cref{claim:ipp}); we wonder whether it can be brought down to $\poly\log(1/\epssound)$. For example, one could encode the witness not by its multilinear extension but by an extension of higher individual degree (though perhaps lower total degree), admitting a more query-efficient proximity test; the higher individual degree would propagate into the $\SumCheck$s and hence into the field size. We have not attempted this direction beyond this sketch.
\end{remark}

\begin{remark}[A single evaluation point]\label{rem:single-point}
	The two direct $\SumCheck$s of \Cref{alg:full-protocol} range over the same universe $\Qcal$, and their soundness does not require independent randomness: the union bound in the proof of \Cref{thm:explicit-iop} sums the two error events whatever their joint distribution. The Verifier may therefore run both $\SumCheck$s on a single stream of coins, making $\hat q' = \hat q$: the final checks \eqref{eq:full:close} then evaluate the model at one queroid (the value $\hat Q(\hat q)$ serves both equalities), and the three $\VerifyMLE$ calls of \Cref{thm:explicit-iop-pp} become two. The one requirement is that each round's two messages be committed before that round's coin is revealed. Running the $\SumCheck$s in lockstep (both round-$k$ polynomials, then the coin) satisfies it with a single Prover, and the protocol remains an $\IPCP$, so \Cref{sec:delegation} and the compilers of \Cref{rem:succinct-argument} still apply. Running them sequentially with one Prover does not: having seen the first transcript, the Prover knows every coin of the second $\SumCheck$ in advance, and can choose its final round's coefficients so that the residual claim lands exactly on $\hat Q(\hat q)$, passing on a false claim with probability $1 - O(1/|\F|)$. Sequential order is sound again when the $\norm{Q}_1$-$\SumCheck$ is played by a prover isolated from the first transcript: formally a second prover, and in the self-proving instantiation one further context reset (\Cref{rem:spmip}). %
\end{remark}

\begin{remark}[The gap and the exact-witness instantiation]\label{rem:exact-witness}
	Running \Cref{alg:full-protocol} with the exact rational witness of \Cref{lem:rationalopt} (no gap) preserves completeness and soundness, but the threshold $\tau^22^{2(B+B')} \norm{Q}_1$ has bit-length $2^{O(\ell d + B)}$, so the field required by the proof of \Cref{thm:explicit-iop} becomes single-exponentially large. The IPCP statement above does not apply to this instantiation. This does not contradict \Cref{cor:mip}: that corollary obtains its polynomial-time Verifier by Cook--Levin arithmetization of the $\NP$ verifier underlying \Cref{thm:inNP}, foregoing Carath\'eodory structure entirely. \Cref{thm:explicit-iop} should be read as an explicit but (slightly) weaker counterpart of \Cref{cor:mip}.

	Nor can the support-only design of \Cref{thm:inNP} be imported here to dispense with the gap, though in the explicit setting it costs no completeness slack at all. Its Verifier cannot even write down the linear system it would solve, the side of which is $m = \norm{Q}_1 + 1 = 2^{\Theta(\ell d)}$, and delegating the solve to the Prover does not help: that design relocates the weights' \emph{computation}, not their \emph{precision}. By \Cref{lem:rationalopt} the exact weights are ratios of determinants of $\Theta(m(B+\log m))$ bits, and whatever carries them, oracle or messages, forces $\log|\F|$ of that order. It would also owe a \emph{local} check of positivity for exponentially many weights, which $\Code_\F$ (\Cref{def:code}) gets for free from its unsigned representation; and delegating a generic computation is the move \Cref{cor:mip} already makes in full, which \Cref{rem:implementation} is our reason for not making.
\end{remark}

\subsection{Improved efficiency with delegation of computation}\label{sec:delegation}
The dependence on the degree $\Delta$ in \Cref{thm:explicit-iop}, exponential in the circuit depth in the worst case, stems from a single design choice: the Verifier reads each model value $\hat P(\hat q), \hat Q(\hat q), \hat Q(\hat q')$ by evaluating the circuit as the high-degree polynomial it computes. The standard machinery of \emph{delegation of computation} removes it. We replace each direct evaluation by a call to $\VerifyMLE$ (\Cref{fact:thaler}), now letting $P_\MLE, Q_\MLE$ denote the \emph{multilinear} extensions of $P, Q$: rather than evaluate the circuit itself, the Verifier receives an alleged value from the Prover and has $\VerifyMLE$ certify it through interaction. This is the same move, from a quantity the Verifier computes to one a sub-protocol certifies, as the passage from $\VerifyMarginoid$ to $\VerifyMarginal$ in \Cref{sec:codecheck}.

The fact underlying $\VerifyMLE$ is a doubly-efficient interactive proof for circuit evaluation: it delegates the evaluation to the Prover, at a cost scaling with the circuit's depth $D$ rather than its degree.
\begin{fact}[\cite{Thaler13}]\label{fact:thaler}
	Let $C$ be a layered arithmetic circuit over a field $\F$ with $n$ inputs, depth $D$, and size $S$, computing a function $\bool^n \to \{0,\dots,M-1\}$ for $M < |\F|$, and let $C_\MLE \colon \F^n \to \F$ be its multilinear extension.\footnote{We emphasize that $C_\MLE$ is the multilinear extension of the \emph{function} computed by $C$, not the polynomial $\hat C$ of \Cref{fact:arith}; the two agree only on $\bool^n$. The primitive below, with these exact costs, is obtained by running the protocol of \cite{Thaler13} on the wide circuit made of $2^n$ parallel copies of $C$, the $w$th hardwired to input $w$: the multilinear extension of its output layer is $C_\MLE$, and its width enters the costs only through its logarithm, $n + \log S$.} There is a Verifier $\VerifyMLE_\F(C, \hat x, y)$ such that, for any $\hat x \in \F^n$ and $\hat y \in \F$,
	\begin{itemize}
		\item \emph{Completeness:} there is an honest Prover for which, if $\hat y = C_\MLE(\hat x)$, the Verifier accepts with probability $1$.
		\item \emph{Soundness:} if $\hat y \ne C_\MLE(\hat x)$, the Verifier accepts with probability $O(D(n + \log S)/|\F|)$, regardless of the Prover.
		\item \emph{Complexity:} the communication is $O(D(n + \log S))$ field elements and the Verifier runs in $O(S + D(n + \log S))$ field operations. The $S$ term shrinks to $\polylog S$ if $C$ is logspace-uniform.
	\end{itemize}
\end{fact}

\begin{theorem}[Explicit IPCP for $\MODELCONSIST$, with model delegation]\label{thm:explicit-iop-pp}
	Let $(P,Q)$ be a predictive model whose circuits are layered, each with depth $\leq D$ and size $\leq S$. Run \Cref{alg:full-protocol} with the direct circuit evaluations replaced by the calls $\VerifyMLE_\F(P, \hat q, P_\MLE(\hat q))$, $\VerifyMLE_\F(Q, \hat q, Q_\MLE(\hat q))$, and $\VerifyMLE_\F(Q, \hat q', Q_\MLE(\hat q'))$ of \Cref{fact:thaler}, and the two direct $\SumCheck$s at degrees $5$ and $1$. The result is an $\IPCP$ Verifier for $\MODELCONSIST$ with completeness and soundness as in \Cref{thm:explicit-iop}, with PCP length and query complexity unchanged, and with the following bounds, polynomial in the depth $D$ in place of the degree $\Delta$, over any prime field with $|\F| \ge 2^{O(\ell d + B + \log(1/\epsgap))}\cdot\poly(\ell,d,B,D,\log S)/\epssound$:
	\begin{itemize}
		\item \emph{Verifier runtime.} $\poly(\ell, d, B, \log(1/\epsgap), 1/\epssound)$, plus the three $\VerifyMLE$ calls, $O(|P| + |Q| + D(\ell d + \log S))$ field operations.
		\item \emph{Communication.} $\poly(\ell, d, B, \log(1/\epsgap), \log(1/\epssound))$, plus the three $\VerifyMLE$ transcripts, $O(D(\ell d + \log S))$ field elements.
	\end{itemize}
\end{theorem}
\begin{proof}
	The substitution changes the protocol in three places; everywhere else the analysis of \Cref{thm:explicit-iop} applies verbatim.

	\emph{Model values.} The values $P_\MLE(\hat q), Q_\MLE(\hat q), Q_\MLE(\hat q')$ are now supplied by the Prover and certified by $\VerifyMLE$ rather than computed by the Verifier. By \Cref{fact:thaler}, with input length $n = \log|\Qcal| = \ell(d+1)+d$, each call is complete on the true value and, on a false value, accepts with probability $O(D(\log|\Qcal| + \log S)/|\F|)$. In the soundness union bound this replaces the zero error of direct evaluation with $O(D(\log|\Qcal| + \log S)/|\F|)$ from the three $\VerifyMLE$ calls, restoring the bound of the GKR-based protocol.

	\emph{$\SumCheck$ degrees.} With $P_\MLE, Q_\MLE$ multilinear the $\Incdag$-summand $Q_\MLE(q)(2^B\Qmass - P_\MLE(q)\Cmass)^2$ has individual $q$-degree $\deg Q_\MLE + 2(\deg P_\MLE + 1) = 5$ and the $\norm{Q}_1$-summand $Q_\MLE(q)$ degree $1$ (the marginoids being multilinear in $q$ as before). The two direct $\SumCheck$s thus contribute $6\log|\Qcal|/|\F|$ to the soundness bound and $O(\log|\Qcal|)$ field elements each to the communication, and the degree $\Delta$ disappears from every bound. The field size \eqref{eq:full:fieldsize} and the wraparound argument are unchanged (they never involved $\Delta$), so collecting the soundness terms gives the stated $|\F| \ge 2^{O(\ell d + B + \log(1/\epsgap))}\cdot\poly(\ell,d,B,D,\log S)/\epssound$, dominated as before by the wraparound term.

	\emph{Verifier work and communication.} The Verifier no longer evaluates $P, Q$; its circuit-dependent cost is the three $\VerifyMLE$ calls, each $O(S + D(\ell d + \log S))$ field operations and $O(D(\ell d + \log S))$ communication by \Cref{fact:thaler}, hence $O(|P| + |Q| + D(\ell d + \log S))$ and $O(D(\ell d + \log S))$ in total. The per-call $S$ term is the cost of evaluating the layered circuit once, as \Cref{fact:thaler} requires. A circuit that is not already layered is first layered: this preserves the depth $D$ and increases the size to at most $O(SD)$, so the $|P| + |Q|$ term grows to $(|P| + |Q|)D$, while $\log(SD) = O(\log S)$ (as $D \le S$) leaves the field size and communication unchanged.
\end{proof}
Since $D \le S$, every bound above is polynomial in the model size for circuits of any depth, removing the constant- or logarithmic-depth restriction under which \Cref{thm:explicit-iop} is polynomial. And by the argument of \Cref{cor:boolean-ipcp}, \Cref{thm:explicit-iop-pp} too holds with the proof oracle over the Boolean alphabet, with $\Delta$ replaced by $\poly(\ell,d,B,D,\log S)$ in the PCP length.

\begin{remark}[Verifier efficiency for uniform models]\label{rem:doubly-efficient}
	The only Verifier costs in \Cref{thm:explicit-iop-pp} that depend on the size of $(P,Q)$ are reading the input circuits and the three $\VerifyMLE$ calls. If the model is instead given \emph{succinctly}, as a logspace-uniform layered circuit family accessed through its description, then the uniform clause of \Cref{fact:thaler} brings each $\VerifyMLE$ call to $\poly(D, \ell d, \log S)$ field operations, and the Verifier runs in $\poly(\ell, d, B, D, \log S, \log(1/\epsgap), 1/\epssound)$ time overall---polylogarithmic in the circuit size. The honest Prover remains exponential-time (it writes the proof oracle), so this is Verifier-side efficiency only, in the spirit of doubly-efficient interactive proofs \cite{GoldwasserKR15}.
\end{remark}

\begin{remark}[Succinct arguments for $\MODELCONSIST$]\label{rem:succinct-argument}
	The delegated Verifier of \Cref{thm:explicit-iop-pp} runs in polynomial time, reads its oracle at polynomially many points, and exchanges polynomially many bits, so standard compilers turn the protocol into a \emph{succinct argument} for $\MODELCONSIST$ at no asymptotic cost. Two apply, each trading an assumption for a guarantee:
	\begin{itemize}
		\item \emph{Collision-resistant hashing} yields a public-coin interactive argument \citep{Kilian92}: the Prover commits the oracle by a single Merkle tree, and the interaction runs in the clear. The oracle is exponentially long, but the honest Prover is already exponential-time (it writes $\pi$), so the commitment costs nothing asymptotically, as in the universal arguments of \cite{BarakG02}.
		\item \emph{A random oracle} yields a non-interactive argument \citep{Micali94}: apply the Fiat--Shamir transform to the resulting public-coin interactive oracle proof \citep{Ben-SassonCS16}.
	\end{itemize}
	Two caveats attach. Negligible soundness is obtained by parallel repetition before the interactive compilation; repetition does not survive the Fiat--Shamir transform (a cheating prover grinds the hash round by round), so the non-interactive route instead requires round-by-round soundness, which our protocol is expected to acquire once its plain oracle reads are replaced by self-corrected ones (\Cref{fact:selfcorrect}); we do not prove this here.  The knowledge guarantee is only the weak, bit-by-bit one of \cite{BarakG02}, an exponentially long witness being too large to extract outright. Here ``succinct'' means $\poly(|x|,\lambda)$ communication and Verifier time (polynomial in the instance, not polylogarithmic), which, unlike the base \Cref{thm:explicit-iop}, holds for circuits of \emph{any} depth by \Cref{thm:explicit-iop-pp}. A user can thus certify a model's consistency against a $\poly(|x|,\lambda)$ certificate, at the price of computational rather than unconditional soundness.
\end{remark}

\section{Conclusions and open problems}\label{sec:conclusions}

We gave an explicit Interactive PCP for deciding whether a predictive model $(P,Q)$ is consistent up to an $\ell_2$ threshold, with a Verifier whose only access to the model is a few point evaluations of its circuits, together with a few queries to a proof oracle encoding a witnessing distribution (\Cref{thm:explicit-iop}).
In the process, we also demonstrated that $\EXPCONSIST$ problem, even tolerant to $(\ell_2)$ inconsistency, is in $\NP$. 
We close by discussing the high-level properties of our constructions, and the questions that these properties open.

\paragraph{A ``natural'' factorization.} Against an honest Prover, every object of \Cref{alg:full-protocol} carries a meaning: the oracle $(Z,A)$ is a Reed--Muller encoding of a (sparsified) witnessing distribution $\mu$, the \SumChecks express the inconsistency and the threshold $\norm{Q}_1$ as marginals of $\mu$, and $\VerifyEncoding$ pins the encoding's validity to fixed constants, the way a checksum does. The protocol thus factors into a model-specific, semantically probabilistic \emph{witness} core and a model-agnostic \emph{proving} overlay.

Following the pipeline of \Cref{sec:mip:witness}, the route from a trained model to one that proves its own consistency then partitions into stages of decreasing ``magic'': learning $(P,Q)$, finding a witnessing distribution, sparsifying it, and finally a purely classical proving overlay (the encoding and the \SumChecks) that could itself be replaced by a learned component, towards an end-to-end Self-Proving model in the spirit of \cite{AmitGPR25}. The two roles the IPCP requires, a stateless oracle and a stateful Prover, come ``for free'' from a single model that controls its context (\Cref{rem:spmip}). We leave a formal account of this factorization, and its consequences for learnability, to future work.

\paragraph{An exact Prover for an approximate statement.} The statement our protocol proves is, by design, only approximate: the Verifier certifies consistency up to the threshold $\tau$, promising acceptance only when the model is $(\tau-\epsgap)$-consistent. The honest Prover, however,  enjoys none of this slack. \Cref{alg:full-protocol} requires the Prover to compute the inconsistency with a witness \emph{exactly}, down to the smallest bit: $\Incdag$ (\Cref{eq:Incdag}) is a sum over all of $\Qcal$, an integer of $O(\ell d + B + \log(1/\epsgap))$ bits, and the Prover commits to it as $\vInc$ before the \SumChecks begin, each round of which then demands an exact partial sum of the same kind. A Prover that commits to anything else, say an estimate of $\Incdag$ from sampled queries, is caught by the final check \eqref{eq:full:close} exactly as a cheating Prover would be (\Cref{fact:sumcheck}); the tolerance enters only once, as the integer comparison \eqref{eq:full:threshold} of the exact committed value against the threshold. (The single-exponential field of \Cref{thm:explicit-iop} is the same exactness on the Verifier's side: $\F$ must carry $\Incdag$ at full precision, without wraparound.) So although we never ask whether the model is fully consistent, we do ask the Prover for its exact degree of inconsistency:
a Prover that estimates rather than sums, as a learned Prover would, appears to have no accepting strategy. 

One potential remedy is a \SumCheck in which approximation is native. The sum-check protocol for approximate computations of \cite{BitanDGIKT26} proves exactly such claims, $\sum_{x} g(x) \approx H$ with $g$ low-degree over the reals, its Verifier accepting Prover messages that are themselves only approximately consistent. Whether our protocol can be soundly instantiated over it, so that the honest Prover need only estimate its witness's inconsistency to within the gap, is a main question we leave open.

\paragraph{Square error and other error measures of violation.} 
Our inconsistency-tolerant framing of the $\EXPCONSIST$ problem measures violation by $\ell_2$ squared error, which has the benefits of being smooth (a nice property for gradient-based methods) and polynomial (which enables the traditional arithmetization techniques that we lean on in \cref{sec:succinct}). But how does this choice compare with other possible ways of measuring consistency?

Simply departing from $\ell_1$ residuals already takes the $\EXPCONSIST$ problem out of the classical linear programming framework, which is why we need to develop all of the machinery in \cref{sec:NP} to establish membership in $\NP$ despite conceptually-similar results for the fully-consistent case of $\tau=0$, where the violation measure does not matter (see related work on probabilistic consistency in \cref{sec:related}). 
This is a sturdier foundation that should more easily generalize to other error measures. 
Our exact gap-free result (\cref{thm:inNP}) is tailored to $\ell_2$, since we rely on the fact that stationary conditions form a linear system, but the approximate result (\cref{prop:near-linear-witnesses}) that we build on to solve $\MODELCONSIST$ is less sensitive this fact.  The \SumCheck protocol then accommodates the higher-degree summand of an $\ell_p$ measure directly. 

Relative entropy, or KL divergence \cite{kullback1951information} has uniquely nice properties as a divergence measure \cite{jiao2014curious,hobson1969theorem,leinster2019short}, earning it a dominant place in machine learning \cite{kullback1959information,joyce2025kullback}. The corresponding inconsistency measure \cite{RichardsonH21} has non-polynomial summand, however, and its exact numerical value is typically transcendental, placing its exact calculation (see our previous remark on the exact Prover) beyond the arithmetization route taken in this work. 
Calculating relative entropies approximately is straightforward though, so extending our work to this setting is gated by the requirement of an exact Prover. 

\paragraph{In search of more symmetric witnesses.}

Recall that the witnesses that underlie our solutions to $\EXPCONSIST$ and $\MODELCONSIST$ are joint distributions with sparse support---sparsity that is critical to the efficiency of our proof system. 
Sparse distributions are, however, unnatural in several respects, and it would be valuable to have proof systems whose witnesses need not be sparse. We now describe in more detail. 

In the case of $n$ variables and $m$ claims, a sparse witness $\mu \in \Delta (\bool^n)$ is supported on only $m+1$ of its $2^n$ possible outcomes;
as such, it has at most $\log_2(m+1)$ bits of entropy.
Yet among all distributions (maximally) consistent with given constraints, the one of maximum entropy is the standard choice of epistemic state \cite{jaynes1957information,friedman1971jaynes}---for, in Jaynes' interpretation, it assumes the least additional information about the outcome.
In our case, the support of the witness typically must be chosen arbitrarily to remain sparse, and it need not respect the symmetries of the instance, such as invariance under a permutation of the variables that maps a collection $\PP$ of probabilistic claims to itself.
The maximum-entropy witness is unique (as entropy is strictly concave on the convex set of witnesses), invariant under such permutations, and assumes no unnecessary dependencies between variables. 
For example, for the $m =n$ claims ``variable $i$ is a fair coin,'' the maximum-entropy witness is the uniform distribution, whereas $\tfrac12(\delta_z + \delta_{\bar z})$, with $\bar z$ the complement of $z$, is a witness for every $z \in \bool^n$ and is permutation-invariant only for $z \in \{0^n, 1^n\}$ (when $n \ge 3)$.
Furthermore, there exists no sparse witness that exhibits mutual independence between more than $\log_2(m+1)$ (non-constant) variables, since $k$ such variables have a joint marginal on all $2^k$ points of $\bool^k$. 
In the case of $\MODELCONSIST$, this phenomenon is subtle: the IPCP oracle is tailored to the context size $\ell$, and attains its manageable size by destroying independencies between collections of more than $\Omega(\ell)$ variables.

If we could replace the sparse witness with a succinctly represented \emph{non-sparse} distribution (\eg, a distribution with few nonzero Fourier coefficients rather than few support points, \cf\ \Cref{sec:related}), 
then the resulting witness (or corresponding IPCP) could be more closely related to the formation of the predictive model itself, and perhaps easier to obtain via the same training process.
Exploring the feasibility of non-sparse witnesses, and what is desired of them, is left for future work.

\paragraph{Zero-knowledge auditing.}
Having cast consistency as a problem about Interactive Proofs, we inherit that field's toolkit; zero knowledge in particular is one tool we expect to be able to import. We close with a speculative application, offered as an illustration of how this tool could be used. 
But what might one want to hide when verifying probabilistic consistency?
In both scenarios below, the secret is not just the witnessing distribution (which zero knowledge hides by default) but the probabilistic claims themselves.

In the setting of $\EXPCONSIST$, consider a scenario in which a model provider logs the queries its users pose and the probabilities with which they were answered, so that over any period it holds a Collection of Probabilistic Claims $\PP$, and an auditor wants assurance that the model has been answering consistently, \ie, that $\D(\PP) \le \tau$. The provider will not disclose $\PP$, which contains its users' data, nor a witness $\mu$, which reveals a great deal about the answers. Since $\EXPCONSIST \in \NP$ (\Cref{thm:inNP}), it has zero-knowledge proofs assuming one-way functions \cite{GoldreichMW91}, but applied directly these hide the wrong thing: zero knowledge conceals the witness, whereas the statement is common input, and here the statement $(\PP,\tau)$ contains sensitive data. Instead, one replaces $\PP$ with a cryptographic commitment to it. Namely, the provider commits to each query--answer tuple $(x,y,p)$ as it is served, say via a Merkle tree with root $c$. The auditor receives $c$, the alleged consistency tolerance $\tau$, and the dataset parameters $n$, $m$, $B$. The statement is that there exist a CPC $\PP$, openings of the leaves under $c$ to the tuples of $\PP$, and a certificate $\pi$ accepted by \Cref{alg:verifier} on input $(\PP,\tau,\pi)$. The same commitment $c$ binds the statement to the claims users were actually served: users who opt in keep the openings of their own leaves, so each can check that their query is in the audited tree, and a query the provider left out is detectable by its author. Finally, \cite{Kilian92,Micali94} may be used to bring the auditor's work down to poly-logarithmic in $m$, which matters at over two billion queries per day (see \cref{sec:intro}).

A similar application may be possible for $\MODELCONSIST$, with both the weights $(P,Q)$ and the proof oracle encoding the witness kept private. Our Verifier \Cref{alg:full-protocol} evaluates the circuits itself at two points, so concealing the weights calls for the delegated evaluations of \Cref{sec:delegation}, carried out against a commitment to the weights (the architecture may stay public), together with a zero-knowledge compilation of the rest of the protocol. The rest consists of \SumChecks and the tests of the Reed--$\mu$ller encoding, which would need to be made zero-knowledge e.g. via \cite {BCGV16,BCFGRS17,CFS17}. Binding the committed model to the answers users actually receive is the separate, and increasingly studied, problem of verifiable inference \cite{GhodsiGG17,ZKMLsurvey}.
	
\section*{AI Use Disclosure}
We used AI tools for research assistance, including brainstorming, identifying relevant references, suggesting potentially useful theorems or high-level organizational skeletons, and double-checking our work. All material in the paper was written and verified by the authors.

\section*{Acknowledgements}
We are grateful to Avishay Tal for important comments and contributions to earlier versions of this work. OP thanks Thomas Bourgeat, Alessandro Chiesa, L\'ena\"ic Chizat, Nicolas Flammarion, Tal Herman, Viktor Kun\v{c}ak, Nicholas Spooner, Justin Thaler, Thomas Vidick and Guy Weissenberg for helpful comments.

\bibliography{bib}
\bibliographystyle{alpha}

\end{document}